\documentclass[twoside,11pt]{article}

\usepackage{hyperref}

\usepackage{authblk}

\RequirePackage{algorithm}

\usepackage[utf8]{inputenc}
\usepackage[dvipsnames]{xcolor}

\usepackage{hyperref}
\hypersetup{
    colorlinks=true,
    citecolor=blue,
    filecolor=blue,
    linkcolor=blue,
    urlcolor=blue,
}
\usepackage{natbib}
\usepackage{enumitem}
\usepackage{microtype}
\usepackage{amsmath,amsthm,amssymb}
\DeclareMathOperator*{\argmin}{arg\,min}
\usepackage{dsfont}
\usepackage{mathrsfs}
\usepackage[noabbrev]{cleveref}
\usepackage{stmaryrd}

\usepackage{graphicx}
\usepackage{subcaption}
\usepackage{booktabs}
\usepackage{dblfloatfix} 
\usepackage{multirow}
\usepackage{wrapfig}
\usepackage{diagbox}

\makeatletter
\newcommand{\dashrule}[1][black]{%
  \color{#1}\rule[\dimexpr.5ex-.2pt]{4pt}{.4pt}\xleaders\hbox{\rule{4pt}{0pt}\rule[\dimexpr.5ex-.2pt]{4pt}{.4pt}}\hfill\kern0pt%
}
\makeatother

\usepackage{tikz}
\usetikzlibrary{calc}
\usetikzlibrary{positioning}

\theoremstyle{plain}
\newtheorem{theorem}{Theorem}[section]
\newtheorem{proposition}[theorem]{Proposition}
\newtheorem{lemma}[theorem]{Lemma}
\newtheorem{corollary}[theorem]{Corollary}

\theoremstyle{definition}
\newtheorem{definition}[theorem]{Definition}
\newtheorem{assumption}{Assumption}

\theoremstyle{remark}
\newtheorem{remark}[theorem]{Remark}

\newtheorem{model}[theorem]{Model}

\usepackage{algcompatible}

\usepackage{breqn}

\definecolor{blindorange}{RGB}{215,131,37}
\definecolor{blindblue} {RGB}{81,172,226}
\definecolor{blindpurple} {RGB}{193,114,177}
\definecolor{blindgreen} {RGB}{33,145,106}
\definecolor{blindpink} {RGB}{252,174,226}
\definecolor{blindbrown} {RGB}{203,145,100}
\definecolor{mLightBrown}{RGB}{230,120,20}

\newcommand{\ind}{\perp\!\!\!\!\perp} 

\newcommand{\yindmx}{$Y\!\!\ind\!\!M\; | X$\xspace}
\newcommand{\yindmxbold}{$\mathbf{Y\!\!\ind\!\!M\; | X}$\xspace}
\newcommand{\ynotindmx}{$Y\!\!\not\ind\!\!M\; | X$\xspace}

\usepackage[hang,flushmargin]{footmisc}

\usepackage{xspace}

\usepackage{wasysym}
\usepackage{pifont}

\def\mis{\mathrm{mis}}
\def\obs{\mathrm{obs}}

\def\Tr{\mathrm{Tr}}
\def\Cal{\mathrm{Cal}}

\def\ind{\perp\!\!\!\!\perp }

\def\p{\mathcal{P}}
\def\exch{^{\mathrm{exch}(n+1)}}
\def\iid{^{\otimes(n+1)}}
\def\mcar{_{\mathrm{MCAR}}}
\def\mar{_{\mathrm{MAR}}}
\def\mnar{_{\mathrm{MNAR}}}
\def\ymx{_{\mathrm{Y\!\!\ind\!\!M\; | X}}}
\def\mcarymx{_{\mathrm{MCAR,Y\!\!\ind\!\!M\; | X}}}

\DeclareMathAlphabet{\mathdutchcal}{U}{dutchcal}{m}{n}

\newcommand{\mask}{\texttt{CP-MDA-Nested}\xspace}
\newcommand{\masksub}{\texttt{CP-MDA-Exact}\xspace}
\newcommand{\newnested}{\texttt{CP-MDA-Nested}$^{\star}$\xspace}

\usepackage[most]{tcolorbox}
\newtcolorbox{mybox}[1][]{%
  enhanced,
  opacityback=0, 
  frame hidden,
  overlay unbroken and first ={%
    \draw[thick] ([xshift=50pt]frame.north west) -| ([yshift=-50pt]frame.north west);
    \draw[thick] ([xshift=-50pt]frame.south east) -| ([yshift=50pt]frame.south east);
  }
}

\newcommand{\cna}{{\mathcal{C}}_{\alpha}}

\newcommand{\cmarkuq}{\textcolor{blindblue}{\ding{51}}}
\newcommand{\xmarkuq}{\textcolor{mLightBrown}{\ding{55}}}

\def\missing{\mathrm{missing}}
\def\observed{\mathrm{observed}}

\usepackage{cancel}

\definecolor{mydarkblue}{rgb}{0,0.08,0.45}
\hypersetup{
    colorlinks=true,
    linkcolor=mydarkblue,
    citecolor=mydarkblue,
    filecolor=mydarkblue,
    urlcolor=mydarkblue,
}

\usepackage{geometry}
\providecommand{\keywords}[1]{\textbf{Keywords:} #1}

\usepackage{times}

\begin{document}

\title{Predictive Uncertainty Quantification with Missing Covariates}

\author[,1,2]{Margaux Zaffran\thanks{Corresponding author: \texttt{margaux.zaffran@inria.fr}}}
\author[1]{Julie Josse}
\author[3]{Yaniv Romano}
\author[2]{Aymeric Dieuleveut}

\date{}

\affil[1]{PreMeDICaL project team, INRIA Sophia-Antipolis, Montpellier, France}
\affil[2]{CMAP, CNRS, École polytechnique, Institut Polytechnique de Paris, Palaiseau, France}
\affil[3]{Departments of Electrical Engineering and of Computer Science, Technion - Israel Institute of Technology, Haifa, Israel}

\maketitle

\begin{abstract} 

Predictive uncertainty quantification is crucial in decision-making problems. We investigate how to adequately quantify predictive uncertainty with missing covariates. A bottleneck is that missing values induce heteroskedasticity on the response's predictive distribution given the observed covariates. Thus, we focus on building predictive sets for the response that are valid \emph{conditionally} to the missing values pattern. We show that this goal is impossible to achieve informatively in a distribution-free fashion, and we propose useful restrictions on the distribution class. Motivated by these hardness results, we characterize how missing values and predictive uncertainty intertwine. Particularly, we rigorously formalize the idea that the more missing values, the higher the predictive uncertainty. Then, we introduce a generalized framework, coined \newnested, outputting predictive sets in both regression and classification. Under independence between the missing value pattern and both the features and the response (an assumption justified by our hardness results), these predictive sets are valid conditionally to any pattern of missing values. Moreover, it provides great flexibility in the trade-off between \emph{statistical variability} and \emph{efficiency}. Finally, we experimentally assess the performances of \newnested beyond its scope of theoretical validity, demonstrating promising outcomes in more challenging configurations than independence.
 
\end{abstract}

\keywords{predictive uncertainty quantification, missing values, conformal prediction, distribution-free inference} 

\section{Introduction}

\paragraph{Predictive uncertainty quantification.} Over the last decades, major research efforts on statistical and machine learning algorithms have enabled them to leverage large data sets. They are now used to support high-stakes decision-making problems such as medical, energy, or civic applications, to name just a few. To ensure the safe deployment of these models and their adoption by society, it is crucial to acknowledge that these point predictions remain uncertain, and to quantify this uncertainty, communicating the limits of predictive performance. Therefore, uncertainty quantification has received much attention in recent years, particularly in the form of building prediction sets. 

Formally, the aim is to build a predictive set for the response $Y \in \mathcal{Y}$, after observing the random vector $X \in \mathcal{X} \subseteq \mathds{R}^d$ which contains $d \in \mathds{N}^*$ explanatory variables. Given a \textit{miscoverage level} $\alpha \in [0,1]$, a \textit{marginally valid} predictive set $\mathcal{C}_\alpha (\cdot)$ is a function satisfying

\begin{equation}
\mathds{P}( Y \in \mathcal{C}_\alpha (X) ) \geq 1 - \alpha.
\label{eq:pi_mv}
\end{equation}

The goal is that $\mathcal{C}_\alpha(\cdot)$ is as small as possible while being marginally valid. Distribution-free uncertainty quantification tools are powerful as they require minimal assumptions on the data generation process---typically only access to a sequence of $n$ exchangeable data points---making them usable on a wide range of applications, unlike traditional probabilistic approaches.

Importantly, it has to be noted that \Cref{eq:pi_mv} averages among all probable $(X,Y)$, and thus might over-cover easy data points (say, e.g., young patients) at the cost of under-covering harder data points (say, e.g., older patients). Therefore, one branch of the literature studied how \Cref{eq:pi_mv} could be turned into a stronger goal. Specifically, \citet{vovk_conditional_2012,lei_distribution-free_2014,barber_limits_2021} emphasize trade-offs between theory and practice. They investigate the implications of designing a practical distribution-free method, that is one which outputs sets $\mathcal{C}_\alpha (\cdot)$ such that
\begin{equation}
\mathds{P}( Y \in \mathcal{C}_\alpha (x) | X = x) \geq 1 - \alpha, \text{ for any } x \in \mathcal{X}.
\label{eq:pi_cv}
\end{equation}

Unfortunately, they showed that \Cref{eq:pi_cv} is impossible to achieve in an informative way (i.e., typically $ \mathcal{C}_\alpha (\cdot) \equiv \mathcal{Y}$ with high probability) if no assumptions on the data distributions are made. Moreover, finding natural relaxations that are compatible with informative distribution-free predictive sets seems also hard: restrictions to conditioning on $x\in\mathdutchcal{X}$, for an arbitrary mass positive $\mathdutchcal{X}\subseteq\mathcal{X}$,  is still hard to achieve informatively \citep{barber_limits_2021}.

\paragraph{Missing values.} Somewhat paradoxically, as the quantity of data rises, the number of missing data also increases. This phenomenon is modeled through the introduction of a third random variable called the \textit{mask} or \textit{missing pattern}, denoted by $M \in \mathcal{M} \subseteq \{0,1\}^d$, encoding which variables have not been observed. That is, the mask $M$ is the indicator vector such that for any $j \in \llbracket 1,d \rrbracket$, $M_j = 1$ whenever $X_j$ is missing (not observed), and $M_j = 0$ otherwise. As a consequence, we are working on $\p := $\{distributions on $(\mathcal{X},\mathcal{M},\mathcal{Y})$\}. For a given pattern $m \in \mathcal{M}, X_{\text{obs}(m)}$ is the random vector of observed features, and $X_{\text{mis}(m)}$ is the random vector of unobserved ones. For example, if we observe $(\texttt{NA},6,2)$ then $m = (1,0,0)$ and $x_{\text{obs}(m)} = (6,2)$. Notice that the number of different missing patterns, i.e., the size or cardinality of $\mathcal{M} := \#\mathcal{M}$, typically grows exponentially in the dimension (for $\mathcal{M} = \{0,1\}^d$ there are $2^d$ different patterns).

The way we deal with those missing values will typically depend on the downstream task at hand. While there is a vast range of studies in the inferential setting \citep{Little2019,Josse2018StatScience} with numerous implementations \citep{mayerrmisstastic}, the research effort is scarcer on the prediction framework \citep{josse2019, lemorvan2020, lemorvan2020neumiss, lemorvan2021, ayme2022, vanness, pmlr-v202-ayme23a, pmlr-v202-zaffran23a, ayme2024random}. It is mostly limited to \textit{point prediction}, except for \citet{pmlr-v202-zaffran23a}. The literature on both inference and prediction highlights the necessity of taking into account the missingness distribution. Following \citet{rubin1976inference}, we consider three well-known missingness mechanisms.

\begin{definition}[Missing Completely At Random (MCAR)] 
The missing pattern distribution is said to be Missing Completely At Random (MCAR) if $M\ind X$.  We denote $\p\mcar$ the corresponding set of distributions, i.e.~$\p\mcar := \{ P \in \p,$ such that for any $m \in \mathcal{M}, \mathds{P}_P\left(M = m | X\right) = \mathds{P}_P\left(M = m\right)$, that is $M\ind X\}$.
\label[definition]{def:MCAR}
\end{definition}

\begin{definition}[Missing At Random (MAR)]
\label[definition]{def:MAR}
The missing pattern distribution is said to be Missing At Random (MAR) if $M$ only depends on the observed components of $X$. We denote $\p\mar$ the corresponding set of distributions, i.e.~$\p\mar := \{ P \in \p,$ such that for any $m \in \mathcal{M}$, $\left.\mathds{P}_P\left(M = m | X\right) = \mathds{P}_P\left(M = m | X_{\obs(m)}\right)\right\}$.
\end{definition}

\begin{definition}[Missing Non At Random (MNAR)]
\label[definition]{def:MNAR}
The missing pattern distribution is said to be Missing Non At Random (MNAR) if $M$ can depend on the observed values of $X$ but also on its missing components. We denote $\p\mnar$ the corresponding set of distributions, i.e.~$\p\mnar := \p$. 
\end{definition}

\begin{remark}
We thus have $\p\mcar \subset \p\mar \subset \p\mnar = \p$.
\end{remark}

\paragraph{Predictive framework with missing covariates.} In a predictive framework, the dependence between $Y$ and $M$ plays a key role, maybe even bigger than the relationship between $(X,M)$. Indeed, in some situations, $Y$ can be a direct function of $M$: the missingness conveys in itself information about the label. Therefore, these cases are particularly challenging in a predictive framework, as some patterns on the one hand can induce an important label distributional shift, and on the other hand be rarely observed due to the high cardinality of $\mathcal{M}$. Thus, we focus on settings where there is \textit{not} such a direct dependency, that is \Cref{ass:y_ind_m}. Yet, as we will show in the paper, it remains that the lack of observation of some features influences the uncertainty of the prediction of~$Y$ from $X_{\obs(M)}$.

\begin{assumption}[$M$ does not explain $Y$]
\label[assumption]{ass:y_ind_m} We say that $Y$ is independent of $M$ given $X$ if \yindmx. The associated distribution belongs to~$\p\ymx$.
\end{assumption}

\Cref{def:MAR,def:MCAR,def:MNAR} and \Cref{ass:y_ind_m} will be our main assumptions on the  joint distribution of $(X,M,Y)$  throughout the manuscript. Our interest is in building predictive sets from $n$ observations $\left( X^{(k)}, M^{(k)}, Y^{(k)} \right)_{k=1}^{n}$ on a new test point $\left( X^{(n+1)}, M^{(n+1)}, Y^{(n+1)} \right)$. We thus also make assumptions on the \textit{links between those samples}: the usual backbone assumption is that we have access to $n+1$ independent and identically distributed (i.i.d.)~draws from a distribution $Q$ in a set $\mathcal Q$, with  $\mathcal Q$ being typically one of $\p_\mcar$, $\p_\mar$, $\p$, etc. The data distribution thus belongs to $\left\{Q\iid, Q \in \mathcal Q \right\}$, which we denote $\mathcal Q \iid$. Furthermore, we also consider here a relaxation of i.i.d., namely \emph{exchangeability}, which is often sufficient to obtain guarantees in distribution-free predictive inference.

\begin{assumption}[exchangeability]
\label[assumption]{ass:iid}
The random variables $\left( X^{(k)}, M^{(k)}, Y^{(k)} \right)_{k=1}^{n+1}$ are exchangeable, i.e., their distribution does not change when we permute them. We denote $\mathcal Q \exch = \left\{Q\exch, Q \in \mathcal Q \right\}$ the set of distributions of exchangeable random variables, with marginal distribution in $\mathcal Q$.
\end{assumption}

An i.i.d. sequence is a fortiori exchangeable, while the reverse is not true (for example, sampling without replacement leads to a sequence that is exchangeable but not i.i.d.). 

\begin{remark}
We thus have that for any $\mathcal{Q}$, $\mathcal{Q}\iid\subset \mathcal{Q}\exch$.
\end{remark}

\paragraph{Predictive uncertainty quantification under missing covariates.} When features are missing, \Cref{eq:pi_mv} extends with $\mathcal{C}_\alpha$ a function of two arguments: $X$ and $M$. Specifically, $\mathcal{C}_\alpha$ is a \textit{marginally valid} predictive set for the test response $Y$ given its corresponding covariates $X$ and the mask $M$ if:

\begin{equation}
\mathds{P}( Y \in \mathcal{C}_\alpha (X,M) ) \geq 1 - \alpha.
\label{eq:na_mv}
\tag{MV}
\end{equation}

However, marginal validity \eqref{eq:na_mv} is not enough from an equity stand point and might result in discriminating between observations depending on their missing pattern \citep{pmlr-v202-zaffran23a}. Indeed, missing values create heteroskedasticity in the resulting distribution of $Y$ given $X_{\obs(M)}$. Therefore, they argue that when facing missing values one should aim at \textit{mask-conditional-validity} (MCV) even in the MCAR setting, i.e.:

\begin{equation}
\mathds{P}( Y \in \mathcal{C}_\alpha (X,M) | M ) \geq 1 - \alpha.
\label{eq:na_mcv}
\tag{MCV}
\end{equation}

\Cref{eq:na_mcv} is similar in spirit and motivation than \Cref{eq:pi_cv} but on a discrete space. Hence the impossibility results on $X$-conditional coverage do not hold anymore. However, \eqref{eq:na_mcv} is a challenging goal as it requires the coverage to be controlled on \emph{any} mask $m \in \mathcal{M}$, even those rarely observed at train time. 

In the sequel, to highlight the underlying dependencies and randomness, any estimator of $\mathcal{C}_{\alpha}(\cdot, \cdot)$ fitted on a data set $\left( X^{(k)}, M^{(k)}, Y^{(k)} \right)_{k=1}^{n}$ is denoted as $\widehat{C}_{n,\alpha} (\cdot, \cdot)$. We call a \emph{method} a function that, for any $\alpha\in[0,1]$, takes as input $\left( X^{(k)}, M^{(k)}, Y^{(k)} \right)_{k=1}^{n}$ and outputs an estimator $\widehat{C}_{n,\alpha} (\cdot, \cdot)$. \Cref{tab:notations} reminds the notations.

\begin{table}[!b]
    \centering
\begin{center}
\resizebox{\linewidth}{!}{
\begin{tabular}{lll}
\toprule
Name & Definition & Comment\\
\midrule
\#$A$ & Cardinal of the set $A$ & \\
$\mathdutchcal{P}(A)$ & Power set of $A$ & \\
\multicolumn{3}{@{}c@{}}{{\dashrule}} \\
$d$ & Number of features & \\
$\mathcal{X}$ & Features space & $\mathcal{X} \subseteq \mathds{R}^d$ \\
$\mathcal{Y}$ & Label space & \\
\multicolumn{3}{@{}c@{}}{{\dashrule}} \\
$\mathcal{M}$ & Missing values pattern space & $\mathcal{M} \subseteq \{0,1\}^d$ \\
$\texttt{NA}$ & Not Available (or missing value) & \\
$\obs(m)$ & Indices of the observed components for mask $m\in\mathcal{M}$ & $\obs(m)\in\mathds{N}^{|\obs(m)|}$\\
& (there are $|\obs(m)| := \sum_{i=1}^d m_i$ of them) & \\
$\mis(m)$ & Indices of the missing components for mask $m\in\mathcal{M}$ & $\mis(m)\in\mathds{N}^{|\mis(m)|}$\\
& (there are $|\mis(m)| := d - \sum_{i=1}^d m_i$ of them) & \\
$\mathcal{P}$ & Set of distributions on $\left( \mathcal{X},\mathcal{M},\mathcal{Y} \right)$ & \\
$\mathcal{P}\mar$ & Set of distributions on $\left( \mathcal{X},\mathcal{M},\mathcal{Y} \right)$ such that $X$ is Missing At Random & \\
$\mathcal{P}\mcar$ & Set of distributions on $\left( \mathcal{X},\mathcal{M},\mathcal{Y} \right)$ such that $X$ is Missing Completely At Random & \\
$\mathcal{P}\ymx$ & Set of distributions on $\left( \mathcal{X},\mathcal{M},\mathcal{Y} \right)$ such that \yindmx & \\
\multicolumn{3}{@{}c@{}}{{\dashrule}} \\
$n$ & Number of training observations & $n+1$ is the test index \\
$P\iid$ & Product distribution of $P$ with itself $n+1$ times & $P \in \p$ \\
& (i.e., distribution of $\left( X^{(k)}, M^{(k)}, Y^{(k)} \right)_{k=1}^{n+1}$ drawn i.i.d. with marginal $P$) & \\
$\mathcal{Q}\iid$ & $\left\{Q\iid, Q \in \mathcal Q \right\}$ & $\mathcal{Q} \subseteq \p$ \\
$P\exch$ & Exchangeable distribution of $n+1$ random variables of distribution $P$ & $P \in \p$ \\
$\mathcal{Q}\exch$ & $\left\{Q\exch, Q \in \mathcal Q \right\}$ & $\mathcal{Q} \subseteq \p$ \\
\multicolumn{3}{@{}c@{}}{{\dashrule}} \\
$\alpha$ & Miscoverage rate & $\alpha \in [0,1]$ \\
$\mathcal{C}_\alpha\left(\cdot,\cdot\right)$ & Predictive set function aiming at $1-\alpha$ coverage & $\mathcal{C}_\alpha : \mathcal{X} \times \mathcal{M} \longrightarrow \mathdutchcal{P}\left(\mathcal{Y}\right)$  \\
$\widehat{C}_{n,\alpha}\left(\cdot,\cdot\right)$ & Estimator for $\mathcal{C}_\alpha\left(\cdot,\cdot\right)$ based on $\left( X^{(k)}, M^{(k)}, Y^{(k)} \right)_{k=1}^n$, through a \emph{method} & \\
MV & Marginal validity & \\
MCV & Mask-conditional-validity & \\
\bottomrule
\end{tabular}
}
\end{center}
    \caption{Summary of notations.}
    \label{tab:notations}
\end{table}

\subsection{Literature's background}

Very recent papers have investigated uncertainty quantification with missing values. Both \citet{gui2023conformalized} and \citet{shao2023distributionfree} consider the question of distribution-free uncertainty quantification for matrix completion tasks. While the former considers building predictive sets for all of the missing entries, the latter focuses on what they call \textit{matrix prediction} where predictive sets are required only for the last ``individual'' of the data set.
\citet{pmlr-v206-seedat23a} addresses the related but distinct problem of missing values in the responses, which is generally known as the semi-supervised setting. They introduce a self-supervised learning approach for incorporating unlabeled training data into the conformalization process. In the same framework, \citet{lee2024simultaneous} leverages tools from the causal inference literature to achieve stronger guarantees such as feature and outcome's missingness conditional coverage, which are, in spirit, close to our focus (yet in a different framework). 

Closer to our work of predictive uncertainty quantification under missing covariates is \citet{pmlr-v202-zaffran23a}, as they focus on the same setting (i.e., to predict $Y$ given $X$, where $X$ might suffer from missing values both at train time and test time). After showing that \textit{impute-then-predict+conformalization} is marginally valid \eqref{eq:na_mv} for any missing mechanism and imputation, they introduce the harder goal of \textit{mask-conditional-validity} \eqref{eq:na_mcv}, motivated by an illustration on the heteroskedasticity generated by the missing values on a Gaussian Linear Model. They present a novel methodology, \textit{Missing Data Augmentation} (MDA), which combines with conformal prediction \citep[CP,][]{vovk_algorithmic_2005} in order to produce MCV sets. CP-MDA includes two algorithms, \texttt{CP-MDA-Exact} and \texttt{CP-MDA-Nested}, the former requiring a strict subsampling step on the training set, while the latter allows to keep the whole training data, which in turns induce large predictive sets. \citet{pmlr-v202-zaffran23a} provide theoretical guarantees on the MCV of \texttt{CP-MDA-Exact} and on a technical minor modification of \texttt{CP-MDA-Nested}, under MCAR and \yindmx assumptions. 

\subsection{Overview of our contributions (and outline)} 

In short, our objective is to tackle the following question: \textbf{when and how is it possible to achieve MCV?} Notably, we are interested in understanding $i)$ what  assumptions are  necessary to ensure MCV, $ii)$ how to design a tailored methodology,  and $iii)$ what happens when these assumptions are not satisfied.

We start by proving hardness results on distribution-free MCV in \Cref{sec:hardness}. Notably, for a MCV method outputting $\widehat{C}_{n,\alpha} (\cdot, \cdot)$ with no assumption except from having access to $n$ i.i.d. draws, we prove that the predictive interval is most often uninformative: for any $m\in\mathcal{M}$ the probability that, say, $\widehat{C}_{n,\alpha} (\cdot, m) \equiv \mathcal{Y}$ is higher than $1-\alpha-\Delta_{m,n}$, where $\Delta_{m,n}$ gets negligible when the mask $m$ is nearly not observed in a sample of size $n$. In other words, a method that is distribution-free MCV will output uninformative intervals on any mask that does not represent a high enough proportion of the training data. We go further and show that the exact same trade-off still holds for a method that is MCV only on distributions that are MAR, or MCAR, or similarly on distributions such that \yindmx, i.e., restricting an algorithm to be MCV only when \yindmx would still output uninformative sets on rarely observed masks: it is necessary to add another assumption on the dependence between $X$ and $M$ (such as MCAR) to allow for informative MCV on any mask. Importantly, this theoretical analysis brings new insights on the achievability of $X$-group-conditional validity, beyond MCV\footnote{Precisely, we provide a rigorous quantification of Vladmir Vovk's comment on $X$-conditional validity: ``of course, the condition that $x$ be a non-atom is essential: if $P_X({x}) > 0$, an inductive conformal predictor that ignores all examples with objects different from $x$ will have $1-\alpha$ object conditional validity and can give narrow predictions if the training set is big enough to contain many examples with $x$ as their object'' rephrased from \citet{vovk_conditional_2012} to match our notations.}.

This motivates the investigation of the quantile regression and missing values interplay in \Cref{sec:glm}, so as to provide guidelines for practical design of probabilistic prediction with missing covariates. This interplay is hard to characterize in general but becomes explicit under missingness~assumptions', or a multivariate Gaussian setting or linear model. Our key findings are ($i$---\Cref{sec:glm_var}) that the uncertainty often increases with more missing values: we analyze different mathematical statements of this main idea (in terms of conditional variance, inter-quantile distance, or predictive interval length) and evaluate theoretically under which distributional assumptions they are satisfied, in particular under MCAR and \yindmx, motivating our methodological design of \Cref{sec:nested}; ($ii$---\Cref{sec:glm_mask}) if the goal is to estimate quantiles, it is essential to be able to retrieve the mask to construct intervals, in contrast to classic mean regression where the mask is not as crucial.

In \Cref{sec:nested}, we propose a unified framework, \newnested, building predictive sets with missing covariates for both regression and classification tasks. Precisely, it bridges the gap between \masksub and \mask introduced in \citet{pmlr-v202-zaffran23a}, by encapsulating these two algorithms as well as any in between with more flexible subsampling schemes, allowing to fix the trade-off between coverage variance (\masksub) and overly conservative predictive sets (\mask). Leveraging the similarity between \newnested and leave-one-out conformal approaches \citep{Vovk2013,barber2021jackknife,gupta} we provide theory on the marginal validity of \newnested without subsampling, which holds regardless of the missingness distribution (without any assumption on the dependence between $M$ and $X$, but also without any assumption on the relationship between $M$ and $Y$ conditionally on $X$). Moreover, we also establish that \newnested is MCV for a wide range of subsampling schemes under MCAR and \yindmx.

Finally, in \Cref{sec:exp} we conduct synthetic experiments beyond the MCAR and \yindmx assumptions. Precisely, we generate missingness that is either MAR (5 different settings), MNAR (11 different settings) or such that \ynotindmx. \newnested empirically maintains MCV under MAR and MNAR missingness. When \yindmx is not satisfied, \newnested does not ensure MCV experimentally, unless the imputation is accurate enough. Overall, these numerical experiments showcase the robustness of \newnested beyond its theoretical scope of validity.

In the following \Cref{tab:sum_res}, we summarize and organize our main contributions. We report the theoretical results on the possibility to achieve informative MCV, either positive results (\cmarkuq) or negative hardness results (\xmarkuq), along with our more general result on marginal validity. Moreover, we locate experimental results by indicating the figures that align with particular setups. In particular, we distinguish two kinds of experiments: \textit{Numerical extension} of results beyond the conditions where the theory is applicable, which demonstrates promising outcomes in more challenging configurations, and \textit{Numerical confirmation} of results anticipated by theoretical analysis, that is the outcomes of the experiment either $i)$ were already expected based on the theory or $ii)$ confirm that the theoretical assumptions can not be relaxed to the corresponding distributional setting.

\begin{table}[!h]
    \begin{center}
   \resizebox{\linewidth}{!}{ \begin{tabular}{c|ccc|l}
    \toprule
    & $\p\mcar$ & $\p\mar$ & $\p\mnar = \p$ & \\ 
    \midrule
    \multirow{6}{*}{$\p\ymx$} & \newnested: \cmarkuq & \multirow{2}{*}{\textcolor{blindpurple}{?}} & Hardness: \xmarkuq & \multirow{2}{*}{\emph{Theory}} \\
    & \emph{\Cref{thm:mcv-nested}} & & \emph{\Cref{prop:tradeoff-mcv-yindm}} & \\
    & \multicolumn{4}{@{}c@{}}{{\dashrule}} \\
    & & \Cref{fig:mar_no_dep,fig:mar_dep} & \Cref{fig:mnar_sm_no_dep,fig:mnar_sm_dep,fig:mnar_q_no_dep,fig:mnar_q_dep} & \emph{Num. extension} \\
    & \multicolumn{4}{@{}c@{}}{{\dashrule}} \\
    & \Cref{fig:mcar} & & \Cref{rem:nested_not_broader_mcv} & \emph{Num. confirmation} \\
    \midrule
    \multirow{7}{*}{$\p$}  & Hardness: \xmarkuq & Hardness: \xmarkuq & Hardness: \xmarkuq & \multirow{4}{*}{\emph{Theory}} \\
    & \emph{\Cref{prop:tradeoff-mcv-mar-mcar}} & \emph{\Cref{prop:tradeoff-mcv-mar-mcar}} & \emph{\Cref{prop:tradeoff-mcv}} & \\
    & & & \textcolor{gray}{\newnested: MV} & \\
    & & & \textcolor{gray}{\emph{\Cref{thm:mv-nested}}} & \\
    & \multicolumn{4}{@{}c@{}}{{\dashrule}} \\
    & \Cref{fig:ydepmx_dep} & & & \emph{Num. extension} \\
    & \multicolumn{4}{@{}c@{}}{{\dashrule}} \\
    & \Cref{fig:ydepmx_no_dep} & \Cref{rem:nested_not_broader_mcv} & \Cref{rem:nested_not_broader_mcv} & \emph{Num. confirmation} \\
    \bottomrule
    \end{tabular}}
    \caption{Summary of the main theoretical results.}\label{tab:sum_res}
    \vspace{-1cm}
    \end{center}
\end{table}

\section{When is Mask-Conditional-Validity (MCV) a too lofty goal?}
\label{sec:hardness}

We will show in this section that purely distribution-free MCV guarantees are often uninformative. As a consequence, we will have to impose some non-parametric assumption on the underlying data distribution. We thus have to define the concept of MCV with respect to a class of distributions $\mathcal{D}$ (MCV-$\mathcal{D}$), and to study the sets $\mathcal{D}$ that allow for informative MCV-$\mathcal{D}$.

\begin{definition}[MCV-$\mathcal{D}$]
\label[definition]{def:mcv}
Let $\mathcal{D}$ be a set of distributions on $\left(\mathcal{X}\times\mathcal{M}\times\mathcal{Y}\right)^{n+1}$. A method outputting $\widehat{C}_{n,\alpha} (\cdot, \cdot)$ based on $\left( X^{(k)}, M^{(k)}, Y^{(k)} \right)_{k=1}^n$ for any $\alpha\in[0,1]$ is MCV-$\mathcal{D}$ if \emph{for any distribution $D \in \mathcal{D}$} and any $\alpha\in[0,1]$, we have:
\begin{equation*}
\mathds{P}_{D}\left(  Y ^{(n+1)} \in \widehat{C}_{n,\alpha} \left( X^{(n+1)},M^{(n+1)} \right) | M^{(n+1)} \right) \overset{a.s.}{\geq} 1-\alpha,
\end{equation*}
i.e., for any $m \in \mathcal{M}$ such that $\mathds{P}\left( M^{(n+1)} = m\right) > 0$, it holds:
\begin{equation*}
\mathds{P}_{D}\left(  Y ^{(n+1)} \in \widehat{C}_{n,\alpha} \left( X^{(n+1)},m \right) | M^{(n+1)} = m \right) \geq 1-\alpha.
\end{equation*}

\end{definition}

If $\mathcal{D} = \p\exch$ we recover the holy grail of being MCV for any exchangeable distribution, i.e., the most distribution-free result we could target. If $\mathcal{D}$ is not specified thereon, it will refer to MCV-$\p\exch$. An easier goal is to aim at MCV-$\p\iid$, that is MCV on i.i.d. distributions.

\begin{remark}
For any sets $\mathcal{D} \subseteq \mathcal{D}'$, a method that is MCV-$\mathcal{D'}$ is also MCV-$\mathcal{D}$, i.e., MCV-$\mathcal{D'}$ $\Rightarrow$ MCV-$\mathcal{D}$.
\label[remark]{rem:inc_mcv}
\end{remark}
\vspace{-0.5cm}
A naive idea to ensure MCV is to split the data set into $\#\mathcal{M}$ sub data sets, one for each mask, and run any marginally valid method on each of the data sets independently. However, as $\#\mathcal{M}$ grows exponentially in the dimension, this is not practical as it will generate small (or even empty) data sets for some masks. In particular, as long as $\mathds{P}(M = m)$ is low with respect to $n$ for a given $m \in \mathcal{M}$, estimation on the sub data set is hard, and even finite sample method such as conformal prediction \citep{vovk_algorithmic_2005} will suffer from important statistical variability or uninformativeness. Therefore, in practice, we usually need to go beyond this solution if we aim to achieve MCV for any mask, even those rarely observed at train times. Nevertheless, the task appears challenging without restricting the link between $M$ and $(X,Y)$, precisely due to the lack of information available in the data set. The question we tackle in this section is the following: \textbf{is it possible to achieve \emph{distribution-free} MCV in an informative way for any mask in $\mathcal{M}$, even those occurring with low probability?} 

\paragraph{Link with group conditional coverage.} More generally, the question is that of finding on which subspace of the features it is possible to obtain meaningful conditional guarantees. Thus, the results demonstrated in this section give some answers to the broader question of when is \emph{group-feature-conditional validity} achievable (a relaxation of \Cref{eq:pi_cv}), which has attracted considerable interest lately \citep[see e.g.,][to name just a few]{romano_malice_2020,barber_limits_2021,Guan2022,jung2023batch,gibbs2023conformal}.
\begin{mybox}
\begin{center}
\textbf{Our hardness results shed light on $X$-group-conditional coverage.}
\end{center}
In the rest of this section, $M$ can be interpreted as any additional random variable, that may (or may not) depend on $X$, on which we aim at achieving distribution-free conditional validity. For example, $\mathcal{M}$ could represent subgroups of $\mathcal{X}$, eventually overlapping. Specifically, assume $\mathcal{M} = \{0,1\}^{|\mathcal{G}|}$ for $\mathcal{G}$ a collection of groups on $\mathcal{X}$, then $M$ is an indicator vector on whether $X$ belongs to each group of $\mathcal{G}$ or not. 

A particular case of this generalization is $\mathcal{G} = \left\{ \{X \in \mathcal{X}: X_j \text{ is missing} \}_{j=1}^d \right\}$, recovering our missing covariates setting with $M$ the missing pattern. While our discussion in this section is written towards the missing covariates setting, the interested reader might replace ``missing pattern'' or ``mask'' by ``groups'' whenever it makes sense\footnotemark, and the corresponding result will hold without further restriction or assumptions on the way the groups are designed.
\end{mybox}\footnotetext{The only result that does not extend is \Cref{prop:tradeoff-mcv-mar-mcar} for $\p\mar$, as by construction it relies on the missingness structure.}
\subsection{Fully distribution-free result}

Our first result, \Cref{prop:tradeoff-mcv}, confirms the previous intuition: any MCV-$\p\iid$ method typically does output the whole set $\mathcal{Y}$ with high probability for any distribution, on low probability masks.

\vspace{0.25cm}
\noindent\fbox{%
    \parbox{\textwidth}{%
    \vspace{-0.25cm}
\begin{theorem}[Trade-off set size and mask probability]

Suppose that a method outputting $\widehat{C}_{n,\alpha}$ is MCV-$\p\iid$. Then \emph{for any $P \in \mathcal{P}$} and any $m \in \mathcal{M}$ such that $P_M(m) > 0$, it holds:

\begin{equation*}
\left\{\begin{aligned}
& \text{if }\mathcal{Y}\subseteq \mathds{R}  \text{ (regression)}: \mathds{P}_{P^{\otimes(n+1)}}\left( \Lambda\left(\widehat{C}_{n,\alpha}(X,m)\right) = \infty \right) \geq \textcolor{blindblue}{1 - \alpha} - \textcolor{mLightBrown}{\Delta_{m,n}},\\
& \text{if }\mathcal{Y}\subseteq \mathds{N} \text{ (classification)}: \forall y\in\mathcal{Y}, \mathds{P}_{P^{\otimes(n+1)}}\left( y \in \widehat{C}_{n,\alpha}(X,m) \right) \geq \textcolor{blindblue}{1 - \alpha} - \textcolor{mLightBrown}{\Delta_{m,n}},
\end{aligned}
\right.
\end{equation*}
with $\textcolor{mLightBrown}{\Delta_{m,n}} := \sqrt{2 \left( 1 - \left(1 - \frac{P_M(m)^2 }{2} \right)^{n+1} \right) }$.
\label[theorem]{prop:tradeoff-mcv}
\end{theorem}
\vspace{-0.25cm}
 }%
}
\vspace{0.25cm}

Since for any $x > 0$ and $n \in \mathds{N}^*$, it holds $1-\left( 1 - x\right)^n < nx$, \Cref{prop:tradeoff-mcv} implies that:
\begin{equation*}
\left\{\begin{aligned}
& \text{if }\mathcal{Y}\subseteq \mathds{R}  \text{ (regression)}: \mathds{P}_{P^{\otimes(n+1)}}\left( \Lambda\left(\widehat{C}_{n,\alpha}(X,m)\right) = \infty \right) \geq \textcolor{blindblue}{1 - \alpha} - \textcolor{mLightBrown}{P_M(m)\sqrt{(n+1)}},\\
& \text{if }\mathcal{Y}\subseteq \mathds{N} \text{ (classification)}: \forall y\in\mathcal{Y}, \mathds{P}_{P^{\otimes(n+1)}}\left( y \in \widehat{C}_{n,\alpha}(X,m) \right) \geq \textcolor{blindblue}{1 - \alpha} - \textcolor{mLightBrown}{P_M(m)\sqrt{(n+1)}}.
\end{aligned}
\right.
\end{equation*}

\Cref{prop:tradeoff-mcv} provides a lower bound on the probability that the predictive set is uninformative for any $m \in \mathcal{M}$ (i.e., typically $\Lambda(\widehat{C}_{n,\alpha}(\cdot,m) )=\infty$ or $\#\widehat{C}_{n,\alpha}(\cdot,m) \geq \#\mathcal{Y}(1-\alpha)$).

\begin{remark}[MCV-$\p\iid$ implies uninformative sets even on simple distributions]
    Crucially, \\this lower bound holds for \emph{any} distribution in $\p$. This implies that the hardness result applies also to smooth, nonpathological, distributions. Particularly, it means that any method that is fully distribution-free MCV (i.e., MCV-$\p\iid$) will be subject to the lower bound even when applied to data whose actual distribution is as simple as possible (e.g., MCAR and \yindmx).  
\label[remark]{rem:hardness_uniform}
\end{remark}
\begin{remark}[Informative sets implies the method is not MCV-$\p\iid$]
    Conversely, for a given method constructing predictive sets $\widehat{C}_{n,\alpha}$, assume that there exists a distribution $P\in\p$ and a mask $m$ such that $P_M(m) > 0$ and $\Delta_{m,n} < \frac{1-\alpha}{2}$ and under which $\widehat{C}_{n,\alpha}$ is consistently of finite measure for different random draws from $P\iid$. Then, this method is not MCV-$\p\iid$, as otherwise under $P\iid$ the predictive set would be of infinite measure with probability at least 0.25 for $\alpha \leq 0.5$ according to \Cref{prop:tradeoff-mcv} (since $1-\alpha-\Delta_{m,n} \geq \frac{1-\alpha}{2} \geq 0.25$).
\label[remark]{rem:exp_finite_nohardness}
\end{remark}

\paragraph{Interpretation of the lower bound.} Let us now decompose the lower bound. The first term, \textcolor{blindblue}{$1-\alpha$}, is an ``irreducible term''. Indeed, the estimator outputting $\mathcal{Y}$ with probability $1-\alpha$ and the empty set $\emptyset$ with probability $\alpha$ (where the probability corresponds to an exogenous Bernoulli random variable) is valid conditionally on everything, thus a fortiori on $M$. Hence, the lower bound has to be smaller than $1-\alpha$ as the set of MCV estimators includes this naive one.

For a given distribution $P$, the second term, \textcolor{mLightBrown}{$\Delta_{m,n}$}, becomes negligible on any $m \in \mathcal{M}$ such that $P_M(m)$ is small with respect to $n$, making the lower bound be nearly $1-\alpha$. This reflects the intuition that it is impossible to achieve informative conditional coverage when conditioning on events whose effective sample size is limited. In other words, the smaller the probability of the event occurring, the larger the training size must be to compensate and make ``sure'' that enough observations are drawn from that event. 

Note that as $\p\iid \subset \p\exch$, any MCV-$\p\exch$ estimator is MCV-$\p\iid$ by \Cref{rem:inc_mcv}. Thus, the conclusion of \Cref{prop:tradeoff-mcv} extends to any MCV-$\p\exch$ estimator, on any $P\iid$ with $P\in\p$.\footnote{The same is true for the subsequent \Cref{prop:tradeoff-mcv-mar-mcar} and \Cref{prop:tradeoff-mcv-yindm}.}

\paragraph{Proof sketch.} For any given distribution $P \in \p$, and a given mask $m\in\mathcal{M}$ such that $P_M(m) > 0$, the idea of the proof is the following. Build another distribution $Q \in \p$, which equals $P$ whenever $M \neq m$, and that ``admits'' an arbitrary spread on $Y$ when $M = m$ (in short, $Q$ is meant to be pathological yet close to $P$). By doing so, two statements can be made. First, $Q\iid$ belongs to $\p\iid$, therefore, as $\widehat{C}_{n,\alpha}$ is MCV-$\p\iid$, under $Q\iid$ the probability of $\widehat{C}_{n,\alpha}$ being uninformative is $1-\alpha$  since $Y$ can typically be anywhere. Second, as $P$ and $Q$ are the same everywhere except on $\{M=m\}$, the total variation distance between them is smaller than $P_M(m)$. This leads to the total variation distance between $P\iid$ and $Q\iid$ being smaller than $\Delta_{m,n}$. Combining these two observations, it finally leads to the probability of $\widehat{C}_{n,\alpha}$ being uninformative under $P\iid$ which is greater than $1-\alpha-\Delta_{m,n}$. The complete proof is given in \Cref{app:mcv_hardness}.

A familiar reader will note the similarity with the proofs given by \citet{lei_distribution-free_2014,vovk_conditional_2012}. The difference is that, on the one hand, \citet{vovk_conditional_2012} proof leverages an ``reductio ad absurdum'' that does not allow to explicitly build the set on which $P \neq Q$. On the other hand, \citet{lei_distribution-free_2014} is constructive. Nonetheless, it relies on a crucial step that implicitly assumes that conditional-validity holds conditionally on the $n$ data points, leading to an inexact statement: the lower bound obtained becomes 1. As we discussed, as well as \citet{vovk_conditional_2012}, the lower bound can not be bigger than $1-\alpha$. We provide an alternate proof to this well-known $X$-conditional impossibility result that is constructive. Another improvement is that our expression of $\Delta_{m,n}$ comes from a tighter inequality than the ones used in \citet{lei_distribution-free_2014} and \citet{vovk_conditional_2012}. Indeed, for the original impossibility result, the lower bound does not really matter as we then take its limit when the ball around $x_0$ shrinks, which is 0. But in our case, this ball is fixed to the event $\{M=m\}$.  

\subsection{Restricting the class of admissible missingness distributions}

Interestingly, the proof of \Cref{prop:tradeoff-mcv} adapts to MCV-$\p\mar\iid$ or MCV-~$\p\mcar\iid$.

\begin{proposition}[Trade-off set size and mask probability on $\p\mar$ or $\p\mcar$]
Let $\mathcal{Q}$ be either $\p\mar$ or $\p\mcar$.
Suppose than an estimator $\widehat{C}_{n,\alpha}$ is MCV-$\mathcal{Q}\iid$ at the level $\alpha$. Then \emph{for any $Q \in \mathcal{Q}$} and any $m \in \mathcal{M}$ such that $Q_M(m) > 0$, it holds:

\begin{equation*}
\left\{\begin{aligned}
& \text{if }\mathcal{Y}\subseteq \mathds{R}  \text{ (regression)}: \mathds{P}_{Q\iid}\left( \Lambda\left(\widehat{C}_{n,\alpha}(X,m)\right) = \infty \right) \geq 1 - \alpha - \Delta_{m,n},\\
& \text{if }\mathcal{Y}\subseteq \mathds{N} \text{ (classification)}: \forall y\in\mathcal{Y}, \mathds{P}_{Q\iid}\left( y \in \widehat{C}_{n,\alpha}(X,m) \right) \geq 1 - \alpha - \Delta_{m,n},
\end{aligned}
\right.
\end{equation*}
with $\Delta_{m,n}$ given in \Cref{prop:tradeoff-mcv}.
\label[proposition]{prop:tradeoff-mcv-mar-mcar}
\end{proposition}

\begin{remark}[no direct implication between results]
\Cref{prop:tradeoff-mcv-mar-mcar} for $\mathcal{Q} = \p\mar$ does not imply \Cref{prop:tradeoff-mcv-mar-mcar} for $\mathcal{Q} = \p\mcar$, nor the contrary. Indeed, on the one hand, as $\p\mcar\iid \subseteq \p\mar\iid$, any method that is MCV-$\p\mar\iid$ is MCV-$\p\mcar\iid$ (\Cref{rem:inc_mcv}). However, on the other hand, \Cref{prop:tradeoff-mcv-mar-mcar} (or \Cref{prop:tradeoff-mcv}) provides a uniform statement over $Q\in\mathcal{Q}$ (\Cref{rem:hardness_uniform}): as $\p\mcar\iid \subseteq \p\mar\iid$, the final statement holds on more distributions for $\mathcal{Q}  = \p\mar$ than for $\mathcal{Q}  = \p\mcar$. Thereofore, \Cref{prop:tradeoff-mcv-mar-mcar} for $\mathcal{Q} = \p\mar$ provides a \emph{stronger statement} over \emph{fewer methods} than \Cref{prop:tradeoff-mcv-mar-mcar} for $\mathcal{Q} = \p\mcar$.

For the same reason, \Cref{prop:tradeoff-mcv-mar-mcar} is not deduced directly from \Cref{prop:tradeoff-mcv}, but from a careful consideration of the construction in its proof: the adversarial distribution built therein does not make any assumption on the relationship between $X$ and $M$, which can be as simple as desired.
\end{remark}

In fact, the key point for the proof of \Cref{prop:tradeoff-mcv} is that the algorithm achieves MCV also on distributions under which $Y$ and $M$ can be dependent even conditionally on $X$: thus, it allows us to construct an adversarial distribution under which $Y$ is equally likely to be anywhere on the label space for a given $m \in \mathcal{M}$. 

In view of this, one could think that in order to break \Cref{prop:tradeoff-mcv}, and therefore to ensure that MCV is achievable in an informative way even on low probability masks, we have to \emph{at least} assume \yindmx (\ref{ass:y_ind_m}). However, in \Cref{prop:tradeoff-mcv-yindm}, we show that even estimators that are only MCV-$\p\ymx\iid$ suffer from the same trade-off on efficiency. 

\begin{proposition}[Trade-off set size and mask probability on $\p\ymx$]

Suppose that an estimator $\widehat{C}_{n,\alpha}$ is MCV-$\p\ymx\iid$ at the level $\alpha$. Then for any $P\in\p\ymx$ and for any $m \in \mathcal{M}$ such that $\frac{1}{\sqrt{2}} \geq P_M(m) > 0$, it holds:

\begin{equation*}
\left\{\begin{aligned}
& \text{if }\mathcal{Y}\subseteq \mathds{R}  \text{ (regression)}: \mathds{P}_{P\iid}\left( \Lambda\left(\widehat{C}_{n,\alpha}(X,m)\right) = \infty \right) \geq \textcolor{blindblue}{1 - \alpha}  - \textcolor{mLightBrown}{\Delta_{m,n}},\\
& \text{if }\mathcal{Y}\subseteq \mathds{N} \text{ (classification)}: \forall y\in\mathcal{Y}, \mathds{P}_{P\iid}\left( y \in \widehat{C}_{n,\alpha}(X,m) \right) \geq \textcolor{blindblue}{1 - \alpha}  - \textcolor{mLightBrown}{\Delta_{m,n}},
\end{aligned}
\right.
\end{equation*}
with $\textcolor{mLightBrown}{\Delta_{m,n}} := \sqrt{2 \left( 1 - \left(1 - \textcolor{mLightBrown}{2} P_M(m)^2 \right)^{n+1} \right) }$.
\label[proposition]{prop:tradeoff-mcv-yindm}
\end{proposition}

All in all, \Cref{prop:tradeoff-mcv-yindm} demonstrates that even the simplest relationship between $Y$ and $M$ does not allow informative predictive sets. This reveals that to ensure that it is possible to obtain informative sets even on low probability masks (or events), one has to design a method that will be conditionally valid \emph{only} on distributions with a constrained structure of dependence between $Y$ and $M$ given $X$, but also between $M$ and $X$. In particular, trying to ensure MCV-$\p\mcarymx\iid$ (where $\p\mcarymx\iid := \p\mcar\iid \cap \p\ymx\iid$) as done in \citet{pmlr-v202-zaffran23a} appears as a natural way to approach the minimal set of assumptions. 

\begin{remark}
In \Cref{fig:mcar}, we illustrate that, on a distribution $P\in\p\mcarymx\iid$, a provably MCV-$\p\mcarymx\iid$ method (introduced in \Cref{sec:nested}) consistently outputs finite length predictive intervals (regression case). Therefore, we can conclude that obtaining a hardness result on $\p\mcarymx\iid$ appears impossible, as such it would induce \Cref{rem:exp_finite_nohardness} (with $\p\mcarymx\iid$ instead of $\p\iid$).
\label[remark]{rem:nohardness_mcar_yindmx}
\end{remark}

\section{Link between missing covariates and predictive uncertainty}
\label{sec:glm}

In light of the previous section, MCV appears hard to achieve. Thus, the problem that we aim to address now is to \textbf{find ways to model properly the missing covariates' influence on predictive uncertainty}.
To understand the relationship between missing values and predictive uncertainty, this section explores simplified distributions on $(X,M,Y)$---such as MCAR and \yindmx---and/or on $(X,Y)$---such as linearity, Gaussianity. We consider the regression case with $\mathcal{Y} = \mathds{R}$.
This exploration aims to facilitate the development of suitable frameworks for probabilistic inference when covariates are missing---i.e., 
models that are as close as possible to achieving MCV.

\subsection{Increasing uncertainty with nested masks}
\label{sec:glm_var}
The hardness results of \Cref{sec:hardness} induce that MCV cannot be (efficiently) achieved without structural assumptions on the links between the predictive distributions conditional on each missing pattern. In this subsection, we gain insights into the underlying reasons for this phenomenon: the predictive uncertainty depends on the missing pattern, a form of \textit{heteroskedasticity}. In summary, we explore the following idea, which is a natural modelization attempt in that direction:
\begin{center}
   \textbf{Idea:} \textit{The predictive uncertainty increases when less covariates are observed.}
\end{center}
In technical words, the aforementioned heteroskedasticity  takes the form of an \textit{isotonicity} (monotony) with respect to the mask, with the  inclusion order given by \Cref{def:included_masks} below. In short: the more missing values, the more uncertainty there is.  

\begin{definition}[Included masks] 
\label[definition]{def:included_masks}
Let $(m,m') \in \mathcal{M} ^2$,  $m \subset m'$ if for any $j \in \llbracket 1,d \rrbracket$ such that $m_j = 1$ then $m'_j = 1$, i.e., $m'$ includes at least the same missing values than $m$.
\end{definition}
\begin{table}[b]
\begin{center} 
\begin{tabular}{lccc}
\toprule
 \diagbox{Property}{Setup} & \Cref{mod:glm}     & \Cref{ex:hetero} &  $\p\mcarymx$ \\
 \midrule
Variance  &   \ref{eq:var_ps}   & \cancel{ \ref{eq:var_ps} }\ \ \ref{eq:var_exp}   &  \ref{eq:var_exp} \\
Inter-quantile &  \ref{eq:IQ_ps} &  \ref{eq:IQ_exp} \\
Length of Oracle PI  & \ref{eq:CIL_isot_ps} & \ref{eq:CIL_isot_exp} & \ref{eq:CIL_isot_exp} \\
  \bottomrule
\end{tabular}
\end{center}
    \caption{Summary of the results from \Cref{sec:glm_var}.}
    \label{tab:my_label}
\end{table}
Hereafter, we formally quantify such a statement, in particular in terms of conditional variance, inter-quantile distance, and predictive interval length. We demonstrate that some of those statements are valid, to different extent, under distributional assumptions, either generic or on specific model or examples. To that end, we introduce several properties, that can be considered as non-parametric assumptions on the underlying distributions. We put together some results of this section in the following \Cref{tab:my_label}, that can be used as a reading guide throughout the section. 

\subsubsection{Conditional Variance Isotony w.r.t.~the missing data patterns}

We start by focusing on the link between $M$ and the \textit{conditional variance} of $Y|X_{\obs(M)}$, that constitutes a natural proxy on the predictive uncertainty. Denote $V(X_{\text{obs}(M)}, M ):= \mathrm{Var} \left( Y | X_{\text{obs}(M)}, M \right)$ the conditional variance of $Y$ given $\left(X_{\text{obs}(M)}, M\right)$. We introduce two properties regarding its ordering with respect to $M$: \eqref{eq:var_ps} and \eqref{eq:var_exp}.
\begin{align}
V(X_{\text{obs}(m)}, m ) & \overset{a.s.}{\leq} V(X_{\text{obs}(m')}, m' ) &\text{ for any } m \subset m', 
\label{eq:var_ps} \tag{Var-1}\\
\mathds{E}\left[ V(X_{\text{obs}(M)}, M ) | M = m \right] & \leq \mathds{E}\left[ V(X_{\text{obs}(M)}, M )  | M = m' \right] 
&\text{ for any } m \subset m'. 
 \tag{Var-2} 
\label{eq:var_exp}
\end{align}
Property \ref{eq:var_ps} is stronger than Property \ref{eq:var_exp} as it is an almost sure result w.r.t.~the covariates $X$.
The following proposition ensures that \eqref{eq:var_exp} is satisfied under $\p\mcarymx$ (that is, assumptions for which no hardness result can exist).
\begin{proposition}\label[proposition]{prop:VarExp}
    Under $\p\mcarymx$, \eqref{eq:var_exp} is valid.
\end{proposition}
The proof of this result is given in \Cref{app:proof:varirianceisot}. This is a first significant result: under general assumptions---i.e., strong assumption on the relation between the mask and both the response and the features, but no assumptions on their distribution---, the averaged variance is always smaller on smaller masks. This establishes the existence of a link between the uncertainties \textit{on patterns that can be compared}, that is patterns that are nested in one another. Note that the order given by \Cref{def:included_masks} is only a partial order: the average variance ordering is only enforced w.r.t.~that partial order.

It is possible that the predictive uncertainty increases on average with the mask (\Cref{eq:var_exp}) but not almost surely on $X$ (\Cref{eq:var_ps}), as illustrated by  \Cref{ex:hetero} below: 
\begin{model}[Unidimensional heteroskedasticity]
\label[model]{ex:hetero}
Consider the following one-dimensional model:
\begin{itemize}[noitemsep,topsep=0pt]
    \item $X \sim \mathcal{N}(0,\sigma^2)$, $\sigma \in \mathds{R}_+$;
    \item $\xi \sim \mathcal{N}(0,\tau^2)$, $\tau \in \mathds{R}_+$, such that $\xi \ind X$;
    \item $Y = \beta X + X \xi$, with $\beta \in \mathds{R}$;
    \item $M \sim \mathcal{B}(\rho)$, with $\rho \in [0,1]$, and $M\ind (X,Y)$.
\end{itemize}
\end{model}
Under this model, we obtain that $M \ind X$ (MCAR) and \yindmx, and
\begin{align*}
& \left\{\begin{aligned}
& \mathrm{Var}(Y | X, M = 0) = \tau^2 X^2 \\
& \mathrm{Var}(Y | M = 1) = (\beta^2+\tau^2)\sigma^2
\end{aligned}\ \ 
\right. \ \ \ \
\Rightarrow & \left\{\begin{aligned}
& \mathds{E}\left[ \mathrm{Var}(Y | X, M = 0) \right] = \tau^2 \sigma^2  \\
& \mathds{E}\left[ \mathrm{Var}(Y | M = 1) \right] = (\beta^2+\tau^2)\sigma^2
\end{aligned}\,\,.
\right.
\end{align*}
Thus \Cref{eq:var_exp} is verified but \Cref{eq:var_ps} is only satisfied for $X$ such that $X^2 \leq \left( 1 + \frac{\beta^2}{\tau^2} \right) \sigma^2$. This is illustrated in \Cref{fig:var_hetero}. The first subplot represents $Y$ depending on $X$, while the third subplot displays $Y-\beta X$ depending on $X$, that is an illustration of the uncertainty of the distribution of $Y | X$. For any $X$ outside the vertical dashed lines (corresponding to $\pm ( 1 + {\beta^2}/{\tau^2} ) \sigma^2$), the conditional variance of $Y$ given $X$ is larger than the overall variance when $X$ is missing. 
Yet, the average variance of $Y$ when $X$ is missing is indeed higher than the average variance of $Y$ when $X$ is observed: this can be seen on the two histograms on subplots 2 and 4.
\begin{figure}[!t]
    \centering
    \includegraphics[width=0.7\textwidth]{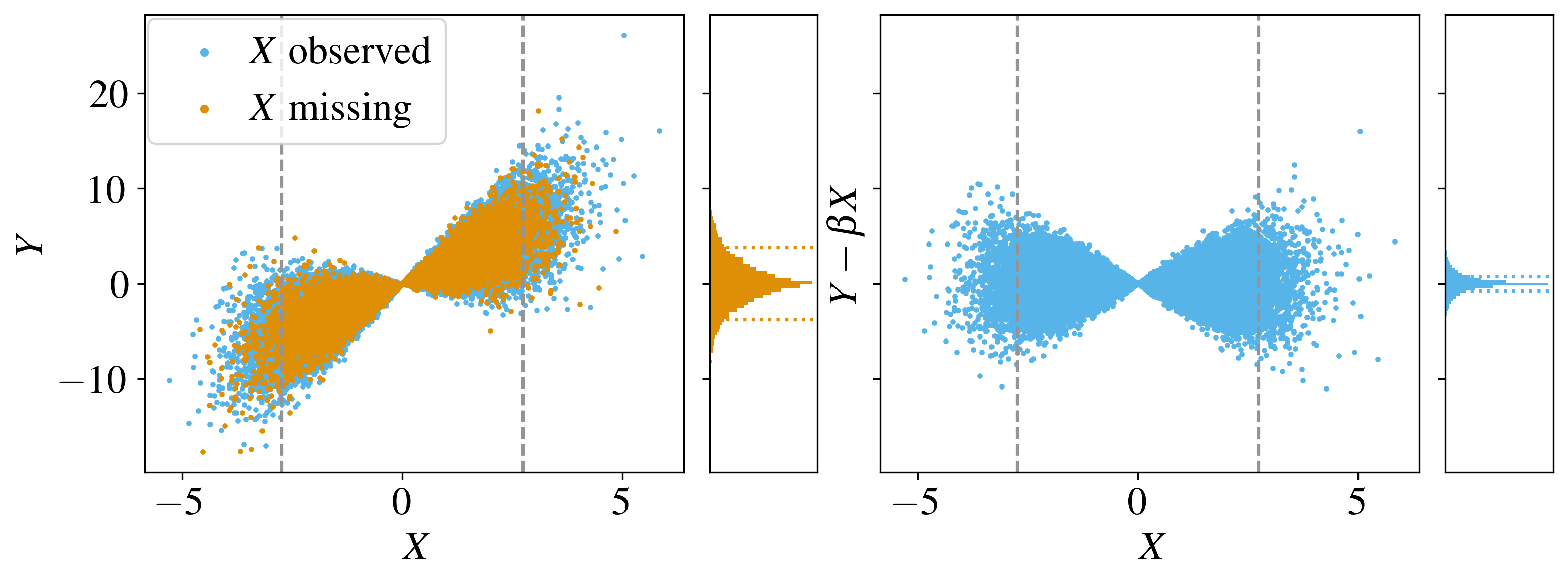}
    \caption{Visualisation of a random draw from the data distribution of \Cref{ex:hetero}, with $100000$ i.i.d. samples, $\rho = 0.2$, $\sigma^2=1.5$, $\tau^2=1$ and $\beta=2$. The colors indicate whether $X$ is observed or missing. The first subplot represents $Y$ depending on $X$, while the third subplot displays $Y-\beta X$ depending on $X$ only for observed $X$, that is an illustration of the uncertainty of $Y | X$. The second subplot is an histogram of $Y$ when $X$ is missing, while the forth subplot is an histogram of $Y-\beta X$ when $X$ is observed, i.e., they represent the predictive distribution of $Y$ depending on whether $X$ is observed or missing.}
    \label{fig:var_hetero}
\end{figure}

Finally,  while \Cref{ex:hetero} shows that \eqref{eq:var_ps} is not always true, even under the assumptions of \Cref{prop:VarExp}, we now show that it can be achieved in the following Gaussian linear model, a particular case of Gaussian pattern mixture model.
\newpage
\begin{model}[Gaussian linear model (GLM)]
\label[model]{mod:glm}
The data is generated according to a linear model and the covariates are Gaussian conditionally to the pattern:
\begin{itemize}[topsep=0pt,noitemsep,leftmargin=*]
	\item $Y = \beta^T X + \varepsilon$, $\varepsilon \sim \mathcal{N}(0, \sigma^2_{\varepsilon}) \perp\!\!\!\!\perp (X,M)$, $\beta \in \mathds{R}^d$.
	\item for all $m \in \mathcal{M}$, there exist $\mu^{m}\in\mathds{R}^d$ and $\Sigma^{m}\in\mathds{R}^{d\times d}$ such that
   $ X | (M = m) \sim \mathcal{N}(\mu^{m}, \Sigma^{m}).$
\end{itemize}
\end{model}
Such a model results in a MCAR distribution when $\Sigma^m \equiv \Sigma$. Indeed under \Cref{mod:glm} the resulting predictive distribution is given by $Y | ( X_{\obs(m)},M = m ) \sim \mathcal{N}\left( \tilde\mu^m, \widetilde\sigma^m \right)$ for any $m \in \mathcal{M}$,
with:
\begin{align*}
\tilde\mu^m = & \; \beta^T_{\obs(m)} X_{\obs(m)} + \beta^T_{\mis(m)} \mu^m_{{\mis}|{\obs}}, \\
\widetilde\sigma^m = & \; \beta^T_{\mis(m)} \Sigma^m_{{\mis}|{\obs}} \beta_{\mis(m)} + \sigma^2_{\varepsilon},
\end{align*}
with $\mu^m_{{\mis}|{\obs}}$ and $\Sigma^m_{{\mis}|{\obs}}$ defined in \Cref{app:glm}  \citep{lemorvan2020,ayme2022,pmlr-v202-zaffran23a}.
Crucially, $\widetilde\sigma^m$ depends on $m$ in a non-linear fashion, even in MCAR. That is, even in MCAR and a homoskedastic model for $Y|X$, the predictive distribution of $Y|X_{\obs(M)}$ is heteroskedastic: basically, the distribution of $Y$ is a mixture of various distributions with the mask being the latent variable. This simple example already illustrates that missing values generate strong heteroskedasticity: in \Cref{prop:glm_var_incr}, we show that under this \Cref{mod:glm} and $\p\mcar$, the variance of the conditional distribution of $Y$ increases when the missing pattern increases (in the sense of \Cref{def:included_masks}).

\begin{proposition}[Conditional variance increases with the mask under MCAR GLM]
\label[proposition]{prop:glm_var_incr}
Under \\ \Cref{mod:glm} and $\p\mcar$, if the covariance matrix $\Sigma$ is positive definite, \Cref{eq:var_ps} is satisfied.
\end{proposition}

To prove that the variance increases with the pattern, we prove that for any $m\subset m'$, $\Sigma^{m'}_{\mis|\obs} \succcurlyeq \begin{pmatrix}
    \Sigma^m_{\mis|\obs} & 0 \\ 0 & \mathbf{0}
\end{pmatrix}$. This  is  proved in \Cref{app:glm}. 

Next, in order to go beyond variances, we focus on inter-quantile distances as a measure of uncertainty, and establish a general result on the expected length of oracle predictive intervals.

\subsubsection{Conditional Inter-quantile Isotony w.r.t.~the missing data patterns}

Ideally, we would like to access the oracle predictive interval (the interval satisfying \Cref{eq:na_mcv} with minimal expected length). Thus, in this section we are interested in characterizing its behavior with respect to $M$, in order to be able to mimic it. We denote this interval $\cna^{*,P}$, that is formally defined for any $m\in\mathcal{M}$ as:
\begin{align*}
  \cna^{*,P} (\cdot, m):= \argmin_{\substack{\cna : \mathcal X \times \mathcal M \to \mathdutchcal{P}(\mathds R) \\ \mathrm{s.t.} \mathds{P}_P( Y \in \mathcal{C}_\alpha (X,m) | M = m ) \geq 1 - \alpha, }} \mathds E_P[\Lambda (  \cna(X_{\obs(m)},m) ) | M=m].
\end{align*}

In fact, under \Cref{mod:glm}, the oracle predictive interval is uniquely defined by the quantiles $\alpha/2$ and $1-\alpha/2$ of the $\mathcal{N}\left( \tilde\mu^m, \widetilde\sigma^m \right)$. More importantly, this oracle interval even achieves $X$-conditional coverage. 
\Cref{prop:glm_var_incr} shows that under $\p\mcar$ and \Cref{mod:glm}, increasing the number of missing values (in a nested way) induces an increase in the predictive uncertainty of $Y$:  the oracle intervals, that are given by inter-quantiles intervals, are nested. Notably, this is true almost surely on $X_{\obs}$ and not only marginally. 

To generalize this property beyond the Gaussian case, we introduce the inter-quantile distance, that encodes the uncertainty for  conditional predictive distribution. 
For all $\beta \le \frac{1}{2}\le \gamma$, we define the inter-quantile space for quantile distributions:
\begin{align*}
    \mathrm{IQ}_{\beta,\gamma} (X_{\text{obs}(M)}, M) = q_\gamma(\mathds P_{Y | X_{\text{obs}(M)}, M }) -q_\beta(\mathds P_{Y | X_{\text{obs}(M)}, M } ).
\end{align*}
And the following two assumptions, that are similar in spirit to \eqref{eq:var_ps} and \eqref{eq:var_exp} 
\begin{align}
\mathrm{IQ}_{\beta,\gamma} (X_{\text{obs}(m)}, m) & \overset{a.s.}{\leq} \mathrm{IQ}_{\beta,\gamma} (X_{\text{obs}(m')}, m') &\text{ for any } m \subset m', 
\label{eq:IQ_ps} \tag{IQ-1}\\
\mathds{E}\left[ \mathrm{IQ}_{\beta,\gamma} (X_{\text{obs}(M)}, M ) | M = m \right] & \leq \mathds{E}\left[ \mathrm{IQ}_{\beta,\gamma} (X_{\text{obs}(M)}, M)   | M = m' \right] \;\; &\text{ for any } m \subset m'. 
 \tag{IQ-2} 
\label{eq:IQ_exp}
\end{align}
The assumptions on the quantiles and the variance are equivalent for Gaussian (conditional) distributions. As a consequence, \eqref{eq:IQ_exp} is satisfied under \Cref{mod:glm} and $\p\mcar$ as well as under \Cref{ex:hetero}, while \eqref{eq:IQ_ps} is satisfied only under \Cref{mod:glm} and $\p\mcar$. Inter-quantile assumptions are related to predictive intervals: for any distribution $P$ such that $P_{Y | X_{\text{obs}(M)}, M }$ is a.s.~unimodal, the oracle predictive interval $\cna^{*,P} $ writes as an inter-quantile interval almost surely, that is there exist functions $\beta,\gamma : \mathcal X \times \mathcal M \to [0,1]$ such that
\begin{align*}
    &\cna^{*,P} (X_{\obs(M)}, M) \overset{\mathrm{a.s.}}{ =} \left[q_{\beta(X_{\obs(M)}, M)} (P_{Y | X_{\text{obs}(M)}, M }); q_{\gamma(X_{\obs(M)}, M)} (P_{Y | X_{\text{obs}(M)}, M })\right]\\
    &\mathds E_P[\gamma(X_{\obs(M)}, M)-\beta(X_{\obs(M)}, M)|M] \overset{\mathrm{a.s.}}{ =} 1-\alpha.
\end{align*}
Indeed, to minimize the average length, the best oracle solution consists in minimizing the length conditionally to $(X_{\obs(M)}, M)$, which is achieved by an inter-quantile interval, under the unimodality assumption. The quantity $\gamma(X_{\obs(M)}, M)-\beta(X_{\obs(M)}, M)$ corresponds to the $(X_{\obs(M)}, M)$--conditional coverage, that is on average, conditionally to $M=m$, the targeted $1-\alpha$.

Yet, in practice, the constructed intervals are not the oracle ones. Therefore, a natural question is whether \eqref{eq:IQ_exp} extends to a non-oracle $\mathcal{C}_\alpha$. As generally $\mathcal{C}_\alpha$ is not based on the underlying true conditional quantiles, we focus on $\mathcal{C}_\alpha$ length instead, a quantity similar in spirit to the inter-quantile. We consider the two following assumptions:
\begin{align}
\Lambda(\cna(X_{\text{obs}(m)}, m )) & \overset{a.s.}{\leq}  \Lambda(\cna(X_{\text{obs}(m')}, m' ))&\text{ for any } m \subset m', 
\label{eq:CIL_isot_ps} \tag{Len-1}\\
\mathds{E}\left[ \Lambda(\cna(X_{\text{obs}(M)}, M ))| M = m \right] & \leq \mathds{E}\left[ \Lambda(\cna(X_{\text{obs}(M)}, M )) | M = m' \right] 
&\text{ for any } m \subset m'. 
\label{eq:CIL_isot_exp} \tag{Len-2}
\end{align}
We have the following \Cref{prop:len-dec} on isotonicity (\ref{eq:CIL_isot_exp}) under $\p\mcarymx$.

\vspace{0.25cm}
\noindent\fbox{%
    \parbox{\textwidth}{%
    \vspace{-0.25cm}
\begin{theorem}\label{prop:len-dec}
Let $\cna$ be an MCV-$\p\mcarymx$ predictive interval. There exists a predictive interval $\widetilde \cna$ which
\begin{enumerate}[label=\roman*)]
    \item is MCV-$\p\mcarymx$,
    \item has conditional length smaller or equal to that of $\cna$ on each pattern,
    \item is averaged-length-isotonic w.r.t.~the patterns, i.e., satisfies~\eqref{eq:CIL_isot_exp}.
\end{enumerate}
\end{theorem}
    \vspace{-0.25cm}
 }%
}
\vspace{0.25cm}

The proof of \Cref{prop:len-dec} exploits the fact that under $\p\mcarymx$, a strategy to ensure conditional coverage w.r.t.~a pattern $m$, is to transform $(X_\obs(m),m)$ into $(X_\obs(m'), m')$ by additionally masking some entries, and using the predictive interval given on pattern $m'$. For $m\subset m'$, we denote  $X_{\obs(\max(m,m'))}$ the point in which we additionally mask elements of $m'$ in $X$. We have that under $\p\mcarymx$, the distribution of the data \textit{post-masking} is equal to that of the data with more missing entries:  $\mathds P_{Y | X_{\text{obs}(\max(M,m'))}, \max(M, m') }=P_{Y | X_{\text{obs}(m')}, M=m'  }$. We can leverage this observation to build intervals: if the averaged length of the predictive interval conditionally to a pattern $m\subset m'$ is larger than that conditionally to a pattern $m\subset m'$, we can replace $\cna(X_{\obs(m)},m)$ by $\cna(X_{\obs(m')},m')$, ensuring both that new interval length is smaller and that we satisfy \eqref{eq:CIL_isot_exp}.
Formally, we proceed by induction: starting from the pattern $m'= (1, \dots, 1)$ (no data observed), we first consider all patterns $m= (1, \dots,1, 0, 1, \dots) $ with a single observed value, and define $\widetilde \cna (X_{\obs(M)}, M)$, conditionally to $M=m$, as either $\cna (X_{\obs(M)}, M)$ or $\cna (X_{\obs(\max(M,m'))},  \max(M,m'))$ (depending on which expected length is smaller). We then repeat the same reasoning inductively. For a pattern $m$, we pick for $\widetilde \cna$ either $\cna(\cdot, m)$  or the minimal-length interval among all $\cna (\cdot, m')$ for all patterns $m'$ that have one more missing data than $m$, and artificially mask on of the components of $X_{\obs(m)}$ when predicting.

\textbf{Interpretation:} we leverage towards predictive interval construction the fact that we can transform an observed point, by removing some covariates, and recover the same distribution as the one with more missing entries.  This idea will be one of the key techniques leveraged in \Cref{sec:nested}.

As consequence of \Cref{prop:len-dec} is the following corollary, that is obtained by a minimality argument for the oracle interval (i.e., knowing that applying the aforedmentioned transformation to the oracle does not change it, as it already has minimal-expected length on each pattern):
\begin{corollary}\label{prop:oracle-len-dec}
    Let $P \in \p\mcarymx$. Then the oracle interval $\cna^{*,P}$ is averaged-length-isotonic w.r.t.~the patterns, i.e., satisfies~\eqref{eq:CIL_isot_exp}.
\end{corollary}
Overall, \eqref{eq:CIL_isot_exp} is thus satisfied by our target sets under $\p\mcarymx$, and thus appears as a reasonable constraint to impose on our predictive sets. Indeed, it seems to be close to the minimal set of assumptions required in order to ensure that no hardness result exists (\Cref{sec:hardness}) while inducing a leverageable structure between patterns that can be compared (\Cref{prop:len-dec}).

\subsection{Guidelines for practitioners: which information through imputation for quantile regression?}
\label{sec:glm_mask}

In this section, we highlight specifities of predictive uncertainty quantification under missing covariates with respect to mean regression, and provide generic guidelines usable in practice.

\paragraph{Impute-then-predict.} Most predictive algorithms can not cope directly with missing covariates. To bypass this, the most common approach is to impute the incomplete data via an imputation function $\Phi$, that maps observed values to themselves and missing values to a function of the observed values. Using notations from \citet{lemorvan2021} we note $\phi^{m} : \mathds{R}^{|\obs(m)|} \rightarrow \mathds{R}^{|\mis(m)|}$ the imputation function which, given a mask $m \in \mathcal{M}$, takes as input observed values and outputs imputed values, i.e., plausible values. Then, the overall imputation function $\Phi$ belongs to $\mathcal{F}^I := \left\{ \Phi : \mathcal{X} \times \mathcal{M} \rightarrow \mathcal{X} : \forall j \in \llbracket 1,d \rrbracket, \left( \Phi \left( X, M \right) \right)_j = X_j\mathds{1}_{M_j = 0} + \left(\phi^{M}\left(X_{\text{obs}(M)}\right)\right)_j \mathds{1}_{M_j = 1}  \right\}$.  
The imputed data set becomes the $n$ random variables $\left(\Phi\left(X, M\right), M, Y\right)$. In practice, $\Phi$ is the result of an algorithm $\mathcal{I}$ trained on $\left\{ \left(X^{(k)},M^{(k)} \right) \right\}_{k=1}^{n+1}$.
The impact of imputation has been studied for mean regression tasks \citep[in particular in][]{lemorvan2021,pmlr-v202-ayme23a,ayme2024random}.

\paragraph{How to account for the missingness when imputing?}  Given the impact of missing covariates on the shape of prediction uncertainty discussed in \Cref{sec:glm_var}, impute-then-predict raises a specific concern: is there a way to impute which incorporates the necessary information on the missing values?

Hereafter, we show that the answer is significantly different if we restrict ourselve to mean regression. Specifically, we show that incorporating the mask (e.g., by concatenating the mask to the features) is more critical for quantile regression. To that end, we provide in \Cref{prop:impute_mean_qr} simple models showcasing that unbiased imputation choices are sufficient to obtain an optimal model for regression but fail for quantile regression. For mean regression, the efficiency of such imputation methods have been established in practice \citep[see e.g.,][]{josse2019,lemorvan2021} and \Cref{prop:impute_mean_qr} support those findings.

\begin{proposition}
Assume $\p\mcarymx$ and $Y = {\beta^*}^T X + \varepsilon$ with $\varepsilon$ s.t. $\mathds{E}\left[\varepsilon | X_{\obs(M)} , M \right] = 0$.
\begin{enumerate}[label=\roman*)]
\item\label{itm:mean} Mean regression
\begin{itemize}[topsep=0pt,noitemsep,leftmargin=*]
    \item if the covariates $(X_j)_{j=1}^d$ are independent, then the optimal linear model taking $\Phi_{\rm{mean}}(X,M)$ as input is Bayes optimal, with $\Phi_{\rm{mean}}$ the imputation by the mean;
    \item the optimal linear model taking $\Phi_{\rm{conditional \; mean}}(X,M)$ as input is Bayes optimal, with $\Phi_{\rm{conditional \; mean}}$ the imputation by the conditional mean;
\end{itemize}
\item\label{itm:qr} Any quantile linear model taking unbiased imputed data as input (i.e., $\mathds{E}\left[ \Phi(X,M) | M  \right] \overset{a.s.}{ =} \mathds{E}\left[ X \right]$) leads to intervals of constant expected length across patterns, thus is not Bayes optimal when $Y\not\ind X$.
\end{enumerate}
\label[proposition]{prop:impute_mean_qr}
\end{proposition}

Point \ref{itm:mean} of \Cref{prop:impute_mean_qr} illustrates that if the learner was able to retrieve the true underlying regression coefficients and the data were imputed by their mean or conditional mean, then the learned model would be the best possible at the task of predicting the conditional expectation, i.e., all necessary information is preserved by using only the imputed data set and not leveraging the associated mask. Although the non-necessity of using the mask in the conditional expectation estimation and MCAR framework does not systematically extend when the data is more complex than linear, it is insightful as even in the linear setting, the same does not hold for quantile~regression. 

Indeed, point \ref{itm:qr} of the same \Cref{prop:impute_mean_qr} highlights that on the contrary a learner accessing the true underlying regression coefficients with the very same unbiased imputed data would not lead to an optimal model, as a method whose resulting predictive interval have constant lengths across the missing patterns does not retrieve the underlying heteroskedasticity induced by the missing values (\Cref{sec:glm_var}), and thereby cannot be MCV. Precisely, the assumption on the imputation for this result corresponds for example to imputing by the feature's expectation (i.e., $\Phi_{\rm{mean}}$), the feature's conditional expectation (i.e., $\Phi_{\rm{conditional \; mean}}$), or a random draw from a distribution whose expectation is the feature's expectation, under $\p\mcar$. This includes MICE \citep{JSSv045i03}, which consists in imputing by random draws from the conditional distribution hence the imputed data have the same expectation than the features themselves. 

Overall, \Cref{prop:impute_mean_qr} tells that $i)$ the state-of-the-art imputation method MICE is not the best choice for predictive uncertainty quantification, $ii)$ by contrast to mean regression, in the linear case imputing by the expectation or the conditional expectation is detrimental. In fact, data-independent constant imputation would result in more adaptive intervals. This is because quantile regression needs to retrieve the information on the patterns to adapt its structure to it. Therefore, when using such imputations, \textbf{a natural idea is to give the information of the mask to the model by concatenating the mask to the features}.

\section{Principled unified Missing Data Augmentation (MDA) framework: \newnested}
\label{sec:nested}

In this section, we go beyond generic guidelines and we introduce a general framework, coined \newnested, to generate predictive sets that achieve MCV under $\p\mcarymx$. Our approach is applicable to both classification and regression tasks, by building upon any conformal score function \citep{vovk_algorithmic_2005}. It combines over-masking ideas introduced in \Cref{sec:glm}, with subsampling techniques, and similar machinery than leave-one-out conformal prediction methods \citep{barber2021jackknife,gupta}. 

\subsection{Presentation of \newnested}

We start by reminding the necessary concepts of split Conformal Prediction (CP) in the complete case, without missing values, before diving into the details of our unified framework \newnested.

\subsubsection{Background on split CP} 

Introduced in \citet{papadopoulos_inductive_2002,vovk_algorithmic_2005,lei_distribution-free_2018}, split CP builds predictive regions by first splitting the $n$ points of the training set into two disjoint sets $\rm{Tr}, \rm{Cal}\subset \llbracket 1,n \rrbracket$, to create a \textit{proper training set}, $\rm{Tr}$, and a \textit{calibration set}, $\rm{Cal}$, of sizes $\#\Tr = (1-\rho)n$ and $\#\Cal=\rho n$ with $\rho\in]0,1]$. On the proper training set, a model $\hat f$ (chosen by the user) is fitted, and then used to predict on the calibration set. \textit{Conformity scores} $\mathcal{S} =  \left\{\left(s\left(X^{(k)},Y^{(k)};\hat{f}\right)\right)_{k \in \rm{Cal}}\right\} \cup \{ +\infty \}$ are computed to assess how well the fitted model $\hat f$ predicts the response values of the calibration points. In regression, usually the absolute value of the residuals is used, i.e.~$s(x,y;\hat{\mu}) = |\hat\mu(x) - y|$. In classification, the simplest score is $s(x,y;\hat{p}) = 1 - \hat{p}(x)_y$ (where $\hat{p}: \mathcal{X} \mapsto [0,1]^{\mathcal{Y}}$ outputs a vector of estimated probabilities for each class). Finally, the $(1-\alpha)$-th quantile  of these scores $q_{1-\alpha}\left( \mathcal{S} \right)$ (i.e., their $\lceil \left(1-\alpha\right)\left( \#\Cal+1 \right)\rceil$ smallest value) is computed to define the predictive region: 
$\widehat{C}_{n,\alpha} ( x ) := \{y \in \mathcal{Y} \text{ such that } s(x,y;\hat{f}) \leq q_{1-\alpha}\left( \mathcal{S} \right)  \}$. In regression with the absolute values of the residual score, this reduces to $\widehat{C}_{n,\alpha} ( x ) := \left[ \hat{\mu}(x) \pm q_{1-\alpha}\left( \mathcal{S} \right) \right]$.

This procedure satisfies \Cref{eq:pi_mv} for any $\hat{f}$, any (finite) sample size $n$, as long as the data points are exchangeable.\footnote{Only the calibration and test data points need to be exchangeable.} For more details on split CP, we refer to \citet{angelopoulos-gentle,vovk_algorithmic_2005}, as well as to \citet{manokhin_valery_2022_6467205}.

\subsubsection{\newnested}

From an high level perspective, the idea is to apply split CP on top of an impute-then-predict pipeline (of imputation function $\Phi$), and to modify the calibration step in order to ensure MCV. This is called CP-MDA, for \emph{conformal prediction with missing data augmentation}. Generally, for a given test point $\left(X^{(n+1)},M^{(n+1)}\right)$, CP-MDA works by artificially masking covariates in the calibration set so as to match \emph{at least} the mask of the test point, by creating a new mask $\widetilde{M}^{(k)} = \max\left(M^{(k)},M^{(n+1)}\right)$ for each $k\in\Cal$. In other words, it corresponds to discarding from the calibration set the covariates that are missing in the test point. This leads to $M^{(n+1)}\subseteq\widetilde{M}^{(k)}$, i.e., all over-masked (or \emph{augmented}) points $\left(X^{(k)},\widetilde{M}^{(k)},Y^{(k)}\right)_{k\in\Cal}$ have at least the missing entries of $\left(X^{(n+1)},M^{(n+1)}\right)$. The points such that $\widetilde{M}^{(k)} = M^{(n+1)}$ can be used directly as under distributional assumptions ($\p\mcarymx\iid$), they now have the same mask and distribution as the test point. Yet for many calibration points it remains that $\widetilde{M}^{(k)} \neq M^{(n+1)}$ (precisely, for all the $k\in\Cal$ such that $M^{(k)} \not\subseteq M^{(n+1)}$). This means that those over-masked points follow another conditional distribution than the one of the test point, and MCV can not be directly ensured. 

An idea is to subsample the calibration set so that the effective calibration set contains only $k\in\Cal$ such that $M^{(k)} \not\subseteq M^{(n+1)}$ (i.e., $\widetilde{M}^{(k)} = M^{(n+1)}$) \citep[this is the approach followed in \masksub,][]{pmlr-v202-zaffran23a}. However, this can lead to overly small calibration set size in high dimension, resulting in a large variance (on the coverage level and thus set size). Therefore, two questions naturally arise: 
\begin{itemize}[itemsep=1pt,topsep=1pt]
    \item How to build the calibration set? 
    \item How to leverage the test point so as to account for the different distributions present in the over-masked calibration set---and with many of them not matching the test mask conditional distribution---when constructing the predictive set?
\end{itemize}
The answers we suggest define our generalized framework \newnested, whose pseudo-code is available in \Cref{alg:mda_nested}, and are illustrated in \Cref{fig:mda_scheme}. 

\paragraph{Construction of the calibration set.}

\newnested includes a subsampling step: it calibrates on the set of indices $\widetilde{\Cal} \subseteq \Cal$ provided by the user, where $\widetilde{\Cal}$ can be obtained with any subsampling strategy, that might even be stochastic, as long as the randomness is independent of the covariates and outputs, $\left( X^{(k)},Y^{(k)} \right)_{k\in\Cal\cup\{n+1\}}$ (it can still depend on the masks). The following strategies work if the data distribution belongs to $\p\mcarymx\iid$ (which is an assumption we make anyway when using \newnested since, as we show precisely in \Cref{thm:mcv-nested}, \newnested is typically MCV-$\p\mcarymx\iid$):  
\begin{enumerate}[label=\roman*),itemsep=1pt,topsep=1pt,leftmargin=20pt]
\item subsampling only the indices $\left\{ k\in\Cal: M^{(k)} \subseteq M^{(n+1)} \right\} := \widetilde{\Cal}$ \citep[this is the strategy of \masksub,][]{pmlr-v202-zaffran23a}; 
\item no subsampling, $\widetilde{\Cal} := \Cal$ \citep[this is the path taken by \mask,][]{pmlr-v202-zaffran23a}; 
\item subsampling only the indices $\left\{ k\in\Cal: M^{(k)} \subseteq m' \right\} := \widetilde{\Cal}$, for some $m' \supseteq M^{(n+1)}$;
\item obtain $\widetilde{\Cal}$ by subsampling from the indices $\left\{ k\in\Cal: M^{(k)} \subseteq m' \right\}$, for some $m' \supseteq M^{(n+1)}$, using a mixture distribution, whose weights only depend on $\left(M^{(k)}\right)_{k\in\Cal\cup\{n+1\}}$.
\end{enumerate}
Then, for any $k\in\widetilde{\Cal}$, the over-mask is constructed, defining $\widetilde{M}^{(k)} = \max\left(M^{(k)},M^{(n+1)}\right)$. This is schematized in \Cref{fig:mda_scheme}.

\paragraph{Leveraging temporary test points.} After the subsampling step aforedmentioned, the over-masked calibration points and the test point do not necessarily have the same conditional distribution conditionally to the mask, as $M^{(n+1)}\subseteq \widetilde{M}^{(k)}$ without equality in general. In order to match those distributions, an idea is to create \textcolor{blindgreen}{temporary test points} (one for each calibration point) and to apply $\widetilde{M}^{(k)}$ to it. This is illustrated in \textcolor{blindgreen}{green} in \Cref{fig:mda_scheme}. \newnested evaluates the number of over-masked calibration points that have a conformity score smaller than that of the \textcolor{blindgreen}{\emph{corresponding over-masked test point}} for a given $y \in \mathcal{Y}$. Then, the predictive set includes only the $y\in\mathcal{Y}$ such that this number is small enough (a threshold that depends on $\alpha$ and the effective calibration size). This careful treatment of the test point allows to compare scores obtained from identical distributions conditionally on their associated mask. 

\begin{figure}[!b]
    \centering
    \includegraphics[width=\textwidth]{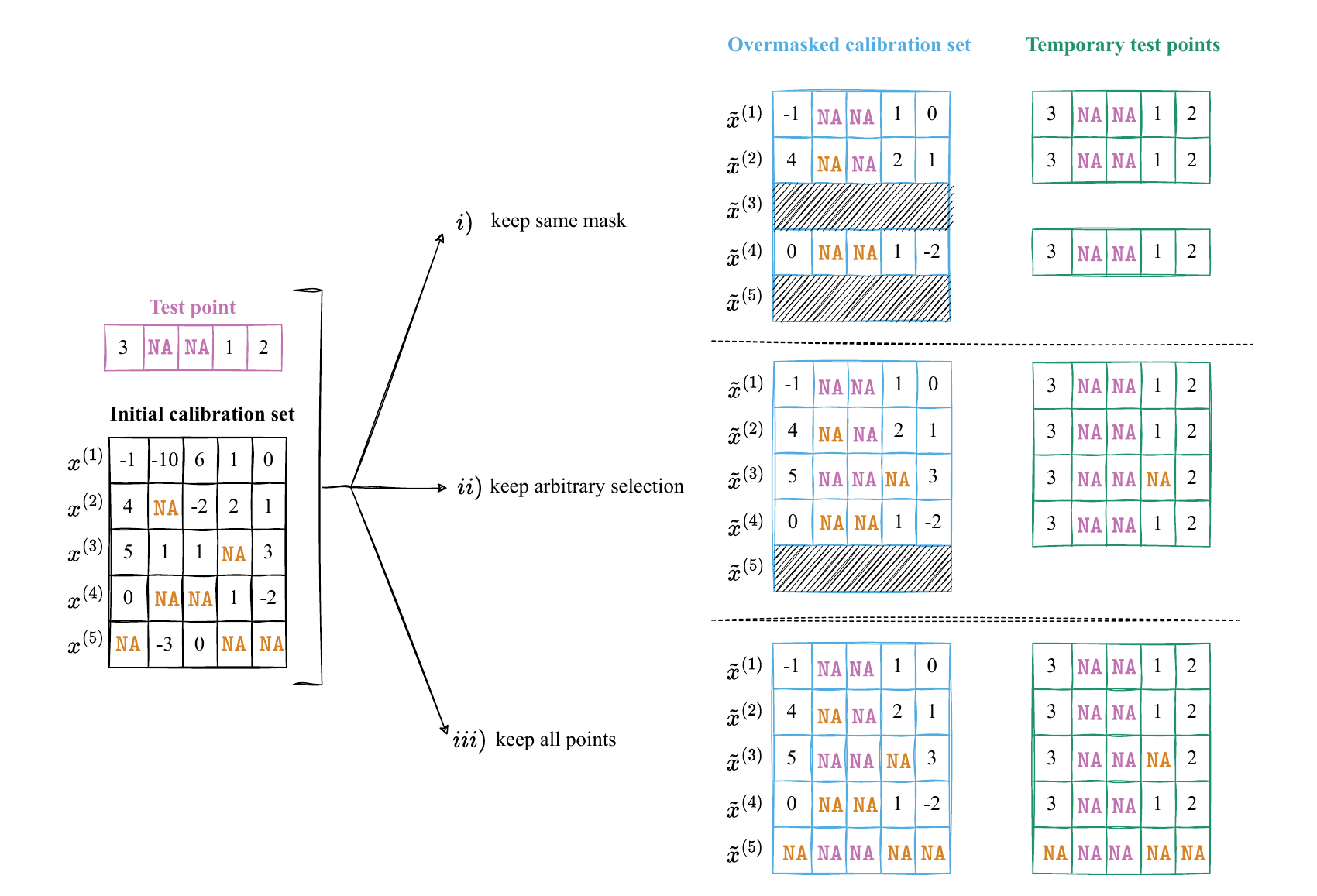}    
    \caption{\newnested illustration. Different subsampling strategies are shown, with their associated \textcolor{blindblue}{over-masked calibration set} and \textcolor{blindgreen}{temporary test points} according to one \textcolor{blindpurple}{test point}.}
    \label{fig:mda_scheme}
\end{figure}

\begin{algorithm}[!htb]
\caption{\newnested}
\label{alg:mda_nested}
\begin{algorithmic}[1] 
\REQUIRE Training set $\left\{ \left(X^{(k)},M^{(k)},Y^{(k)}\right)\right\}_{k=1}^n$, imputation algorithm $\mathcal{I}$, learning algorithm $\mathcal{A}$ taking its values in $\mathcal{F} := \mathcal{Y}^{\mathcal{X}\times\mathcal{M}}$, calibration proportion $\rho\in]0,1]$, $\left\{\Tr, \Cal, \Phi, \hat{A}\right\}$ the output of the splitting \Cref{alg:split_train} ran on $\left\{ \left\{ \left(X^{(k)},M^{(k)},Y^{(k)}\right)\right\}_{k=1}^n, \mathcal{I}, \mathcal{A}, \rho\right\}$, conformity score function $s\left(\cdot,\cdot;f\right)$ for $f \in \mathcal{F}$, significance level $\alpha$,  test point $\left( \textcolor{blindpurple}{ X^{({n+1})},M^{(n+1)} } \right)$, subsampled set of calibration indices $\widetilde{\Cal} \subseteq \Cal$
\ENSURE Prediction set ${\widehat{C}}_{n,\alpha}^{\text{MDA-Nested}^{\star}}\left( X^{(n+1)},M^{(n+1)} \right)$
\STATEx \textcolor{blindblue}{// Generate an over-masked calibration set:}  
\FOR {$k \in \widetilde{\Cal}$}
\COMMENT {Additional nested masking}
\STATE $\widetilde M^{(k)} = \max( M^{(k)}, \textcolor{blindpurple}{M^{(n+1)}} )$ 
\ENDFOR \COMMENT{\textcolor{blindblue}{\;\; Over-masked calibration set generated. //} }
\STATE $\begin{aligned}[t]& {\widehat{C}}_{n,\alpha}^{\text{MDA-Nested}^{\star}} \left( \textcolor{blindpurple}{ X^{(n+1)},M^{(n+1)}} \right) := \Bigl\{ y\in\mathcal{Y}: (1-\alpha)(1+\#\widetilde{\Cal}) > \\
& \sum_{k\in\Cal} \mathds{1} \left\{ s\left( \left(X^{(k)}, \widetilde{M}^{(k)}\right), Y^{(k)};\hat{A}\left(\Phi\left(\cdot,\cdot\right),\cdot\right)\right) \left. \;\; < s\left( \textcolor{blindgreen}{\left(X^{(n+1)}, \widetilde{M}^{(k)}\right)}, y;\hat{A}\left(\Phi\left(\cdot,\cdot\right),\cdot\right)\right) \right\} \right\}\end{aligned}$
\end{algorithmic}
\end{algorithm}

\begin{algorithm}[!htb]
\caption{Split and train}
\label{alg:split_train}
\begin{algorithmic}[1] 
\REQUIRE Imputation algorithm $\mathcal{I}$, learning algorithm $\mathcal{A}$ taking its values in $\mathcal{F} := \mathcal{Y}^{\mathcal{X}\times\mathcal{M}}$, training set $\left\{ \left(X^{(k)},M^{(k)},Y^{(k)}\right)\right\}_{k=1}^n$, calibration proportion $\rho\in]0,1]$
\ENSURE Splitted sets of indices $\Tr$ and $\Cal$, imputation function $\Phi$, fitted predictor $\hat{A}$
\STATE Randomly split $\{1, \ldots, n\}$ into 2 disjoint sets $\Tr$ \& $\Cal$ of sizes $\#\Tr = (1-\rho)n$ and $\#\Cal=\rho n$
\STATE Fit the imputation function: ${\Phi(\cdot,\cdot) \leftarrow \mathcal{I}\left(\left\{ \left( X^{(k)}, M^{(k)} \right), k \in \rm{Tr}\right\}\right)}$ 
\STATE Fit the learning algorithm $\mathcal{A}$: ${\hat{A}\left( \cdot, \cdot \right) \leftarrow \mathcal{A}\left(\left\{ \left( \Phi\left(X^{(k)},M^{(k)}\right), M^{(k)} \right), k \in \rm{Tr}\right\}\right)}$ 
\end{algorithmic}
\end{algorithm}

\subsubsection{Key comments on \newnested}

In summary, \newnested bridges the gap between \masksub and \mask by proposing a tighter generalized framework. On the one hand, \masksub comes with a potentially small calibration set, thus with increased variability. On the other hand, by leveraging all calibration points, including those with very few observed covariates, the average interval length produced by \mask is typically larger than that of \masksub (cf.~\eqref{eq:CIL_isot_exp}). Furthermore, \mask is less generic than CP in the sense that it is specific to predictive \textit{intervals} (unlike \masksub which is as generic as CP and can be plugged with any score function, including classification). Overall, \newnested unifies this framework for any score function and provides high flexibility in the trade-offs between \emph{efficiency} and \emph{variability}:
\begin{itemize}[itemsep=1pt,topsep=1pt]
    \item At the extreme of no subsampling at all, we obtain a generalization of \texttt{CP-MDA-Nested} which encapsulates the classification setting;
    \item This generalization provides tighter sets than that of \texttt{CP-MDA-Nested} in the particular case of interval-based scores (see \Cref{rem:nested_in_newnested});
    \item At the other extreme of the strictest subsampling procedure, we retrieve \texttt{CP-MDA-Exact};
    \item Any other less restrictive subsampling (possibly with a random selection between various augmented mask) belongs to this framework, providing more flexibility in the trade-offs between exact validity and statistical variability.
\end{itemize}
This overview is summarized in \Cref{tab:nested_summary}.

\begin{table}[!b]
    \centering
\begin{center}
\resizebox{\linewidth}{!}{
\begin{tabular}{lccc}
\toprule
 Method    &  \masksub & \newnested \textbf{(new)}  & \mask \\
 \midrule
 Size of actual calibration set & $\#$ points in $\Cal$ with $M\subseteq M^{(n+1)}$ &  Any  & $\# \Cal$\\
 Mask of the points used for calibration & exactly $M^{(n+1)}$ &  &  all, leading to $\widetilde{M}$ s.t. $M^{(n+1)} \subseteq \widetilde{M}$\\
 Overall behavior  & \textcolor{mLightBrown}{Too few Cal points $\to$ high coverage variance} & Flexible & \textcolor{mLightBrown}{Too large intervals (cf. \eqref{eq:CIL_isot_exp})}\\
 \midrule
 Applies to classification &\cmarkuq &\cmarkuq (new) & \xmarkuq \\
 Outputs non-interval sets &\cmarkuq &\cmarkuq (new)& \xmarkuq \\
 \midrule
 Marginal  guarantee (MV)  &\cmarkuq &\cmarkuq (new) & \cmarkuq (new)\\
 Conditional guarantee (MCV)  &\cmarkuq  & \cmarkuq (new)&\cmarkuq (new) \\
\bottomrule
\end{tabular}
}
\end{center}
    \caption{Summary of the high-level characteristics of MDA algorithms, coming from the literature, as well as our novel contributions indicated by ``(new)''. Characteristics are given for a test point $\left(X^{(n+1)}, Y^{(n+1)}, M^{(n+1)}\right)$. Details regarding guarantees are given in \Cref{tab:nested_guarantees}.}
    \label{tab:nested_summary}
\end{table}

In the case where the nested predictive sets are intervals and $\widetilde{\Cal} = \Cal$, then the final predictive sets obtained through \newnested are included in the ones of \mask.

\begin{remark}
When $\widetilde{\Cal} = \Cal$, and using absolute value of the residuals scores or conformalized quantile regression scores \citep{romano_conformalized_2019}, or any score such that $\{y \in \mathcal{Y} \text{ such that } s(x,y;\hat{f}) \leq b \}$ for some $b$ is an interval, then ${\widehat{C}}_{n,\alpha}^{\text{MDA-Nested}^{\star}}(\cdot) \subseteq {\widehat{C}}_{n,\alpha}^{\text{MDA-Nested}}(\cdot)$ (see \Cref{app:nested}).
\label[remark]{rem:nested_in_newnested}
\end{remark}

This unification allows us to provide theoretical guarantees, stated in \Cref{sec:nested_guar}, leveraging the deep connections between \newnested and leave-one-out conformal methods \citep[such as][]{barber2021jackknife,gupta}. Indeed, the rationale for predicting on masked test points, using the augmented calibration masked, is that we want to treat the test and calibration points in a symmetric way. We summarize them in the following \Cref{tab:nested_guarantees}.
\begin{table}
    \begin{center}
  \resizebox{\linewidth}{!}{
  \begin{tabular}{lcc}
    \toprule
    Guarantees &  MV & MCV \\ 
    \midrule
    \texttt{CP-MDA-Exact} & $\p\iid\mcarymx$, level $\alpha$, & $\p\iid\mcarymx$, level $\alpha$, \\
    i.e., \newnested with subsampling & with upper bound, & with upper bound, \\
    only $k\in\Cal$ s.t. $M^{(k)} \subseteq M^{(n+1)}$& from \citet{pmlr-v202-zaffran23a} & from \citet{pmlr-v202-zaffran23a} \\
    \midrule
    \newnested & $\p\iid\mcarymx$, level $2\alpha$ & $\p\iid\mcarymx$, level $2\alpha$ \\
    \midrule
    \newnested without subsampling & $\p\exch$, level $2\alpha$ & $\p\iid\mcarymx$, level $2\alpha$ \\
    \bottomrule
    \end{tabular}
  }
    \end{center}
    \caption{Theoretical guarantees of \newnested depending on the subsampling choice.}
    \label{tab:nested_guarantees}
\end{table}

\subsection{Theoretical guarantees on CP-MDA-Nested and \newnested leveraging their connection to leave-one-out CP}
\label{sec:nested_guar}
Hereafter, we give our theoretical results on the coverage of our \newnested algorithm.

\begin{theorem}[Marginal validity of \newnested]
\newnested with $\widetilde{\Cal} = \Cal$ (and thus \mask) is MV-$\p\exch$ at the level $1-2\alpha$.
\label{thm:mv-nested}
\end{theorem}

\Cref{thm:mv-nested} provides a lower bound on \newnested and \mask coverage as $1-2\alpha$.  This result is important as it equips \newnested with $\widetilde{\Cal} = \Cal$ and \mask with controlled coverage on any exchangeable distribution: they are marginally valid even on MNAR distributions or when \ynotindmx. It means that despite modifying the data set independently from $X$ and $Y$ and breaking the structure of $(X,M,Y)$, the obtained estimator makes reliable predictions including when $X,M$, and $Y$ are strongly dependent. This originates from the fact that the whole data set has been treated equally, including the test~point.  

\begin{theorem}[Conditional validity of \newnested]
\newnested with $\widetilde{\Cal}$ independent of the data set $\left( X^{(k)},Y^{(k)}\right)_{k\in\Cal\cup\{n+1\}}$ (and thus \mask) is MCV-$\p\mcarymx\iid$ at the level $1-2\alpha$.
\label{thm:mcv-nested}
\end{theorem}

The proofs of \Cref{thm:mv-nested,thm:mcv-nested} are deferred to \Cref{app:nested_mv} and \Cref{app:nested_mcv} respectively. They are heavily based on the deep connections between \newnested with Jackknife+ and general leave-one-out or $k$-fold CP \citep{barber2021jackknife,Vovk2013,gupta}. Indeed,  one can observe that, for each $k\in\Cal$, we evaluate the conformity score of the test point $(X^{(n+1)},M^{(n+1)},Y^{(n+1)})$ using the $k$-th augmented mask, as the equivalent of evaluating the conformity score of the test point with the fitted model that has left-out the $k$-th calibration point. This connection between \newnested and leave-one-out conformal approaches directly stems from the same core motivations: $i)$ both enforce a design that use all the observations of the training or calibration sets to handle small sample sizes, $ii)$ both need to avoid invalid designs that arise naturally when keeping all these points, such as comparing scores obtained with different predictors. 

\paragraph{On the factor 2 and link with empirical quantile aggregation.} Despite the coverage guarantee being of $1-2\alpha$ instead of the desired $1-\alpha$, in practice, our experiments in \Cref{sec:exp} show that \newnested without subsampling (i.e., \mask) tends to over-cover. This aligns with Figure 2 of \citet{barber2021jackknife}, that illustrates the fact that leave-one-out conformal methods achieve empirically exactly $1-\alpha$ coverage, while $K$-fold conformal approaches over-cover. The reason behind this phenomenon is still unclear in the community, and is likely to be the same than the reason for \newnested over-coverage, as one can see \newnested as having access to many folds of calibration points, since each augmented calibration mask typically appears many times in the calibration set. 
In particular, \citet{pmlr-v202-zaffran23a} provide MCV-$\p\iid\mcarymx$ guarantees at the level $1-\alpha$ on a modified version of \mask which leverages this folding point of view by calibrating only on one (arbitrarily) chosen such fold. Similarly than for $K$-fold and leave-one-out conformal methods, we can look at \newnested as a way to aggregate many valid empirical quantiles or $p$-values, one for each fold, i.e., one for each augmented mask. Due to the strong dependencies between these random variables, such an aggregation does not lead to a valid aggregated quantile or $p$-value, and induces a loss of coverage. 

\paragraph{\Cref{thm:mcv-nested} proof approach: coupling our algorithm with a leave-one-out conformal method on a virtual complete data set.} We work conditionally to the mask of the test point, $M^{(n+1)}$. Then, we introduce a randomized predictor, for which ``training'' consists in randomly picking one individual predictor among a bag of individual predictors, each of them corresponding to an augmented calibration mask. This bag contains exactly $2^{|\obs(M^{(n+1)})|}$ possible individual predictors, where $|\obs(M^{(n+1)})|$ is the number of 1s in $M^{(n+1)}$, i.e., the number of observed features in the test point. Each individual predictor in the bag is thus parametrized by a \textit{super/over-mask} of $M^{(n+1)}$. We call such a predictor a mixture-predictor, as it basically consists in picking randomly one individual predictor in a mixture of individual predictors. That sampling has to be made independently of the data the mixture predictor is applied to, but non necessarily uniformly.
Furthermore, we ensure that the individual predictor indexed by a mask $M$ only relies on the covariates $X_{\obs(M)}$ for the prediction, in order for this algorithm to be applicable in practice (e.g., an invalid design would be individual predictors that require the knowledge of some of the $X_{\mis(M)}$, unobserved in practice, in order to make predictions).

We then show that our algorithm \newnested, applied to the data set with missing entries $\left(X^{(k)}_{\obs\left(M^{(k)}\right)}, Y^{(k)}, M^{(k)}\right)_{k=1}^{n+1}$, has the same guarantees in expectation as the leave-one-out conformal that uses the mixture predictor, applied onto a virtual complete data set $\left(X^{(k)},Y^{(k)}\right)_{k=1}^{n+1}$, if we make some assumptions on the missingness distribution. More specifically, we show that \textit{there exists a coupling between the two algorithms}, that ensures that the output (and thus coverage) have the same distribution. This ultimately enables us to reuse existing guarantees for leave-one-out conformal estimators. 

\section{A practical glimpse on the impacts of breaking the distribution's assumptions}
\label{sec:exp}

In this concluding section, we investigate the numerical performances of \newnested mainly outside its theoretical set of assumptions. Experiments under $\p\mcarymx$ are provided in \Cref{sec:axp_mcar}, then \Cref{sec:exp_beyond_mcar} presents numerical results when the data distribution either belongs to $\p\mar$ or $\p\mnar$, and finally \Cref{sec:exp_ymx} reports empirical performances when \ynotindmx. 

In all experiments, the data are imputed using iterative regression (\texttt{iterative ridge} implemented in Scikit-learn, \citet{scikit-learn}). The predictive models are fitted on the imputed data concatenated with the mask. The prediction algorithm is a neural network, fitted to minimize the pinball loss \citep{chr}. For the vanilla QR, we use both the training and calibration sets for training. The training set contains 500 data points, and the calibration set 250 data points. To evaluate the marginal coverage, a test set is generated with missing values following the same distribution as on the training and calibration sets. Then, to estimate mask-conditional coverage (i.e., $\mathds{P} (Y \in \widehat C_{n,\alpha}(X,m) | M = m )$ for each $m \in \mathcal{M}$), we generate another test set by imposing that the number of observations per pattern is fixed to 100, in order to ensure that the variability is not impacted by $\mathds{P}\left( M = m \right)$. Each experiment is repeated 100 times (unless stated otherwise). 

\subsection{Experiments under $\p\mcarymx$}
\label{sec:axp_mcar}

\textbf{Data generation.}  The data is generated with $d=10$ according to \Cref{mod:glm} (regression), $Y = \beta^T X + \varepsilon$ with $X \sim \mathcal{N}\left(\mu, \Sigma \right)$, $\mu = (1,\cdots,1)^T$ and $\Sigma = \varphi (1,\cdots,1)^T(1,\cdots,1)+(1-\varphi)I_d$, $\varphi\in\{0,0.8\}$ depending on the experiment, Gaussian noise $\varepsilon \sim \mathcal{N}(0,1) \ind (X,M)$ and the following regression coefficients $\beta = (1, 2, -1, 3, -0.5, -1, 0.3, 1.7, 0.4, -0.3)^T$. Each of these 10 features is missing with probability $0.2$ independently from anything else. 

\subsubsection{\newnested provides flexibility} 

In our first experiments, we compare CQR to \masksub and \mask, as well as \newnested where we subsample all the calibration points that have at most two features that are missing among the observed features of the test point. As $d=10$, there are 1024 different patterns, avoiding to display the performances of the algorithms on each of the patterns. Therefore, instead, we represent the coverage and the length of the predictive intervals depending on the mask size, a proxy for mask-conditional coverage. For each pattern size, 100 observations are drawn according to the distribution of $M | \text{size}(M)$ in the test set. In this subsection only, the number of repetition is of 50.

\begin{figure}[!b]

    \centering
    \begin{subfigure}{\textwidth}
        \centerline{\includegraphics[width=\textwidth]{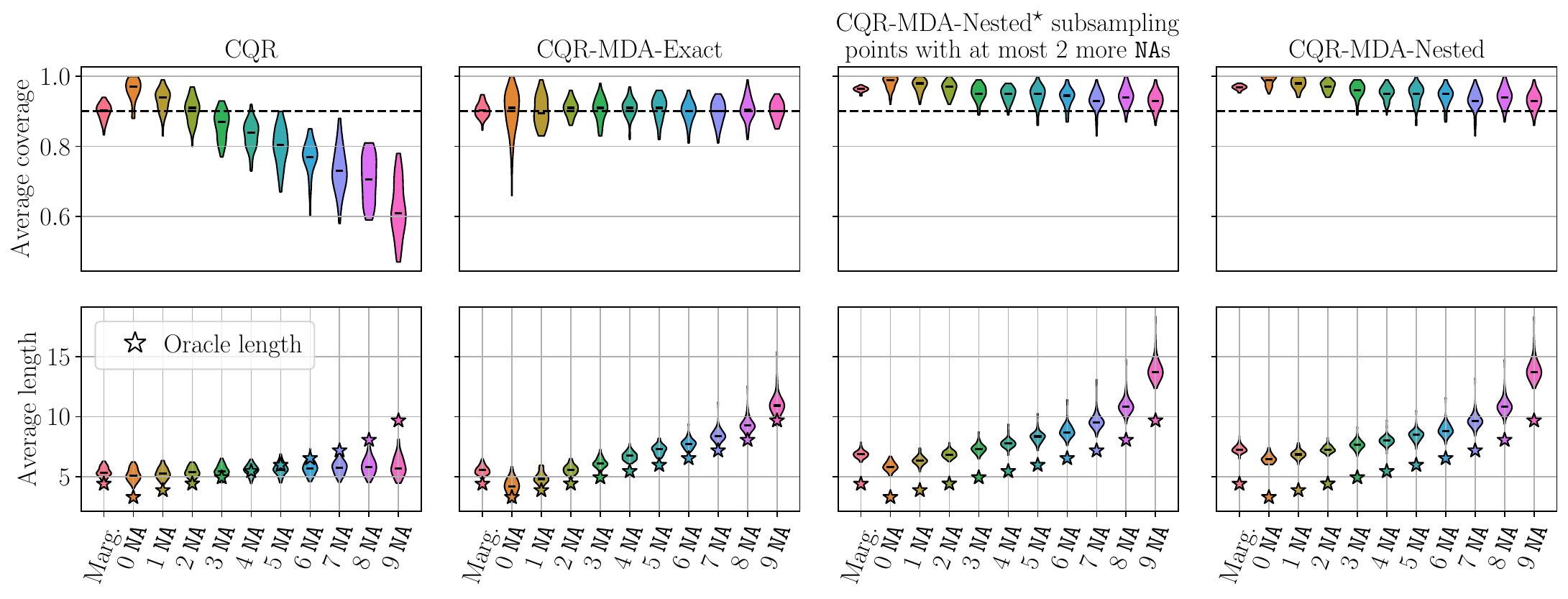}}
        \caption{Each features is missing with probability 0.2.}
        \label{fig:mcar_newnested}
    \end{subfigure}
    \hfill
    \begin{subfigure}{\textwidth}
        \centerline{\includegraphics[width=\textwidth]{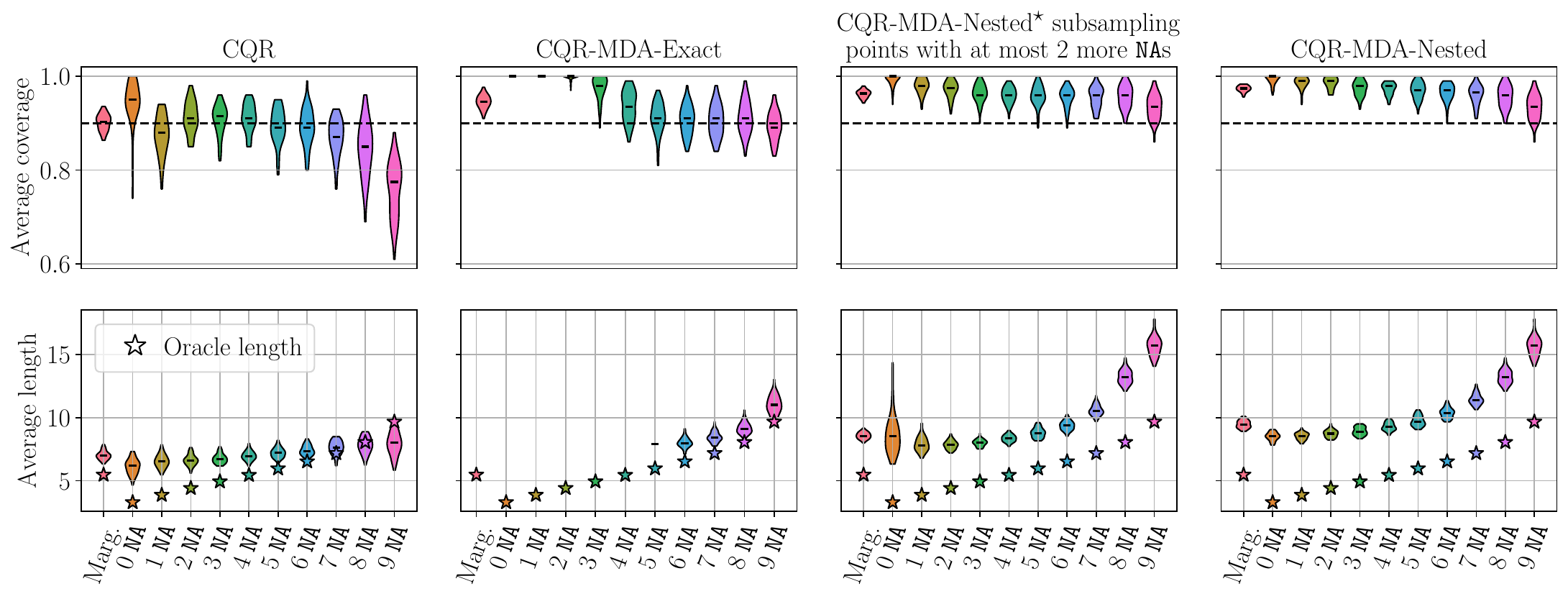}}
        \caption{Each features is missing with probability 0.4.}
        \label{fig:mcar_newnested_40}
    \end{subfigure}
    \caption{Validity and efficiency with \textbf{MCAR missing values} on dependent Gaussian features, with $\varphi = 0.8$, and such that \yindmxbold. Average coverage (top) and length (bottom) as a function of the missing pattern sizes. The black horizontal line in each violin plot corresponds to the median. The first violin plot shows the marginal coverage. The marginal test set includes 2000 observations. The mask-conditional test set includes 100 individuals for each missing data pattern size.}
    \label{fig:mcar_newnested_both}
    \end{figure}

\Cref{fig:mcar_newnested} displays the results of this experiment. As noticed in \citet{pmlr-v202-zaffran23a}, CQR is not MCV-$\p\mcarymx\iid$ as its intervals undercover or overcover depending on the number of missing values. The three versions of \newnested ensure that the coverage is at least $1-\alpha$ for any pattern size, as supported by our theory (\Cref{sec:nested_guar})\footnote{Note that MCV implies validity on any mask size, but not the contrary.} Comparing \masksub and \mask, we observe that \masksub is more efficient as it produces smaller intervals and its coverage is exactly of $1-\alpha$ on average, while suffering for more variability than \mask. The intermediate version of \newnested allows to reduce \masksub variability while improving the efficiency of the intervals by 5.5\% marginally (the comparison consists in computing the difference between \newnested and \mask interals' median length, and normalize it by \mask intervals' median length), reaching nearly 10\% of improvement on the test points that have no missing values. For 7 to 9 missing values, this \newnested is equivalent to \mask as the subsampling scheme of \newnested boils down to keeping all the calibration points on these patterns.

\mask reveals all its interest over \masksub in settings where the exact subsampled calibration set contains really few points for some masks (e.g., in high dimension or when the probability of missing values is high). In \Cref{fig:mcar_newnested_40}, the probability of each features being missing is increased to 0.4. We observe that \masksub outputs infinite intervals more than half of the time on the marginal test, as well as on the test sets containing between 0 and 4 missing values. This is particularly unpractical. On the contrary, \mask produces finite length intervals on any test point, at the cost of being overly conservative. The improvements brought by \newnested with subsampling only the calibration points with at most 2 additional missing values are more stringent. In particular, the efficiency is improved by nearly 9.5\% marginally, and is in between 8.5\% and 10\% on test points that have between 1 and 6 missing values. 

\begin{figure}[!b]
\centerline{\includegraphics[width=0.5\textwidth]{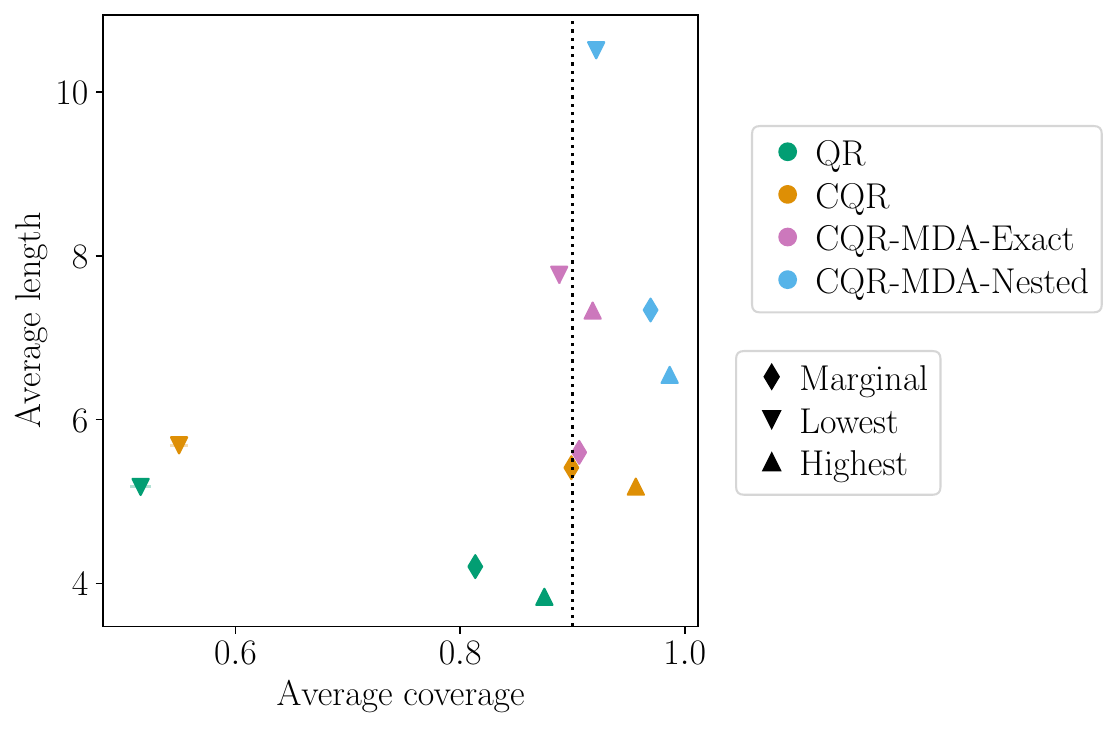}}
\caption{Validity and efficiency with \textbf{MCAR missing values} on dependent Gaussian features, with $\varphi = 0.8$, and such that \yindmxbold. Colors represent the methods. Diamonds ($\blacklozenge$) represent marginal coverage while the patterns giving the lowest and highest coverage are represented with triangles ($\blacktriangledown$ and $\blacktriangle$). Vertical dotted lines represent the target coverage of 90\%. Experimental details: $\#\Tr = 500$; $\#\Cal = 250$; the marginal test set includes 2000 observations; the mask-conditional test set includes 100 individuals for each missing data pattern.}
\label{fig:mcar}
\end{figure}

Note that this is only one example of \newnested for a given subsampling strategy, and that in practice the choice of strategy is highly dependent on the settings and could lead to even better performances. From now on, we restrict the subsequent experiments with \newnested to the two extremes---\masksub and \mask---as their main goal is to investigate the robustness beyond $\p\mcarymx$. For the same reason, we do not want to restrict ourselves to the mask-size conditional coverage, as it does not imply mask conditional coverage. Therefore, we use another visualization approach that was introduced in \citet{pmlr-v202-zaffran23a}. As an appetizer, \Cref{fig:mcar} presents the results under $\p\mcarymx\iid$ for QR, CQR, \masksub and \mask, using this visualization. The $x$-axis represents the average coverage and the average length is in the $y$-axis. The marker colors are associated to the different methods. A method is MCV if all the markers of its color are at the right of the vertical dotted line (90\%). The design of \Cref{fig:mcar}, and the following figures, requires a cautious interpretation. For each method we report, for the pattern having the highest (or lowest) coverage, its length and coverage. However, as this pattern may depend on the method, the length for the highest/lowest should not be directly compared between methods.

This \Cref{fig:mcar} illustrates that \newnested is MCV-$\p\iid\mcarymx$. Our hardness results of \Cref{sec:hardness} provide a new perspective on these results:
\begin{remark}
If \newnested was MCV on a broader class of distributions than $\p\iid\mcarymx$ for which a hardness result exists, then it would produce uninformative intervals on any distribution within this class, including $\p\iid\mcarymx$. Therefore, the fact that \newnested obtain finite length intervals in this experiment (\Cref{fig:mcar}) tends to confirm (with high probability) that the theory on the \newnested MCV can not be extended to $\p\iid\ymx$ or $\p\iid\mar$ nor $\p\iid\mcar$. This analysis is included in \Cref{tab:sum_res}, as a numerical confirmation on \newnested theory. 
\label[remark]{rem:nested_not_broader_mcv}
\end{remark}

\subsection{Beyond MCAR}
\label{sec:exp_beyond_mcar}

\paragraph{Beyond MCAR experiments.} To generate missing values under MAR or MNAR distribution, we select 6 variables (denote this set $X_\missing$) out of 10 that can be missing (the 4 others form the set $X_\observed$). Especially, $X_\missing = \{X_1,X_2,X_3,X_5,X_8,X_9\}$ in order to include different range of associated regression coefficients. We used the GitHub repository associated to \citet{muzellec2020} in order to introduce missing values in $X_\missing$ according to the following mechanisms, fixing the proportion of missing entries to be 20\%. For each of these following settings, we run two sets of experiments: one in which the correlation between the features is high ($\varphi = 0.8$) and therefore imputing through iterative regression allows to recover quite accurately the missing values, and one in which the features are independent ($\varphi = 0$) leading to an imputation that can not be better than the marginal expectation of the features.

$\blacktriangleright$ \underline{MAR experiments (\Cref{fig:mar}).} Missing values in $X_\missing$ are introduced under a MAR mechanism. To do so, a logistic model of arguments $X_\observed$ determines the probability of the variables in $X_\missing$ to be missing. This setting is declined 5 times, with different weights for the logistic model. Within each one, the experiments are repeated 100 times to assess for the variability. 

$\blacktriangleright$ \underline{MNAR self-masked (\Cref{fig:mnar_sm}).} Missing values in $X_\missing$ are introduced under a MNAR self masked mechanism. To do so, a logistic model determines the probability of each variable in $X_\missing$ to be missing by taking as input the exact same variable. This setting is declined 5 times, with different weights for the logistic model. Within each one, the experiments are repeated 100 times to assess for the variability.

$\blacktriangleright$ \underline{MNAR quantile censorship (\Cref{fig:mnar_q}).} Missing values in $X_\missing$ are introduced under a quantile censorship MNAR mechanism. In particular, missing values are introduced at random in each $q$-quantile of the variables in $X_\missing$. $q$ varies between 0.5, 0.75, 0.8, 0.85, 0.9 and 0.95 and this way we obtain 6 different settings. Within each one, the experiments are repeated 100 times to assess for the variability.

\begin{figure}

    \centering
    \begin{subfigure}{\textwidth}
        \centerline{\includegraphics[width=\textwidth]{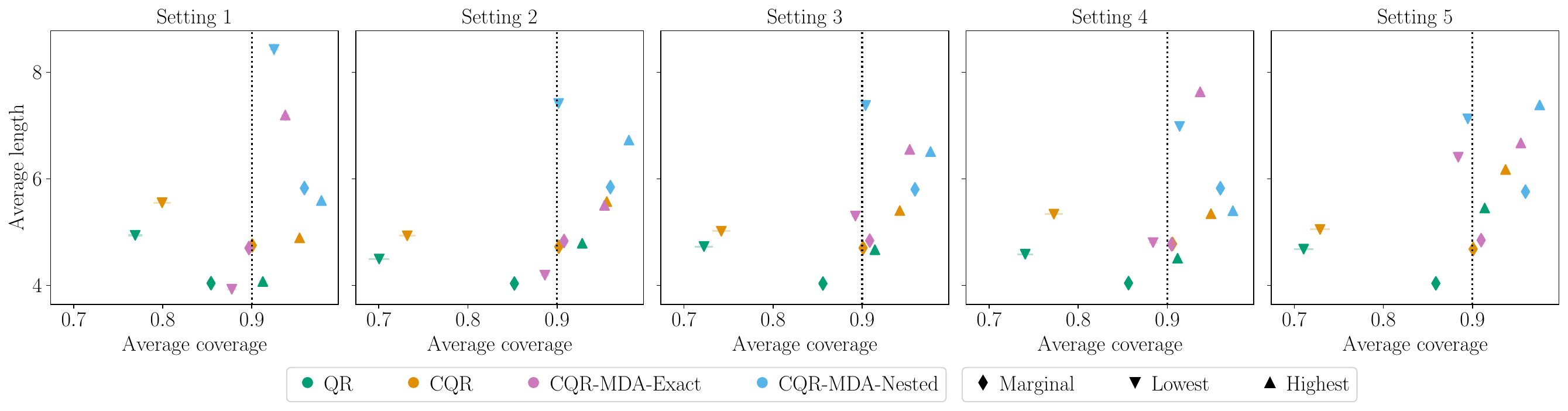}}
        \caption{Dependent Gaussian features, with $\varphi = 0.8$.}
        \label{fig:mar_dep}
    \end{subfigure}
    \hfill
    \begin{subfigure}{\textwidth}
        \centerline{\includegraphics[width=\textwidth]{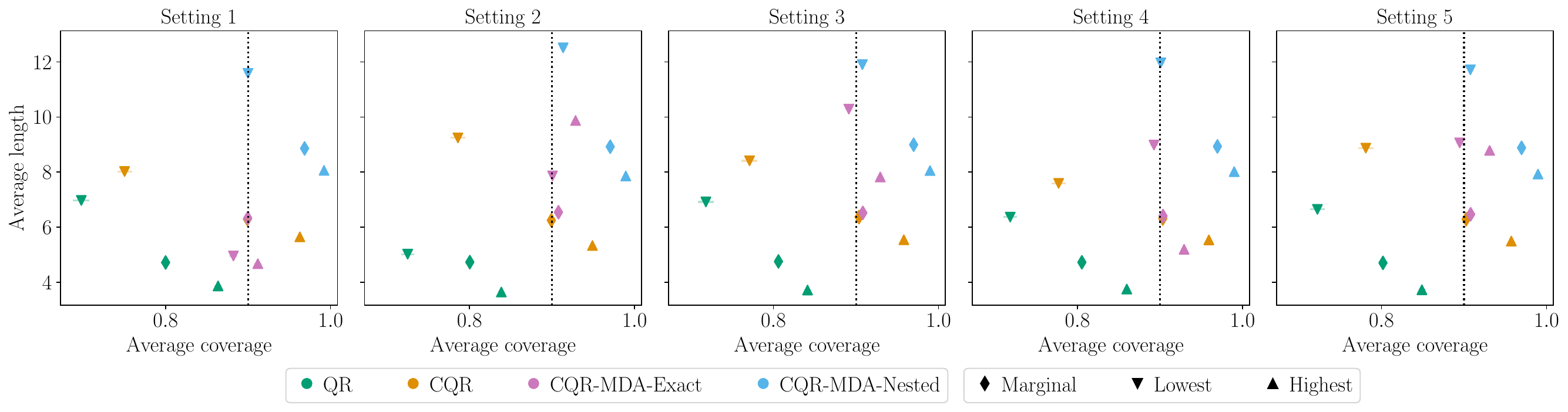}}
        \caption{Independent Gaussian features.}
        \label{fig:mar_no_dep}
    \end{subfigure}

\caption{Same caption than \Cref{fig:mcar}, for \textbf{MAR missing values}, each panel representing a different setting (set of parameters) for the missingness distribution.}
\label{fig:mar}
\end{figure}

\begin{figure}

    \centering
    \begin{subfigure}{\textwidth}
        \centerline{\includegraphics[width=\textwidth]{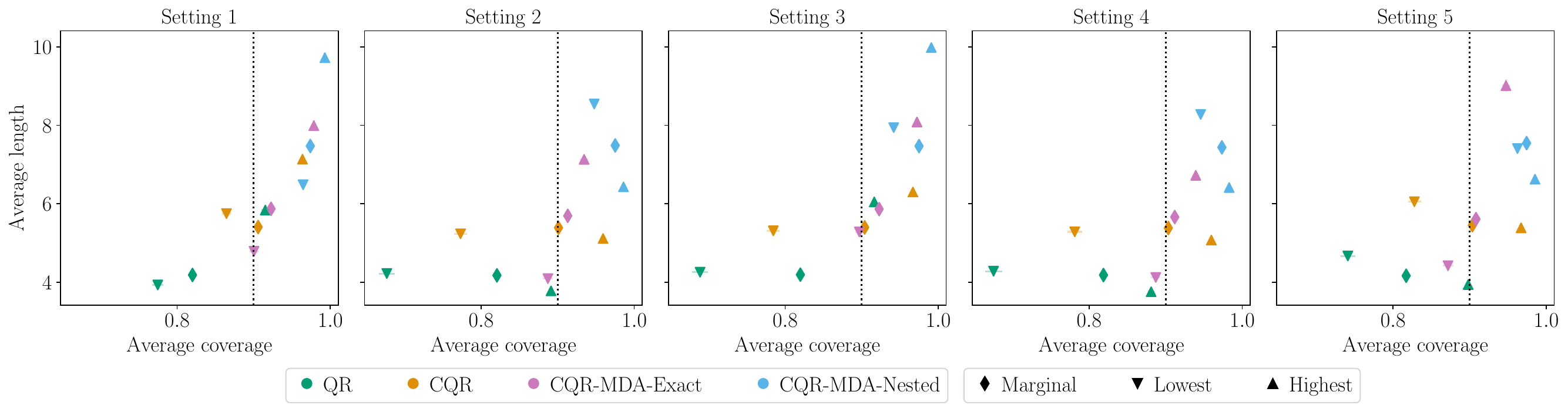}}
        \caption{Dependent Gaussian features, with $\varphi = 0.8$.}
        \label{fig:mnar_sm_dep}
    \end{subfigure}
    \hfill
    \begin{subfigure}{\textwidth}
        \centerline{\includegraphics[width=\textwidth]{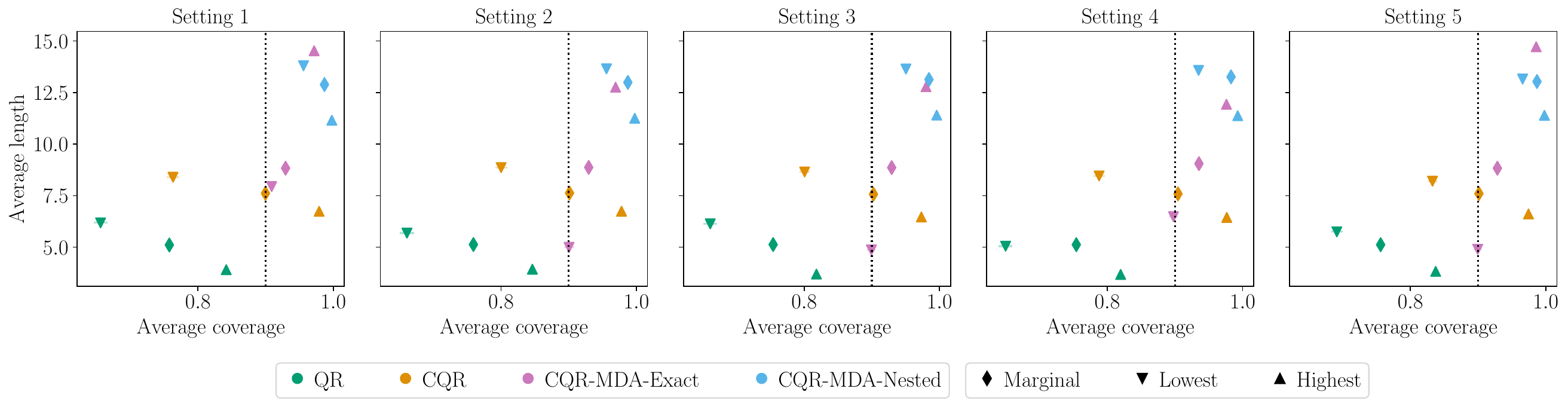}}
        \caption{Independent Gaussian features.}
        \label{fig:mnar_sm_no_dep}
    \end{subfigure}
    
\caption{Same caption than \Cref{fig:mar}, for \textbf{MNAR self masked missing values}.}
\label{fig:mnar_sm}
\end{figure}

\begin{figure}

    \centering
    \begin{subfigure}{\textwidth}
        \centerline{\includegraphics[width=0.8\textwidth]{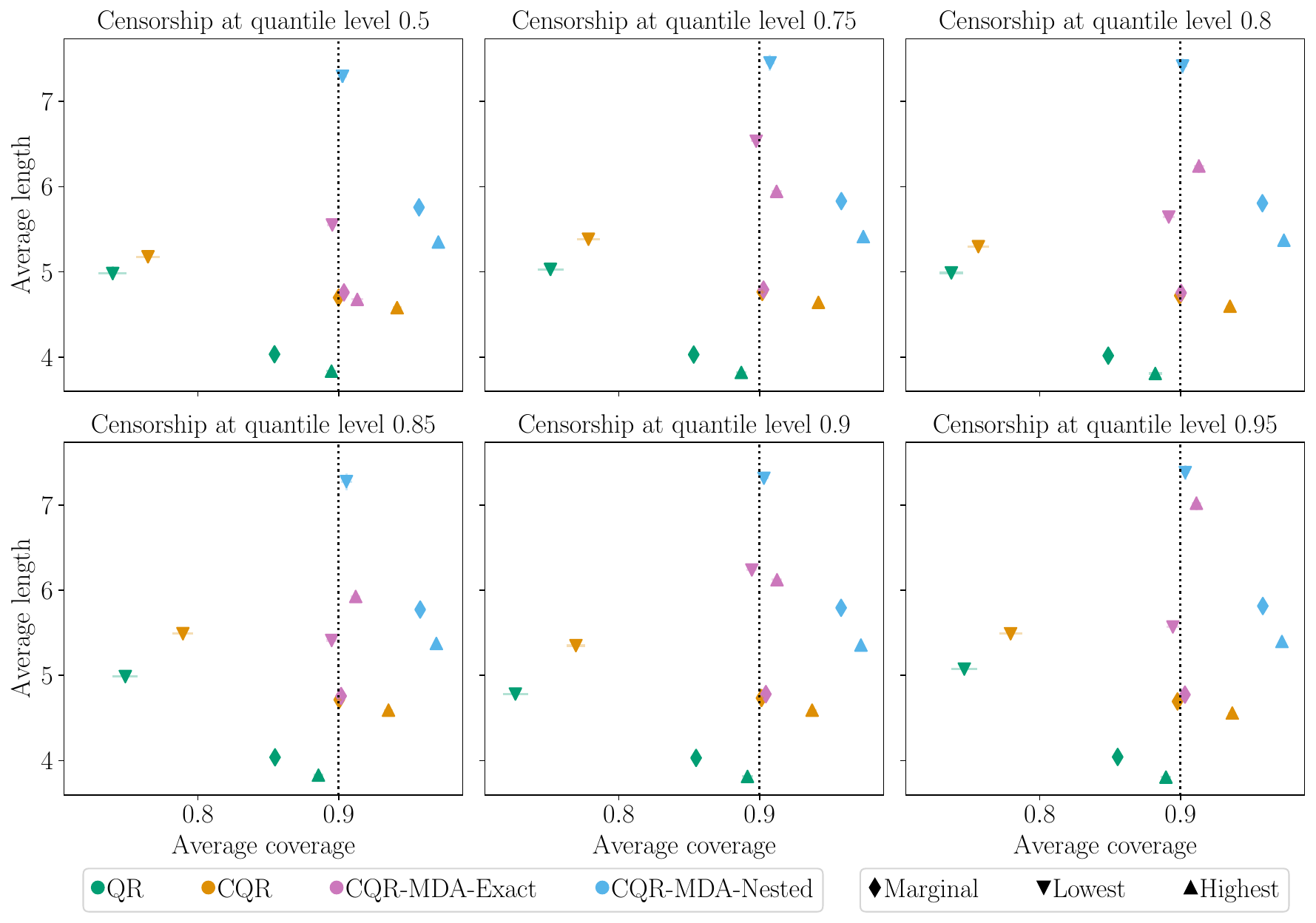}}
        \caption{Dependent Gaussian features, with $\varphi = 0.8$.}
        \label{fig:mnar_q_dep}
    \end{subfigure}
    \hfill
    \begin{subfigure}{\textwidth}
        \centerline{\includegraphics[width=0.8\textwidth]{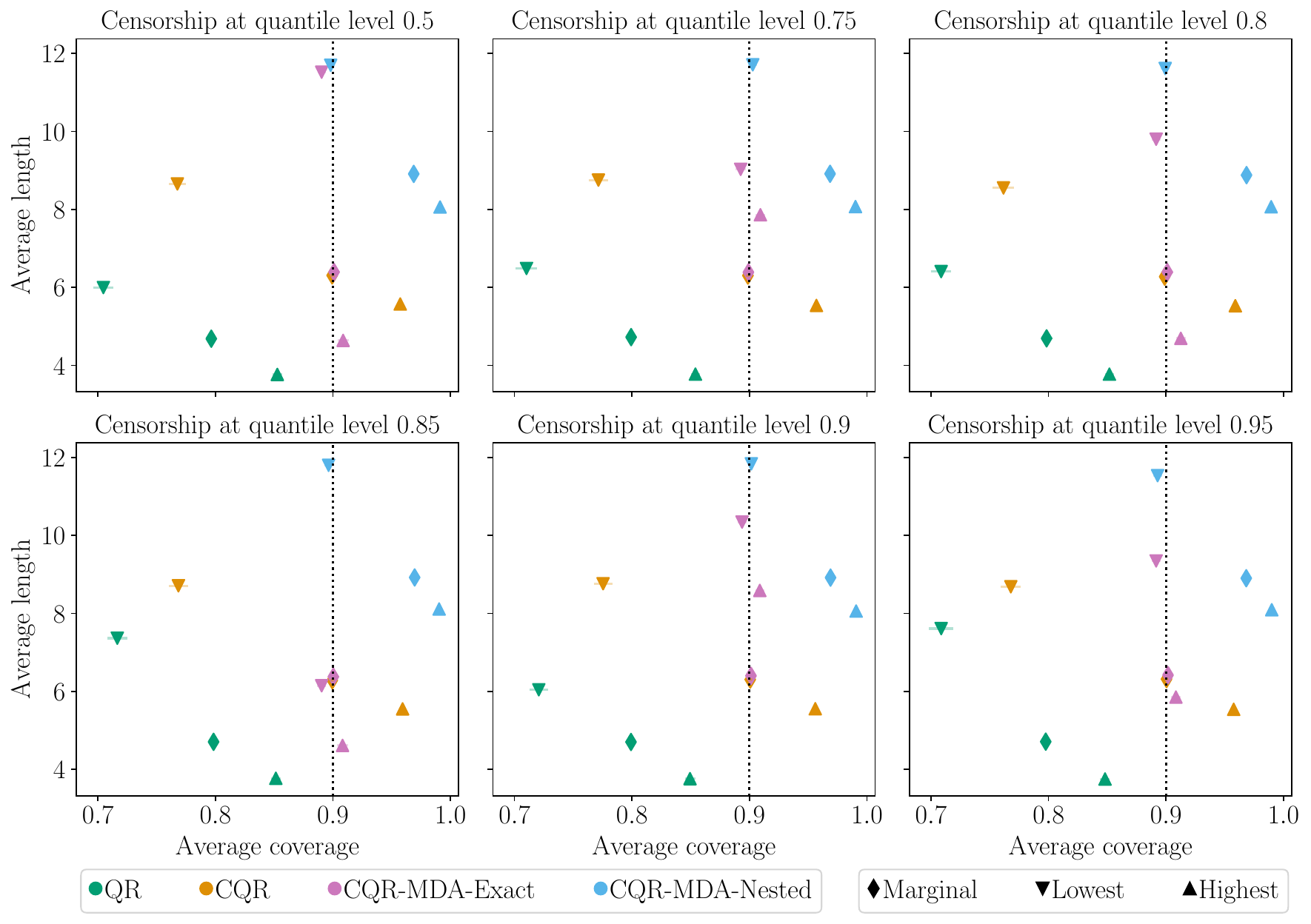}}
        \caption{Independent Gaussian features.}
        \label{fig:mnar_q_no_dep}
    \end{subfigure}
    
\caption{Same caption than \Cref{fig:mar}, for \textbf{MNAR quantile censorship missing values}.}
\label{fig:mnar_q}
\end{figure}

These experiments show that impute-then-CQR is marginally valid even under $\p\mar$ and $\p\mnar$. This is expected due to Proposition 3.3 of \citet{pmlr-v202-zaffran23a}, that demonstrates that vanilla impute-then-SplitCP is marginally valid for any missing mechanism as long as the initial data set is exchangeable. However, it is not MCV, which is also expected for the same reason that the fact that it is not MCV under $\p\mcarymx$. Most importantly, \newnested, through \masksub and \mask, is both marginally valid and MCV, despite the MCAR assumption not being satisfied, even when the imputation can not retrieve more information than the features expectation (i.e., when $\varphi=0$). This is a positive empirical result that hints robustness of \newnested on more complex relationships between $X$ and $M$ than independence.

\subsection{Breaking \yindmx Assumption}
\label{sec:exp_ymx}

Our last set experiments aim at breaking the \yindmx assumption. We focus on $d=3$ to be able to display all of the patterns and thus better illustrate the phenomenon. We generate data with $\varepsilon \sim \mathcal{N}(0, 1) \ind (X,M)$, $X \sim \mathcal{N}\left(\mu, \Sigma \right)$, $\mu = (1,1,1)^T$, $\Sigma = \varphi (1,1,1)^T(1,1,1)+(1-\varphi)I_d$, $\varphi\in\{0,0.8\}$ depending on the experiment, and $M_i \sim \mathcal{B}(0.2)$ for any $i\in\llbracket1,3\rrbracket$, independently from $X$ and $\varepsilon$. Finally: $Y = X_1\mathds{1}\left\{M_1=0\right\} + 2X_1\mathds{1}\left\{M_1=1\right\} + 3X_2\mathds{1}\left\{ M_2 = 1, M_3 = 1\right\}+\varepsilon$.
Note that according to this data generation process, the masks for which at least $X_1$ is missing, and the mask where $X_2$ and $X_3$ are missing, have important predictive power. As there are only 3 features that can be missing in this setting, \Cref{fig:ydepmx_dep,fig:ydepmx_no_dep} represent the 7 different missing patterns. 

These figures highlight that in the easiest setting where the conditional expectation imputation is able to reconstruct the missing values quite accurately ($\varphi=0.8$, \Cref{fig:ydepmx_dep}) \newnested manages to maintain MCV. However, in the hardest case of uncorrelated features ($\varphi=0$, \Cref{fig:ydepmx_no_dep}), it does not achieve MCV as it undercovers on the masks that have predictive power. Yet, \newnested still improves upon vanilla impute-then-predict+CQR, and in particular \mask is slightly more robust than \masksub.

\begin{figure}[!t]

    \centering
    \begin{subfigure}{\textwidth}
        \centerline{\includegraphics[width=\textwidth]{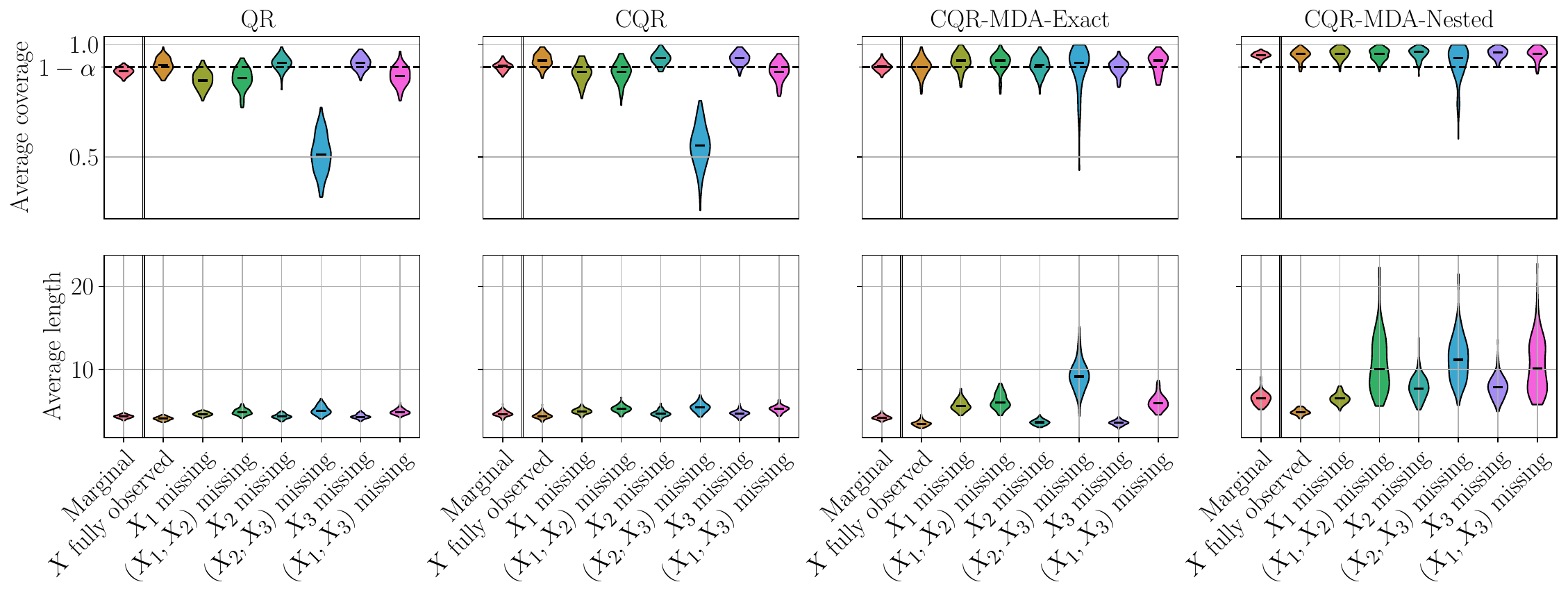}}
        \caption{Dependent Gaussian features, with $\varphi = 0.8$.}
        \label{fig:ydepmx_dep}
    \end{subfigure}
    \hfill
    \begin{subfigure}{\textwidth}
        \centerline{\includegraphics[width=\textwidth]{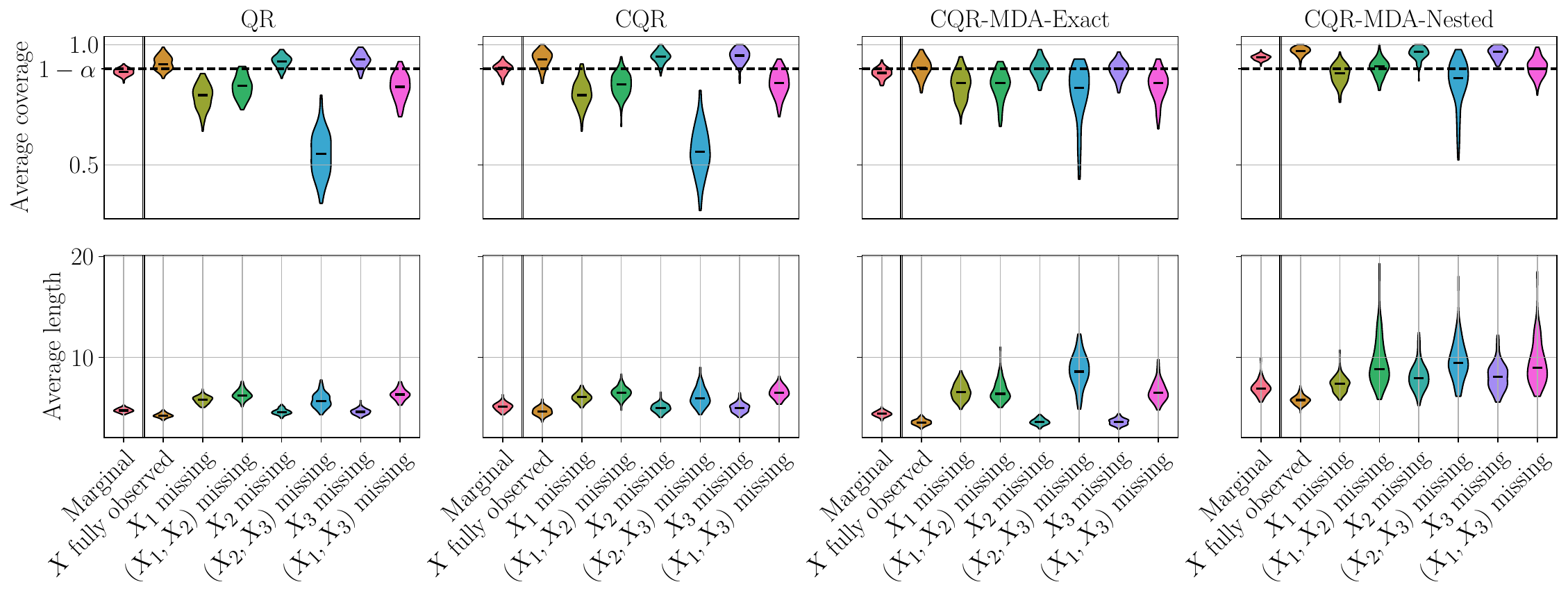}}
        \caption{Independent Gaussian features.}
        \label{fig:ydepmx_no_dep}
    \end{subfigure}
    
\caption{$Y$ and $M$ are not independent given $X$, and the features are Gaussian dependent with $\varphi=0.8$. Average coverage (top) and length (bottom) as a function of the missing patterns. The first violin plot shows the marginal coverage. The marginal test set includes 2000 observations. The mask-conditional test set includes 100 individuals for each missing data pattern.}
\label{fig:ydepmx}
\end{figure}

\section*{Acknowledgements}

This work was supported by a public grant as part of the Investissement d'avenir project, reference ANR-11-LABX-0056-LMH, LabEx LMH. M. Zaffran has been awarded the 2022 Scholarship for Mathematics granted by the Séphora Berrebi Foundation which she gratefully thanks for its support. The work of J. Josse is partially supported by ANR-16-IDEX-0006. Y. Romano was supported by the ISRAEL SCIENCE FOUNDATION (grant No. 729/21). He also thanks the Career Advancement Fellowship, Technion, for providing additional research support. The work of A. Dieuleveut is partially supported by ANR-19-CHIA-0002-01/chaire SCAI and Hi!~Paris.

\appendix
\part*{Appendices}

The appendices are organized as follows. 

\Cref{app:hardness} provides a the proofs for the hardness results presented in \Cref{sec:hardness}.

\Cref{app:model_uq} contains the proofs of the \Cref{sec:glm} results.

\Cref{app:jk} reminds the proof of leave-one-out CP in the case of randomized algorithms.

\Cref{app:nested} derives \newnested theoretical validities proofs, marginal and conditional.

\section{Hardness results}
\label{app:hardness}

\subsection{Reminders on object conditional impossibility results}

Let's start with some reminders on the impossibility result on $X$-conditional coverage (also known as object conditional coverage), originally stated in \cite{lei_distribution-free_2014}, Lemma 1, and re-written more accurately and more generally in \cite{vovk_conditional_2012}, Proposition 4.

\begin{definition}[$X$-conditional coverage]
An estimator $\widehat{C}_{n,\alpha}$ achieves $X$-conditional coverage at the level $\alpha$ if for any distribution $P$ and any $x$:
\begin{equation*}
\mathds{P}_{P^{\otimes(n+1)}}\left(Y^{(n+1)} \in \widehat{C}_{n,\alpha}\left(x\right) | X^{(n+1)}=x \right) \geq 1-\alpha.
\end{equation*}
\label[definition]{def:x-cond}
\end{definition}

\begin{lemma}[$X$-conditional coverage is not achievable in an informative way \citep{vovk_conditional_2012}]
\; \\ Suppose that an estimator $\widehat{C}_{n,\alpha}$ achieves $X$-conditional coverage at the level $\alpha$. Then, for any distribution $P$ and any $x_0$ such that $x_0$ is a non-atomic point of $P$:
\begin{equation*}
\left\{\begin{aligned}
\mathds{P}_{P^{\otimes(n)}}\left(\Lambda \left( \widehat{C}_{n,\alpha}\left(x_0\right) \right) = + \infty \right) \geq 1-\alpha, \;\;\;\;\; & \text{if }\mathcal{Y}\subseteq \mathds{R}  \text{ (regression)},\\
\forall y\in\mathcal{Y}: \;\; \mathds{P}_{P^{\otimes(n)}}\left( y \in \widehat{C}_{n,\alpha}\left(x_0\right) \right) \geq 1-\alpha, \;\;\;\;\; & \text{if }\mathcal{Y}\subseteq \mathds{N} \text{ (classification)}.
\end{aligned}
\right.
\end{equation*}
where $\Lambda$ is the Lebesgue measure.
\label{lem:x-cond-imp}
\end{lemma}

\begin{proof}

Assume $\widehat{C}_{n,\alpha}$ be $X$-conditionally valid, as defined in \Cref{def:x-cond}.  

Let $P$ a distribution on $\mathcal{X} \times \mathcal{Y}$, and let $x_0 \in \text{non-atom}\left(P_{X}\right)$. 

Let $\varepsilon > 0$. Let $\varepsilon_n = \sqrt{2 \left( 1 - \left( 1 - \frac{\varepsilon^2}{8} \right)^{1/n} \right)}$.

Let $E \subseteq \mathcal{X}$ such that $x_0 \in E$ and $0 < P_{X}(E) \leq \varepsilon_n$ (this is possible as a non-atom of a distribution $P_{X}$ belongs to its support).

Before diving in the details of the proof, let us define the total variation distance between two distributions $P$ and $Q$ on $\mathcal{Z}$, denoted $TV(P,Q)$: 
\begin{equation*}
    TV(P,Q) := \sup_{Z \in \mathcal{Z}} | P(Z) - Q(Z)|.
\end{equation*}

$\hookrightarrow$ \underline{Classification case.}

Let $y \in \mathcal{Y}$.

Define $Q$ another distribution on $\mathcal{X}\times\mathcal{Y}$ such that for any $A \subseteq \mathcal{X}$ and for any $B \subseteq \mathcal{Y}$:
\begin{equation*}
    Q\left(A \times B \right) = P\left(A\cap E^c \times B\right) + P_X\left( A \cap E \right)S_y(B),
\end{equation*}
with $S_y$ defined on $\mathcal{Y}$, which is a dirac on $y$.

On the one hand, exactly as in the regression case, by construction, $TV(P,Q) \leq P_X(E) \leq \varepsilon_n$. Hence, using \Cref{lem:bound-tsybakov}, $TV\left(P^{\otimes(n)},Q^{\otimes(n)}\right) \leq \varepsilon$. Therefore, for any $A \subseteq \mathcal{X}$ and for any $B \subseteq \mathcal{Y}$:
\begin{equation}
P^{\otimes(n)}\left(A \times B\right) \geq Q^{\otimes(n)}\left(A \times B\right) - \varepsilon.
\label{eq:tv-donoho}
\end{equation}

On the other hand, let $x \in E$. As $\widehat{C}_{n,\alpha}$ is distribution-free $X$-conditionally valid, it satisfies:
\begin{align*}
1 - \alpha & \leq \mathds{P}_{Q\iid} \left( Y^{(n+1)} \in \widehat{C}_{n,\alpha}(x) | X^{(n+1)} = x \right) \\
& = \mathds{E}_{Q^{\otimes(n)}}\left[ \mathds{E}_Q \left[ \mathds{1}\left\{ Y^{(n+1)} \in \widehat{C}_{n,\alpha}(x) \right\} | X^{(n+1)} = x \right] \right] \\
& = \mathds{E}_{Q^{\otimes(n)}}\left[ \mathds{E}_Q \left[ \mathds{1}\left\{ y \in \widehat{C}_{n,\alpha}(x) \right\} | X^{(n+1)} = x \right] \right] \\
& = \mathds{E}_{Q^{\otimes(n)}}\left[ \mathds{1}\left\{ y \in \widehat{C}_{n,\alpha}(x) \right\} \right] \\
& = \mathds{P}_{Q^{\otimes(n)}}\left( y \in \widehat{C}_{n,\alpha}(x) \right).
\end{align*}

Combining with \Cref{eq:tv-donoho}, we finally get:
\begin{equation*}
\mathds{P}_{P^{\otimes(n)}}\left( y \in \widehat{C}_{n,\alpha}(x) \right) \geq 1 - \alpha - \varepsilon,
\end{equation*}
which concludes the proof for the classification case by letting $\varepsilon \rightarrow 0$.

$\hookrightarrow$ \underline{Regression case.}

Let $D > 0$.

Define $Q$ another distribution on $\mathcal{X} \times \mathcal{Y}$ such that for any $A \subseteq \mathcal{X}$ and for any $B \subseteq \mathcal{Y}$:
\begin{equation*}
    Q\left(A\times B\right) := P\left(A \cap E^c \times B\right) + P_X\left(A \cap E\right)R\left(B\right),
\end{equation*}
with $R$ defined on $\mathcal{Y}$, uniform on $[-D;D]$.

On the one hand, by construction, $TV(P,Q) \leq P_X(E) \leq \varepsilon_n$. Hence, using \Cref{lem:bound-tsybakov}, $TV\left(P^{\otimes(n)},Q^{\otimes(n)}\right) \leq \varepsilon$. Therefore, for any $A \subseteq \mathcal{X}$ and for any $B \subseteq \mathcal{Y}$:
\begin{equation}
P^{\otimes(n)}\left(A \times B\right) \geq Q^{\otimes(n)}\left(A \times B\right) - \varepsilon.
\tag{\ref{eq:tv-donoho}}
\end{equation}

On the other hand, let $x \in E$. As $\widehat{C}_{n,\alpha}$ is distribution-free $X$-conditionally valid, it satisfies:
\begin{align*}
    1 - \alpha & \leq \mathds{P}_{Q\iid} \left( Y^{(n+1)} \in \widehat{C}_{n,\alpha}(x) | X^{(n+1)} = x \right) \\
    & = \mathds{E}_{Q^{\otimes(n)}}\left[\int_{\widehat{C}_{n,\alpha}(x)} q(y|x) \mathrm{d}y \right]\\
    & = \mathds{E}_{Q^{\otimes(n)}}\left[\Lambda\left( \widehat{C}_{n,\alpha}(x) \cap [-D;D] \right) \times \frac{1}{2D} \right].
\end{align*}
Note that $\Lambda\left( \widehat{C}_{n,\alpha}(x) \cap [-D;D] \right) \times \frac{1}{2D} \leq 1$. Therefore, using \Cref{lem:rv-bounded}, for any $t > 0$:
\begin{align*}
    \mathds{P}_{Q^{\otimes(n)}}\left(\Lambda\left( \widehat{C}_{n,\alpha}(x) \cap [-D;D] \right) \times \frac{1}{2D} \geq 1-t \right) & \geq 1-\frac{\alpha}{t} \\
    \mathds{P}_{Q^{\otimes(n)}}\left(\Lambda\left( \widehat{C}_{n,\alpha}(x) \cap [-D;D] \right) \geq (1-t)2D \right) & \geq 1-\frac{\alpha}{t} \\
    \Rightarrow \mathds{P}_{Q^{\otimes(n)}}\left(\Lambda\left( \widehat{C}_{n,\alpha}(x) \right) \geq (1-t)2D \right) & \geq 1-\frac{\alpha}{t}.
\end{align*}
Let $t = 1-\frac{1}{\sqrt{D}}$ and obtain $\mathds{P}_{Q^{\otimes(n)}}\left(\Lambda\left( \widehat{C}_{n,\alpha}(x) \right) \geq 2\sqrt{D} \right) \geq 1-\frac{\alpha}{1-\frac{1}{\sqrt{D}}}$.

Combining with \Cref{eq:tv-donoho}, we finally get:
\begin{equation*}
\mathds{P}_{P^{\otimes(n)}}\left( \Lambda\left( \widehat{C}_{n,\alpha}(x) \right) \geq 2\sqrt{D} \right) \geq 1-\frac{\alpha}{1-\frac{1}{\sqrt{D}}}-\varepsilon.
\end{equation*}
Letting $\varepsilon \rightarrow 0$ and $D \rightarrow +\infty$, the result is proven for the regression case.

\end{proof}

\vspace{-1cm}

\subsection{Proofs of \Cref{sec:hardness}}
\label{app:mcv_hardness}

\subsubsection{Most general distribution-free result: \Cref{prop:tradeoff-mcv}}

\begin{proof}

Let $n \in \mathds{N}^*$ the total training size (proper training and calibration). 

Let $\alpha \in ]0,1[$. 

Let $\widehat{C}_{n,\alpha}$ be MCV, as defined in \Cref{def:mcv}.  

Let $P$ a distribution on $\mathcal{X} \times \mathcal{M} \times \mathcal{Y}$.

Let $m_0 \in \mathcal{M}$. 

Denote by $\rho := P_M(\{m_0\})$.

$\hookrightarrow$ \underline{Regression case.}
\smallbreak

Let $D > 0$.

Define $Q$ another distribution on $\mathcal{X} \times \mathcal{M} \times \mathcal{Y}$ such that for any $A \subseteq \mathcal{X}$, for any $L \subseteq \mathcal{M}$ and for any $B \subseteq \mathcal{Y}$:
\begin{equation*}
    Q\left(A\times L\times B\right) := P\left(A \times L\setminus\{m_0\} \times B\right) + P_{(X,M)}\left(A \times \{m_0\} \right)R\left(B\right),
\end{equation*}
with $R$ defined on $\mathcal{Y}$, uniform on $[-D;D]$.

Recall that the total variation distance between two probability distributions on $\mathcal{Z}$, say $P$ and $Q$, is defined as: $TV(P,Q) := \sup_{Z \in \mathcal{Z}} \vert P(Z) - Q(Z) \vert$.

On the one hand, by construction, $TV(P,Q) \leq P_M(\{m_0\}) = \rho$. Hence, using \Cref{lem:bound-tsybakov}: $TV(P\iid,Q\iid) \leq \sqrt{2\left( 1-\left( 1-\frac{\rho^2}{2}\right)^{n+1}\right)}$. Therefore, for any $A \subseteq \mathcal{X}$, for any $L\subseteq \mathcal{M}$ and for any $B \subseteq \mathcal{Y}$:
\begin{equation}
P\iid\left(A \times L \times B\right) \geq Q\iid\left(A \times L \times B\right) - \sqrt{2\left( 1-\left( 1-\frac{\rho^2}{2}\right)^{n+1}\right)}.
\label{eq:tv-donoho-mcv}
\end{equation}

On the other hand, as $\widehat{C}_{n,\alpha}$ is MCV, it satisfies:
\begin{align*}
    1 - \alpha \leq & \mathds{P}_{Q\iid} \left( Y^{(n+1)} \in \widehat{C}_{n,\alpha}\left( X^{(n+1)}, m_0 \right) | M^{(n+1)} = m_0 \right) \\
    & = \mathds{E}_{Q\iid} \left[ \mathds{1}\left\{ Y^{(n+1)} \in \widehat{C}_{n,\alpha}\left( X^{(n+1)}, m_0 \right) \right\} | M^{(n+1)} = m_0 \right] \\
    & = \begin{aligned}[t]
     \mathds{E}_{Q^{\otimes(n)}} \left[ \mathds{E}_{Q} \right.& \left[ \mathds{1}\left\{ Y^{(n+1)} \in \widehat{C}_{n,\alpha}\left( X^{(n+1)}, m_0 \right) \right\} \right.\\
     &\left.\left. | M^{(n+1)} = m_0, \left( X^{(k)},M^{(k)},Y^{(k)} \right)_{k=1}^{n} \right] \right]
    \end{aligned}\\
    & = \begin{aligned}[t]
       \mathds{E}_{Q^{\otimes(n)}} \left[ \mathds{E}_{Q} \left[\mathds{E}_{Q} \right. \right. & \left[ \mathds{1}\left\{ Y^{(n+1)} \in \widehat{C}_{n,\alpha}\left( X^{(n+1)}, m_0 \right) \right\} \right.\\
         &\left. | X^{(n+1)}, M^{(n+1)} = m_0, \left( X^{(k)},M^{(k)},Y^{(k)} \right)_{k=1}^{n} \right] \\ 
         & \left. \left. | M^{(n+1)} = m_0, \left( X^{(k)},M^{(k)},Y^{(k)} \right)_{k=1}^{n} \right] \right]
    \end{aligned} \\
    & = \begin{aligned}[t]
         \mathds{E}_{Q^{\otimes(n)}} \left[ \mathds{E}_{Q} \right.& \left[ \int_{\widehat{C}_{n,\alpha}\left( X^{(n+1)}, m_0 \right)} q\left(y|X^{(n+1)},m_0\right) \mathrm{d}y \right.\\
         & \left. \left.| M^{(n+1)} = m_0, \left( X^{(k)},M^{(k)},Y^{(k)} \right)_{k=1}^{n} \right] \right]
    \end{aligned} \\
    & = \begin{aligned}[t]
         \mathds{E}_{Q^{\otimes(n)}} \left[ \mathds{E}_{Q}\right.& \left[ \Lambda\left( \widehat{C}_{n,\alpha}\left( X^{(n+1)}, m_0 \right) \cap [-D;D] \right) \times \frac{1}{2D} \right.\\
         &\left. \left. | M^{(n+1)} = m_0, \left( X^{(k)},M^{(k)},Y^{(k)} \right)_{k=1}^{n} \right] \right]
    \end{aligned} \\
    & = \mathds{E}_{Q\iid} \left[ \Lambda\left( \widehat{C}_{n,\alpha}\left( X^{(n+1)}, m_0 \right) \cap [-D;D] \right) \times \frac{1}{2D} | M^{(n+1)} = m_0 \right] 
 \end{align*}

Note that $\Lambda\left( \widehat{C}_{n,\alpha}\left( X^{(n+1)}, m_0 \right) \cap [-D;D] \right) \times \frac{1}{2D} \leq 1$ almost surely. Therefore, using \Cref{lem:rv-bounded}, for any $t > 0$:
\begin{align*}
    \mathds{P}_{Q\iid}\left(\Lambda\left( \widehat{C}_{n,\alpha}\left( X^{(n+1)}, m_0 \right) \cap [-D;D] \right) \times \frac{1}{2D} \geq 1-t \right) & \geq 1-\frac{\alpha}{t} \\
    \mathds{P}_{Q\iid}\left(\Lambda\left( \widehat{C}_{n,\alpha}\left( X^{(n+1)}, m_0 \right) \cap [-D;D] \right) \geq (1-t)2D \right) & \geq 1-\frac{\alpha}{t} \\
    \Rightarrow \mathds{P}_{Q\iid}\left(\Lambda\left( \widehat{C}_{n,\alpha}\left( X^{(n+1)}, m_0 \right) \right) \geq (1-t)2D \right) & \geq 1-\frac{\alpha}{t}.
\end{align*}
Let $t = 1-\frac{1}{\sqrt{D}}$ and obtain $\mathds{P}_{Q\iid}\left(\Lambda\left( \widehat{C}_{n,\alpha}\left( X^{(n+1)}, m_0 \right) \right) \geq 2\sqrt{D} \right) \geq 1-\frac{\alpha}{1-\frac{1}{\sqrt{D}}}$.

Combining with \Cref{eq:tv-donoho-mcv}, we finally get:
\begin{equation*}
\mathds{P}_{P\iid}\left( \Lambda\left( \widehat{C}_{n,\alpha}\left( X^{(n+1)}, m_0 \right) \right) \geq 2\sqrt{D} \right) \geq 1-\frac{\alpha}{1-\frac{1}{\sqrt{D}}}- \sqrt{2\left( 1-\left( 1-\frac{\rho^2}{2}\right)^{n+1}\right)}.
\end{equation*}
Letting $D \rightarrow +\infty$, the result is proven.

$\hookrightarrow$ \underline{Classification case.}
\smallbreak

Let $y \in \mathcal{Y}$.

Define $Q$ another distribution on $\mathcal{X} \times \mathcal{M} \times \mathcal{Y}$ such that for any $A \subseteq \mathcal{X}$, for any $L \subseteq \mathcal{M}$ and for any $B \subseteq \mathcal{Y}$:
\begin{equation*}
    Q\left(A\times L\times B\right) := P\left(A \times L\setminus\{m_0\} \times B\right) + P_{(X,M)}\left(A \times \{m_0\} \right)S\left(B\right),
\end{equation*}
with $S$ defined on $\mathcal{Y}$, being null everywhere except on $y$ (a dirac in $y$).

On the one hand, exactly as in the regression case, by construction, $TV(P,Q) \leq P_X(E) \leq P_M(m_0) = \rho$.  $TV(P\iid,Q\iid) \leq \sqrt{2\left( 1-\left( 1-\frac{\rho^2}{2}\right)^{n+1}\right)}$. Therefore, for any $A \subseteq \mathcal{X}$, for any $L\subseteq \mathcal{M}$ and for any $B \subseteq \mathcal{Y}$:
\begin{equation}
P\iid\left(A \times L \times B\right) \geq Q\iid\left(A \times L \times B\right) - \sqrt{2\left( 1-\left( 1-\frac{\rho^2}{2}\right)^{n+1}\right)}.
\tag{\ref{eq:tv-donoho-mcv}}
\end{equation}

On the other hand, as $\widehat{C}_{n,\alpha}$ is MCV, it satisfies:
\begin{align*}
    1 - \alpha & \leq \mathds{P}_{Q\iid} \left( Y^{(n+1)} \in \widehat{C}_{n,\alpha}\left( X^{(n+1)}, m_0 \right) | M^{(n+1)} = m_0 \right) \\
    & = \begin{aligned}
        \mathds{E}_{Q^{\otimes(n)}}\left[ \mathds{E}_Q \left[ \mathds{1}\left\{ Y^{(n+1)} \in \right. \right.\right.&\left.  \widehat{C}_{n,\alpha}\left(X^{(n+1)},m_0\right) \right\}  \\
        &\left.\left. | M^{(n+1)} = m_0, \left( X^{(k)},M^{(k)},Y^{(k)} \right)_{k=1}^n \right] \right] 
        \end{aligned}\\
    & = \mathds{E}_{Q^{\otimes(n)}}\left[ \mathds{E}_Q \left[ \mathds{1}\left\{ y \in \widehat{C}_{n,\alpha}\left(X^{(n+1)},m_0\right) \right\} | M^{(n+1)} = m_0, \left( X^{(k)},M^{(k)},Y^{(k)} \right)_{k=1}^n \right] \right] \\
    & = \mathds{E}_{Q\iid}\left[ \mathds{1}\left\{ y \in \widehat{C}_{n,\alpha}\left(X^{(n+1)},m_0\right) \right\} \right] \\
    & = \mathds{P}_{Q\iid}\left( y \in \widehat{C}_{n,\alpha}\left(X^{(n+1)},m_0\right) \right).
\end{align*}

Combining with \Cref{eq:tv-donoho}, we finally get:
\begin{equation*}
\mathds{P}_{P\iid}\left( y \in \widehat{C}_{n,\alpha}(X^{(n+1)},m_0) \right) \geq 1 - \alpha - \sqrt{2\left( 1-\left( 1-\frac{\rho^2}{2}\right)^{n+1}\right)}
\end{equation*}
which concludes the proof for the classification case.

\end{proof}

The proof of \Cref{prop:tradeoff-mcv} relied on the following \Cref{lem:bound-tsybakov,lem:rv-bounded}.

\begin{lemma}
For $P$ and $Q$ two probability distributions, and $n \in \mathds{N}^*$, it holds:
\begin{equation*}
TV(P^n,Q^n) \leq \sqrt{2\left(1-\left(1-\frac{TV(P,Q)^2}{2}\right)^n\right)}.
\end{equation*}
\label[lemma]{lem:bound-tsybakov}
\end{lemma}

\begin{proof}
The proof of this lemma is based on the relationship between the total variation distance and the Hellinger distance between two probability distributions denoted by $H(\cdot,\cdot)$ \citep[see][]{Tsybakov2009}.

Let $n \in \mathds{N}^*$ and let $P$ and $Q$ be two probability distributions.

On the one hand, note that:
\begin{equation}
TV(P,Q) \leq H(P,Q).
\label{eq:con-tv-hell}
\end{equation}

On the other hand, observe that:
\begin{equation}
H^2(P^n,Q^n) = 2 \left( 1 - \left( 1 - \frac{H^2(P,Q)}{2} \right)^n \right).
\label{eq:hell-iid}
\end{equation}

Therefore, by combining \Cref{eq:con-tv-hell,eq:hell-iid} \citep[that can be found in][]{Tsybakov2009}, we obtain the desired result. 
\end{proof}

\begin{lemma}
Let $W$ be a random variable such that $0 \leq W \leq 1$ and $\mathds{E}\left[W\right]\geq\beta$ with $\beta \in [0,1]$.

Then, for any $t > 0$, it holds $\mathds{P}\left( W \geq 1 - t \right) \geq 1 - \frac{1-\beta}{t}$.
\label[lemma]{lem:rv-bounded}
\end{lemma}

\begin{proof}
    Let $t > 0$.

    As $W \leq 1$, $1-W\geq0$. Therefore, using Markov's inequality:
    \begin{equation*}
        \mathds{P}\left( 1-W \geq t \right) \leq \frac{\mathds{E}\left[ 1 - W \right]}{t} = \frac{1-\mathds{E}\left[W \right]}{t} \leq \frac{1-\beta}{t}
    \end{equation*}
    Noting that:
    \begin{equation*}
        \mathds{P}\left( 1-W \geq t \right) = \mathds{P}\left( W \leq 1 - t \right) = 1 - \mathds{P}\left( W \geq 1 - t \right),
    \end{equation*}
    we finally get $\mathds{P}\left( W \geq 1 - t \right) \geq 1 - \frac{1-\beta}{t}$.
\end{proof}

\subsubsection{Restricting to $\p\ymx$: \Cref{prop:tradeoff-mcv-yindm}}

\begin{proof}

The skeleton of the proof is the exactly the same than the one of \Cref{prop:tradeoff-mcv}, with a careful attention required in the construction of the adversarial distribution $Q$.

Let $n \in \mathds{N}^*$ the total training size (proper training and calibration). 

Let $\alpha \in ]0,1[$. 

Let $\widehat{C}_{n,\alpha}$ be MCV-$\p\ymx\iid$.  

Let $P\in\p\ymx$.

Let $(X,M,Y) \sim P$.

Let $m_0 \in \mathcal{M}$ such that $\rho := P_M(\{m_0\}) > 0$.

$\hookrightarrow$ \underline{Regression case.}
\smallbreak

Let $D > 0$.

We will now define $Q$ another distribution on $\mathcal{X} \times \mathcal{M} \times \mathcal{Y}$ which is:
\begin{enumerate}[label=(\roman*)]
\item close in total variation to $P$ with respect to $\rho$; \label{itm:tv}
\item such that \Cref{ass:y_ind_m} holds (to ensure that $\widehat{C}_{n,\alpha}$ is also MCV under $Q$); \label{itm:yindmx}
\item \label{itm:ximpm} such that there exists some subset of $\mathcal{X}$, say $F_0$, which determines the event of drawing mask $m_0$ under $Q$. This allows to remark that 
\begin{align*}
& \mathds{P}_{Q\iid} \left( Y^{(n+1)} \in \widehat{C}_{n,\alpha}\left( X^{(n+1)}, m_0 \right) | M^{(n+1)} = m_0 \right) \\
= \; & \mathds{P}_{Q\iid} \left( Y^{(n+1)} \in \widehat{C}_{n,\alpha}\left( X^{(n+1)}, m_0 \right) | X^{(n+1)} \in F_0 \right).
\end{align*}
\end{enumerate}

Let $(\tilde{X},\tilde{M},\tilde{Y}) \sim Q$. $Q$ is built in the following way. 

Let $F_0 \subseteq \mathcal{X}$ such that $P_X(F_0) = \rho$.

\begin{equation*}
\left\{\begin{aligned}
\text{if }X \notin F_0  \text{ and } M \neq m_0:& (\tilde{X},\tilde{M},\tilde{Y}) = (X,M,Y),\\
\text{if }X \in F_0  \text{ or } M = m_0:& (\tilde{X},\tilde{M},\tilde{Y}) \sim \mathcal{U}(F_0) \times \delta_{m_0} \times \mathcal{U}([-D,D]).
\end{aligned}
\right.
\end{equation*}

Using this construction, the proof will follow as in \Cref{prop:tradeoff-mcv}. The only ``tricky points'' to check are \ref{itm:tv}, \ref{itm:yindmx}, and \ref{itm:ximpm}.

By construction, \ref{itm:ximpm} is directly satisfied.

Remark that by construction $\mathds{P}\left( (X,M,Y) \neq (\tilde{X},\tilde{M},\tilde{Y}) \right) \leq 2\delta$ (the worst case scenario being if $F_0$ has been chosen such that $\mathds{1}\left\{ X \in F_0 \right\}\mathds{1}\left\{ M = m_0\right\} \overset{a.s.}{=} 0$, leading to an equality in the previous equation). 
Therefore, using \Cref{lem:tv-event}, we get that $TV(P,Q) \leq 2\delta$, therefore verifying \ref{itm:tv}.

The remaining task is to show that \ref{itm:yindmx} is satisfied. Let $B\in\mathcal{Y}$. We have:
\begin{align*}
\mathds{P}\left( \tilde{Y} \in B | \tilde{X},\tilde{M} \right) & = \left\{\begin{aligned}
\mathds{P}\left( Y \in B| X,M \right) & \text{ if } \tilde{X}\in F_0 \\
\Lambda\left( B \cap [-D;D] \right)\frac{1}{2D} & \text{ if } \tilde{X}\notin F_0
\end{aligned}
\right. \\
& = \left\{\begin{aligned}
\mathds{P}\left( Y \in B| X \right) & \text{ if } \tilde{X}\in F_0 \text{ as } P \text{ satisfies \Cref{ass:y_ind_m}}\\
\Lambda\left( B \cap [-D;D] \right)\frac{1}{2D} & \text{ if } \tilde{X}\notin F_0
\end{aligned}
\right. \\
& = \mathds{P}\left( \tilde{Y} \in B | \tilde{X}\right).
\end{align*}

$\hookrightarrow$ \underline{Classification case.}
\smallbreak
The idea is as previously, except that, as in the other hardness results, we replace the uniform distribution by a Dirac. In particular, let $y \in \mathcal{Y}$. 

Let $(\tilde{X},\tilde{M},\tilde{Y}) \sim Q$. $Q$ is built in the following way. 

Let $F_0 \subseteq \mathcal{X}$ such that $P_X(F_0) = \rho$.

\begin{equation*}
\left\{\begin{aligned}
\text{if }X \notin F_0  \text{ and } M \neq m_0:& (\tilde{X},\tilde{M},\tilde{Y}) = (X,M,Y),\\
\text{if }X \in F_0  \text{ or } M = m_0:& (\tilde{X},\tilde{M},\tilde{Y}) \sim \mathcal{U}(F_0) \times \delta_{m_0} \times \delta_y.
\end{aligned}
\right.
\end{equation*}
The conclusion follows as in \Cref{prop:tradeoff-mcv}, since, as shown in the regression case above, $Q$ is such that: (i) $TV(P,Q)\leq 2\rho$, (ii) \Cref{ass:y_ind_m} and (iii) holds by construction.
\end{proof}

\begin{lemma}
Let $\mathds{P}_Z$ and $\mathds{P}_{Z'}$ be two distributions for the random variables $X$ and $X'$ taking their value in $\mathcal{Z}$. $TV(\mathds{P}_Z, \mathds{P}_{Z'}) \leq \mathds{P}(Z \neq Z')$.
\label[lemma]{lem:tv-event}
\end{lemma}

\begin{proof}
\begin{align*}
TV\left( \mathds{P}_{Z}, \mathds{P}_{Z'} \right) & = \sup_{A \subseteq \mathcal{Z}} | \mathds{P}_Z(A) - \mathds{P}_{Z'}(A) | \\
& =  \sup_{A \subseteq \mathcal{Z}} | \mathds{E}\left[\mathds{1}\left\{Z \in A\right\}\right] - \mathds{E}\left[\mathds{1}\left\{Z \in A\right\}\right] | \\
& \leq  \sup_{A \subseteq \mathcal{Z}} \mathds{E}\left[ | \mathds{1}\left\{Z \in A\right\} - \mathds{1}\left\{Z \in A\right\} |\right] \\
& =  \sup_{A \subseteq \mathcal{Z}} \mathds{E}\left[ | \mathds{1}\left\{Z \in A\right\} - \mathds{1}\left\{Z \in A\right\} | \mathds{1}\left\{ Z \neq Z' \right\} \right] \\
& \leq  \sup_{A \subseteq \mathcal{Z}} \mathds{E}\left[\mathds{1}\left\{ Z \neq Z' \right\} \right] \\
& =  \sup_{A \subseteq \mathcal{Z}} \mathds{P}\left( Z \neq Z' \right)
\end{align*}
\end{proof}
\vspace{-0.75cm}
\section{Link between missing covariates and uncertainty}
\label{app:model_uq}

\subsection{Proofs for Conditional Variance results}
\label{app:proof:varirianceisot}
\allowdisplaybreaks

\subsubsection{Results under $\p\mcarymx$ (\Cref{prop:VarExp})}

\begin{proof}
 Under the assumptions, $M\ind (Y, X)$, and thus for any $m$:
 \begin{align*}
     \mathds{E}\left[ V(X_{\text{obs}(M)}, M ) | M = m \right]   & = \mathds{E}\left[ V(X_{\text{obs}(m)}, m) | M = m \right] \\
     & = \mathds{E}\left[ V(X_{\text{obs}(m)}, m) \right]  \\
     & = \mathds{E}\left[ \mathrm{Var} \left( Y | X_{\text{obs}(m)}\right) \right]  
 \end{align*}
Moreover, for any $m \subset m'$, 
\begin{align*}
    \mathrm{Var} \left( Y | X_{\text{obs}(m')}\right) & = \mathds{E}\left[ \mathrm{Var} \left( Y | X_{\text{obs}(m)}\right) | X_{\text{obs}(m')} \right] + \mathrm{Var}(  \mathds{E}\left[ Y | X_{\text{obs}(m)}\right]  |X_{\obs(m')}).   \\
    & \geq \mathds{E}\left[ \mathrm{Var} \left( Y | X_{\text{obs}(m)}\right) | X_{\text{obs}(m')} \right].
\end{align*}
Thus $\mathds{E} \left[\mathrm{Var} \left( Y | X_{\text{obs}(m')}\right) \right] \geq \mathds{E}\left[ \mathrm{Var} \left( Y | X_{\text{obs}(m)}\right)  \right]$.
And finally: 
\begin{equation*}
    \mathds{E}\left[ V(X_{\text{obs}(M)}, M ) | M = m' \right] \geq \mathds{E}\left[ V(X_{\text{obs}(M)}, M ) | M = m \right].
\end{equation*}
\end{proof}

\subsubsection{Results under Gaussian Linear Model and $\p\mcar$}
\label{app:glm}

Previous works \citep{lemorvan2020,ayme2022,pmlr-v202-zaffran23a} have shown that under \Cref{mod:glm}, $Y | ( X_{\obs(m)},M = m ) \sim \mathcal{N}\left( \tilde\mu^m, \widetilde\sigma^m \right)$ for any $m \in \mathcal{M}$, with:
\begin{align*}
\tilde\mu^m = & \; \beta^T_{\obs(m)} X_{\obs(m)} + \beta^T_{\mis(m)} \mu^m_{{\mis}|{\obs}} \\
\mu^m_{{\mis}|{\obs}} = & \; \mu^m_{\mis(m)} +  \Sigma^m_{\mis(m),\obs(m)} ({\Sigma^m_{\obs(m),\obs(m)}})^{-1}(X_{\obs(m)} - \mu^m_{\obs(m)}), \\
\widetilde\sigma^m = & \; \beta^T_{\mis(m)} \Sigma^m_{{\mis}|{\obs}} \beta_{\mis(m)} + \sigma^2_{\varepsilon} \\
\Sigma^m_{{\mis}|{\obs}} = & \; \Sigma^m_{\mis(m),\mis(m)} - \Sigma^m_{\mis(m),\obs(m)} ({\Sigma^m_{\obs(m),\obs(m)}})^{-1}\Sigma^m_{\obs(m),\mis(m)}.
\end{align*}

We now provide the proof of \Cref{prop:glm_var_incr}.
\begin{proof}

Consider \Cref{mod:glm} and assume additionally that the missing mechanism is MCAR. Therefore, for any $m \in \mathcal{M}$, $\Sigma^m = \Sigma$. Hence, for any $m \in \mathcal{M}$:
\begin{equation*}
\text{Var}\left(Y|X_{\obs(m)},M=m\right) = \beta^T_{\mis(m)}\Sigma^m_{\mis|\obs} \beta_{\mis(m)} + \sigma^2_{\varepsilon},
\end{equation*}
with $\Sigma^m_{\mis|\obs} = \Sigma_{\mis(m),\mis(m)} - \Sigma_{\mis(m),\obs(m)} ({\Sigma_{\obs(m),\obs(m)}})^{-1}\Sigma_{\obs(m),\mis(m)}$.

Let $(m,m')\in \mathcal{M}^2$ such that $m \subseteq m'$. Our goal is to show that:
\begin{align*}
    \text{Var}\left(Y|X_{\obs(m')},M=m'\right) - \text{Var}\left(Y|X_{\obs(m)},M=m\right) & \geq 0 \\
    \beta^T_{\mis(m')}\Sigma^{m'}_{\mis|\obs} \beta_{\mis(m')} + \sigma^2_{\varepsilon} - \beta^T_{\mis(m)}\Sigma^m_{\mis|\obs} \beta_{\mis(m)} - \sigma^2_{\varepsilon} & \geq 0 \\
    \beta^T_{\mis(m')}\Sigma^{m'}_{\mis|\obs} \beta_{\mis(m')} - \beta^T_{\mis(m)}\Sigma^m_{\mis|\obs} \beta_{\mis(m)} & \geq 0 \\
    \beta^T_{\mis(m')}\Sigma^{m'}_{\mis|\obs} \beta_{\mis(m')} - \beta^T_{\mis(m')}\begin{pmatrix}
        \Sigma^m_{\mis|\obs} & 0 \\ 0 & \mathbf{0}
    \end{pmatrix}\beta_{\mis(m')} & \geq 0 \\
    \beta^T_{\mis(m')}\left( \Sigma^{m'}_{\mis|\obs} - \begin{pmatrix}
        \Sigma^m_{\mis|\obs} & 0 \\ 0 & \mathbf{0}
    \end{pmatrix} \right) \beta_{\mis(m')} & \geq 0,
\end{align*}
holds for any $\beta$. Therefore, we have to show that $\Sigma^{m'}_{\mis|\obs} - \begin{pmatrix}
        \Sigma^m_{\mis|\obs} & 0 \\ 0 & \mathbf{0}
    \end{pmatrix}$ is semi-definite positive.

The marginal covariance matrix $\Sigma$ can be rewritten by blocks in the following way:
\begin{equation*}
\Sigma = \begin{pmatrix} A & B & C \\ 
B^T & D & E \\
C^T & E^T & F
\end{pmatrix},
\end{equation*}
where:
\begin{equation*}
\left\{\begin{aligned}
A & = \Sigma_{\mis(m),\mis(m)}, \\
\begin{pmatrix} 
D & E \\
E^T & F
\end{pmatrix} & = \Sigma_{\obs(m),\obs(m)}, \\
\begin{pmatrix} 
A & B \\
B^T & D
\end{pmatrix} & = \Sigma_{\mis(m'),\mis(m')}, \\
F & = \Sigma_{\obs(m'),\obs(m')}.
\end{aligned}
\right.
\end{equation*}

Additionally, assume that $\Sigma > 0$ (that is, $\Sigma$ is definite positive).

Therefore, $D > 0, F > 0$. 
Thus $F$ is invertible, of inverse $F^{-1} > 0$. Furthermore, $G := D-EF^{-1}E^T$ is also positive definite, as it is the sum of $D > 0$ and  $EF^{-1}E^T \geq 0$, and thus $G$ is invertible.

$\Sigma^m_{\mis|\obs}$ and $\Sigma^{m'}_{\mis|\obs}$ can be rewritten using the previous decomposition. 

On the one hand, for $m$ it gives:
\begin{align*}
    \Sigma^m_{\mis|\obs} = & A - \begin{pmatrix}
        B & C
    \end{pmatrix}\begin{pmatrix}
        D & E \\ E^T & F
    \end{pmatrix}^{-1} \begin{pmatrix}
        B^T \\ C^T
    \end{pmatrix} \\
    = & A - \begin{pmatrix}
        B & C
    \end{pmatrix}\begin{pmatrix}
        G^{-1} & -G^{-1}EF^{-1} \\ -F^{-1}E^TG^{-1} & F^{-1} + F^{-1}E^TG^{-1}EF^T
    \end{pmatrix} \begin{pmatrix}
        B^T \\ C^T
    \end{pmatrix} \\
    = & A - \begin{pmatrix}
        B & C
    \end{pmatrix}\begin{pmatrix}
        G^{-1}B^T -G^{-1}EF^{-1}C^T \\ -F^{-1}E^TG^{-1}B^T + F^{-1}C^T + F^{-1}E^TG^{-1}EF^TC^T
    \end{pmatrix} \\
    = & A - BG^{-1}B^T + BG^{-1}EF^{-1}C^T \\
    & + CF^{-1}E^TG^{-1}B^T - CF^{-1}C^T - CF^{-1}E^TG^{-1}EF^TC^T \\
    \text{(rearranging)} = & A - CF^{-1}C^T - BG^{-1}B^T + BG^{-1}EF^{-1}C^T \\
    & + CF^{-1}E^TG^{-1}B^T - CF^{-1}E^TG^{-1}EF^TC^T \\
    = & A - CF^{-1}C^T - BG^{-1}\left(B^T - EF^{-1}C^T \right) \\
    & + CF^{-1}E^TG^{-1}\left(B^T - EF^TC^T \right) \\
    = & A - CF^{-1}C^T - \left(B-CF^{-1}E^T\right)G^{-1}\left(B^T - EF^{-1}C^T \right),
\end{align*}
and by denoting $H := B-CF^{-1}E^T$, we finally obtain (as $F$ is symmetric):
\begin{equation*}
    \Sigma^m_{\mis|\obs} = A - CF^{-1}C^T - HG^{-1}H^T.
\end{equation*}

On the other hand, for $m'$:
\begin{align*}
    \Sigma^{m'}_{\mis|\obs} = & \begin{pmatrix}
        A & B \\ B^T & D
    \end{pmatrix} - \begin{pmatrix}
        C \\ E
    \end{pmatrix}F^{-1} \begin{pmatrix}
        C^T & E^T
    \end{pmatrix} \\
    = & \begin{pmatrix}
        A & B \\ B^T & D
    \end{pmatrix} - \begin{pmatrix}
        CF^{-1}C^T & CF^{-1}E^T \\ EF^{-1}C^T & EF^{-1}E^T
    \end{pmatrix} \\
    = & \begin{pmatrix}
        A - CF^{-1}C^T & B - CF^{-1}E^T \\ B^T - EF^{-1}C^T & D -EF^{-1}E^T
    \end{pmatrix} \\
    = & \begin{pmatrix}
        A - CF^{-1}C^T & B - CF^{-1}E^T \\ B^T - EF^{-1}C^T & G
    \end{pmatrix} \\
    \Sigma^{m'}_{\mis|\obs} = & \begin{pmatrix}
        A - CF^{-1}C^T & H \\ H^T & G
    \end{pmatrix}
\end{align*}

Therefore, combining the two terms and rewriting together, we obtain:
\begin{align*}
    \Sigma^{m'}_{\mis|\obs} - \begin{pmatrix}
        \Sigma^m_{\mis|\obs} & 0 \\ 0 & \mathbf{0}
    \end{pmatrix} & = \begin{pmatrix}
        A - CF^{-1}C^T & H \\ H^T & G
    \end{pmatrix} - \begin{pmatrix}
        A - CF^{-1}C^T - HG^{-1}H^T & 0 \\ 0 & \mathbf{0}
    \end{pmatrix} \\
    & = \begin{pmatrix}
        A - CF^{-1}C^T - A + CF^{-1}C^T + HG^{-1}H^T& H \\ H^T & G
    \end{pmatrix} \\
    \Sigma^{m'}_{\mis|\obs} - \begin{pmatrix}
        \Sigma^m_{\mis|\obs} & 0 \\ 0 & \mathbf{0}
    \end{pmatrix} & = \begin{pmatrix}
        HG^{-1}H^T& H \\ H^T & G
    \end{pmatrix}.
\end{align*}

Hence, our objective is to show that $\begin{pmatrix}
        HG^{-1}H^T& H \\ H^T & G
    \end{pmatrix}$ is semi-definite positive. 

Let $z = \begin{pmatrix}
    x & y
\end{pmatrix} \in \mathds{R}^{1 \times (\#m + (\#m' - \#m))}$.

\begin{align*}
    z \begin{pmatrix}
        HG^{-1}H^T& H \\ H^T & G
    \end{pmatrix} z^T & = \begin{pmatrix}
    x & y
\end{pmatrix} \begin{pmatrix}
        HG^{-1}H^T& H \\ H^T & G
    \end{pmatrix} \begin{pmatrix}
    x^T \\ y^T
\end{pmatrix} \\
& = xHG^{-1}H^Tx^T+xHy^T+yH^Tx^T+yGy^T \\
& = xHG^{-1}GG^{-1}H^Tx^T+xHG^{-1}Gy^T+yGG^{-1}H^Tx^T+yGy^T \\
& = xHG^{-1}G\left(G^{-1}H^Tx^T+y^T\right)+yG\left(G^{-1}H^Tx^T+y^T\right) \\
& = \left(xHG^{-1}+y\right)G\left(G^{-1}H^Tx^T+y^T\right) \\
& = \left(xHG^{-1}+y\right)G\left(xHG^{-1}+y\right)^T \\
& \geq 0 \text{ as $G$ is positive definite.} 
\end{align*}

\end{proof}

\subsection{Impact of the imputation under a linear quantile regression model (\Cref{prop:impute_mean_qr})}

To prove \Cref{itm:mean} of \Cref{prop:impute_mean_qr}, we prove the following \Cref{lem:linear}.

\begin{lemma}

Assume $\p\mcar$, and $Y = {\beta^*}^T X + \varepsilon$ with $\varepsilon$ s.t. $\mathds{E}\left[\varepsilon | X_{\obs(M)} , M \right] = 0$. 

Then $\mathds{E}\left[ Y | X_{\obs(M)} , M \right] = {\beta^*}^T \Phi_{\rm{conditional \; mean}}(X,M)$, with $\Phi_{\rm{conditional \; mean}}$ the imputation by the conditional mean. Furthermore, if the covariates are independent, then $\mathds{E}\left[ Y | X_{\obs(M)} , M \right] = {\beta^*}^T \Phi_{\rm{mean}}(X,M)$, with $\Phi_{\rm{mean}}$ the imputation by the mean.
\label[lemma]{lem:linear}
\end{lemma} 

\begin{proof}
    \begin{align*}
        \mathds{E}\left[ Y | X_{\obs(M)} , M \right] =  \mathds{E}\left[ {\beta^*}^T X | X_{\obs(M)} , M \right] & = \sum_{i=1}^d \beta_i^* \mathds{E}\left[  X_i | X_{\obs(M)} , M \right] \\
        & = \sum_{i=1}^d \beta_i^* (\begin{aligned}[t] & X_i\mathds{1}\left\{ i\in\obs(M) \right\} \\ & + \mathds{E}\left[X_i|X_{\obs(M)},M\right]\mathds{1}\left\{ i\not\in\obs(M) \right\}) \end{aligned} \\
        \p\mcar \rightarrow & = \sum_{i=1}^d \beta_i^* (\begin{aligned}[t] & X_i\mathds{1}\left\{ i\in\obs(M) \right\} \\ & + \mathds{E}\left[X_i|X_{\obs(M)}\right]\mathds{1}\left\{ i\not\in\obs(M) \right\}) \end{aligned} \\
        & = \sum_{i=1}^d \beta_i^* \left(\Phi_{\rm{conditional \; mean}}(X,M)\right)_i  \\ 
        \text{if} \left(X_i\right)_{i=1}^d \ind, \mathds{E}\left[X_i|X_{\obs(M)}\right] = \mathds{E}\left[X_i\right] \rightarrow & = \sum_{i=1}^d \beta_i^* \left(\Phi_{\rm{mean}}(X,M)\right)_i 
    \end{align*}
\end{proof}

To prove \Cref{itm:qr} of \Cref{prop:impute_mean_qr}, we prove the following \Cref{prop:qr_linear}. Indeed, the oracle predictive intervals vary at least once in length we respect to the patterns, as, on the one hand, under $\p\mcarymx$ \Cref{eq:CIL_isot_exp} holds and, on the other hand, when $Y\not\ind X$ the variance of $Y$ given $X$ is different than the overall variance of $Y$.
\begin{proposition}[Non-adaptivity of the linear quantile regression]
Assume that: 
\begin{enumerate}[label=\roman*),topsep=0pt,noitemsep,leftmargin=*]
    \item the quantile regression is learned within the class of linear models;
    \item the (random) values used to impute have the same expectation than the feature itself, i.e., $\mathds{E}\left[ \Phi(X,m) | M = m \right] = \mathds{E}\left[ X \right] $ for any $m \in \mathcal{M}$ such that $\mathds{P}(M = m)>0$.\label{ass:imp_exp}
\end{enumerate}
Then the expectation of the predictive intervals length is independent of the missing pattern.

\label[proposition]{prop:qr_linear}
\end{proposition}

\begin{proof}

Since the quantile regression is learned within the class of linear models, the fitted quantile functions (upper and lower) can be written as $\widehat{q}_\delta(z) = \beta^T_\delta z + \beta^0_\delta$,  with $\beta \in \mathds{R}^d$ and $\beta^0 \in \mathds{R}$. Therefore, the length of the resulting interval $L_\alpha$ at some---imputed---point $\Phi(X_{\obs(M)},M)$ will be:
\begin{align*}
    L_\alpha(\Phi(X_{\obs(M)},M)) & := \widehat{q}_{\delta_{(u)}}(\Phi(X_{\obs(M)},M)) - \widehat{q}_{\delta_{(l)}}(\Phi(X_{\obs(M)},M)) \\
    & = \left( \beta^T_{\delta_{(u)}} - \beta^T_{\delta_{(l)}} \right) \Phi(X_{\obs(M)},M) + \beta^0_{\delta_{(u)}} - \beta^0_{\delta_{(l)}},
\end{align*}
with $\delta_{(l)}$ and $\delta_{(u)}$ chosen by the user or fixed by the algorithm such that $\delta_{(u)} - \delta_{(l)} = 1-\alpha$. Thus:
\begin{align*}
    \mathds{E}\left[L_\alpha(\Phi(X_{\obs(M)},M))\right] & = 
    \mathds{E}\left[\left( \beta^T_{\delta_{(u)}} - \beta^T_{\delta_{(l)}} \right) \Phi(X_{\obs(M)},M) + \beta^0_{\delta_{(u)}} - \beta^0_{\delta_{(l)}}\right] \\
    & = 
    \left( \beta^T_{\delta_{(u)}} - \beta^T_{\delta_{(l)}} \right) \mathds{E}\left[ \Phi(X_{\obs(M)},M) \right] + \beta^0_{\delta_{(u)}} - \beta^0_{\delta_{(l)}}.
\end{align*}

Let $m \in \mathcal{M}$ such that $\mathds{P}(M = m)>0$. Conditioning by $m$:
\begin{equation*}
    \mathds{E}\left[L_\alpha(\Phi(X_{\obs(M)},M)) | M = m\right] = 
    \left( \beta^T_{\delta_{(u)}} - \beta^T_{\delta_{(l)}} \right) \mathds{E}\left[ \Phi(X_{\obs(M)},M) | M = m\right] + \beta^0_{\delta_{(u)}} - \beta^0_{\delta_{(l)}}.
\end{equation*}

Given the assumption that  $\mathds{E}\left[ \Phi(X_{\obs(M)},M) | M = m \right] = \mathds{E}\left[ X \right]$, one can conclude that:
\begin{equation*}
    \mathds{E}\left[L_\alpha(\Phi(X_{\obs(M)},M)) | M = m\right] = \sum_{j=1}^d\left( \beta^T_{\delta_{(u)}} - \beta^T_{\delta_{(l)}} \right)_j \mathds{E}\left[ X \right] + \beta^0_{\delta_{(u)}} - \beta^0_{\delta_{(l)}} \ind M.
\end{equation*}
\end{proof}

\section{Leave-one-out predictive sets for randomized algorithms}
\label{app:jk}

We provide in this section a more detailed proof of leave-one-out or $k$-fold cross-conformal \citep{Vovk2013} and jackknife+ \citep{barber2021jackknife} methods which allows us to highlight where exactly the arguments of data exchangeability and symmetrical algorithm play a role. In particular, by emphasizing these precise influences, we can understand how to include a non-deterministic symmetrical algorithm (such as Random Forest or Stochastic Gradient Descent). 

\subsection{On the definition of randomized symmetric algorithms}

\begin{definition}[Randomized learning algorithm]
A randomized learning algorithm is defined as:
\begin{align*}
\mathcal{A}: \left( \bigcup_{n\geq0} \left( \mathcal{X} \times \mathcal{Y} \right)^n \right) \times [0,1] & \mapsto \mathcal{Y}^{\mathcal{X}} \\
\left( X^{(k)}, Y^{(k)} \right)_{k=1}^n \times \xi & \mapsto \hat{A}(\cdot)
\end{align*}
where $\xi$ encodes the randomness of $\mathcal{A}$.
\end{definition}

\begin{definition}[Randomized symmetric algorithm \citep{Kim2023}]
A randomized learning algorithm $\mathcal{A}$ is symmetric if for any data set $\left( X^{(k)}, Y^{(k)} \right)_{k=1}^n$, for any permutation $\sigma$ on $\llbracket1,n\rrbracket$, there exists a coupling that maps $\xi \sim \mathcal{U}([0,1])$ to $\xi' \sim \mathcal{U}([0,1])$, which depends only on $\sigma$, s.t.:
\begin{equation*}
\mathcal{A}\left(\left( X^{(k)}, Y^{(k)} \right)_{k=1}^n ; \xi\right) = \mathcal{A}\left(\left( X^{(\sigma(k))}, Y^{(\sigma(k))} \right)_{k=1}^n ; \xi' \right).
\end{equation*}
\end{definition}

\subsection{Detailing leave-one-out conformal predictors validity proof}

Let $\left( X^{(k)}, Y^{(k)} \right)_{k=1}^{n+1}$ be exchangeable, and $\mathcal{A}$ a (possible randomized) symmetric algorithm. 

Let $s$ be a conformity score function. For $i \in \llbracket 1,n\rrbracket$, denote $\hat{A}_{-i}(\cdot) := \mathcal{A}\left( \left( X^{(k)}, Y^{(k)} \right)_{\substack{k=1\\k\neq i}}^{n} \right)$, that is the fitted left-one-out algorithm, removing data point $i$. 

Consider the leave-one-out conformal estimator defined as:
$$\widehat{C}_{n,\alpha}^{\text{LOO}}(x) := \left\{ y \in \mathcal{Y}: \sum_{k = 1}^n \mathds{1}\left\{ s\left( X^{(k)}, Y^{(k)} ; \hat{A}_{-k} \right) <  s\left( x, y ; \hat{A}_{-k} \right)  \right\} < (1-\alpha)(n+1) \right\}$$.

Previous works \citep{barber2021jackknife,gupta} have proven that under exchangeability of $\left( X^{(k)}, Y^{(k)} \right)_{k=1}^{n+1}$ and symmetry of $\mathcal{A}$, $\mathds{P}\left( Y^{(n+1)} \in \widehat{C}_{n,\alpha}^{\text{LOO}}\left( X^{(n+1)} \right) \right) \geq 1-2\alpha$. We recall below the key proof's steps, detailing the last one which uses the exchangeability and symmetry arguments.

\paragraph{Step 1.} Remark that:
\begin{align*}
& \left\{ Y^{(n+1)} \notin \widehat{C}_{n,\alpha}^{\text{LOO}}\left( X^{(n+1)} \right) \right\} \\
= & \left\{ \sum_{k = 1}^n \mathds{1}\left\{ s\left( X^{(k)}, Y^{(k)} ; \hat{A}_{-k} \right) <  s\left( X^{(n+1)}, Y^{(n+1)} ; \hat{A}_{-k} \right)  \right\} \geq (1-\alpha)(n+1)  \right\} \\
:= & \left\{ \sum_{k = 1}^n \mathds{1}\left\{ S^{(k),n+1} <  S^{(n+1),k}  \right\} \geq (1-\alpha)(n+1)  \right\} \\
:= & \left\{ \sum_{k = 1}^n \mathcal{C}_{n+1,k} \geq (1-\alpha)(n+1)  \right\}.
\end{align*}
with $S^{(i),j} := s\left( X^{(i)}, Y^{(i)} ; \hat{A}_{-(i,j)} \right)$ the score on data point $i$ of the predictor that has been fitted without seeing nor data point $i$ nor data point $j$, for $(i,j)\in\llbracket1,n+1\rrbracket^2$ and extending $\hat{A}_{-i}$ to $\hat{A}_{-(i,j)} := \mathcal{A}\left( \left( X^{(k)}, Y^{(k)} \right)_{\substack{k=1\\k\notin \{i,j\}}}^{n+1} \right)$, where the $n+1$ data point is added.

Denote by $\mathdutchcal{C}_{\mathcal{A}}$ the function building the comparison matrix $\mathcal{C}\in\{0,1\}^{(n+1)\times(n+1)}$: \\
$\mathdutchcal{C}_{\mathcal{A}} \left( \left( X^{(k)},Y^{(k)}\right)_{k=1}^{n+1} \right)_{i,j} = \mathds{1}\left\{ S^{(i),j} > S^{(j),i} \right\} = \mathcal{C}_{i,j}$.

\paragraph{Step 2.} Deterministically, \citet{barber2021jackknife} shows that $\#\{ i \in \llbracket 1,n+1\rrbracket : \sum\limits_{j = 1}^{n+1} \mathcal{C}_{i,j} \geq (1-\alpha)(n+1) \} \leq 2\alpha(n+1)$. This is shown for \emph{any} comparison matrix.

\paragraph{Step 3.} The last (and crucial) step of leave-one-out conformal predictors is to show that for any permutation $\sigma$ on $\llbracket1,n+1\rrbracket$ it holds: $\left(\mathcal{C}_{\sigma(i),\sigma(j)}\right)_{i,j} \overset{d}{=} \left(\mathcal{C}_{i,j}\right)_{i,j}$.

\allowdisplaybreaks

\begin{align*}
\mathcal{C}_{\sigma(i), \sigma(j)} 
    & = \mathdutchcal{C}_{\mathcal{A}} \left( \left( X^{(k)},Y^{(k)}\right)_{k=1}^{n+1} \right)_{\sigma(i),\sigma(j)} \\
    & = 
    \begin{aligned}[t] 
        \mathds{1}\Biggl\{& s\left(Y^{(\sigma(i))}, X^{(\sigma(i))}, \mathcal{A}\left( \left( X^{(k)},Y^{(k)} \right)_{k=1, k\notin\{\sigma(i),\sigma(j)\}}^{n+1} ; \xi \right) \right)  \\
            & > s\left(Y^{(\sigma(j))}, X^{(\sigma(j))}, \mathcal{A}\left( \left( X^{(k)},Y^{(k)} \right)_{k=1, k\notin\{\sigma(i),\sigma(j)\}}^{n+1} ; \xi \right) \right) \Biggr\} 
    \end{aligned} \\
    & = 
    \begin{aligned}[t] 
        \mathds{1}\Biggl\{& s\left(Y^{(\sigma(i))}, X^{(\sigma(i))}, \mathcal{A}\left( \left( X^{(\sigma(k))},Y^{(\sigma(k))} \right)_{k=1, k\notin\{i,j\}}^{n+1} ; \xi'_{\sigma} \right) \right) \\
        & > s\left(Y^{(\sigma(j))}, X^{(\sigma(j))}, \mathcal{A}\left( \left( X^{(\sigma(k))},Y^{(\sigma(k))} \right)_{k=1, k\notin\{i,j\}}^{n+1} ; \xi'_{\sigma} \right) \right) \Biggr\} \quad \mathcal{A} \text{ is symmetric}
    \end{aligned} \\
    & = \mathdutchcal{C}_{\mathcal{A}} \left( \left( X^{(\sigma(k))},Y^{(\sigma(k))}\right)_{k=1}^{n+1} \right)_{i,j}
\end{align*}
Thus, leveraging the fact that $\xi'_{\sigma} \ind \left( X^{(k)},Y^{(k)}\right)_{k=1}^{n+1}$ and that $\left( X^{(k)},Y^{(k)}\right)_{k=1}^{n+1}$ are exchangeable, we obtain that: 
\begin{equation*}
    \left(\mathcal{C}_{\sigma(i), \sigma(j)}\right)_{i,j \in \llbracket1,n+1\rrbracket^2} \overset{d}{=} \mathdutchcal{C}_{\mathcal{A}} \left( \left( X^{(k)},Y^{(k)}\right)_{k=1}^{n+1} \right) = \left(\mathcal{C}_{i,j}\right)_{i,j \in \llbracket1,n+1\rrbracket^2}.
\end{equation*}

Hence, for any permutation $\sigma$ on $\llbracket 1,n+1 \rrbracket$ it holds that $\Pi^T_{\sigma} \mathcal{C} \Pi_{\sigma} \overset{d}{=} \mathcal{C}$, concluding the proof as then each element of $\llbracket 1,n+1\rrbracket$ is equally likely to belong to $\{ i \in \llbracket 1,n+1\rrbracket : \sum\limits_{j = 1}^{n+1} \mathcal{C}_{i,j} \geq (1-\alpha)(n+1) \}$.

\section{Theory on \newnested and \mask}
\label{app:nested}
Let us first remark that $\widehat C_{n,\alpha}^{\text{MDA-Nested}^{\star}}(\cdot) \subseteq \widehat C_{n,\alpha}^{\text{MDA-Nested}}(\cdot)$ when the conformity score function outputs intervals and $\widetilde{\Cal} = \Cal$ (\Cref{rem:nested_in_newnested}). 

\begin{proof}
\begin{align*}
\left\{ Y^{(n+1)} \notin \widehat{C}_{n,\alpha}^{\text{MDA-Nested}} \right.&\left. \left(X^{(n+1)},M^{(n+1)}\right)  \right\}\\
    & = \begin{aligned}[t] \Bigl\{ & Y^{(n+1)} > \widehat{Q}_{1-\alpha}\left( \mathcal{U}_\alpha\left( X^{(n+1)} \right) \right)  \\ & \left. \text{ or } Y^{(n+1)} < \widehat{Q}_{\alpha}\left( \mathcal{L}_\alpha\left( X^{(n+1)} \right) \right) \right\} \end{aligned} \\
    & = \left\{ (1-\alpha)(\#\Cal+1) \leq \sum\limits_{k = 1}^n \mathds{1}\left\{Y^{(n+1)} > u_{\alpha}^{(k)}\left( X^{(n+1)} \right) \right\} \right. \\
    & \; \; \; \; \; \; \left. \text{ or } (1-\alpha)(\#\Cal+1) \leq \sum\limits_{k = 1}^n \mathds{1}\left\{Y^{(n+1)} <  \ell_{\alpha}^{(k)}\left( X^{(n+1)} \right) \right\} \right\} \\
    & \subset \Biggl\{ (1-\alpha)(\#\Cal+1) \leq \sum\limits_{k = 1}^n \begin{aligned}[t] \mathds{1}\Bigl\{ & Y^{(n+1)} > u_{\alpha}^{(k)}\left( X^{(n+1)} \right) \\ & \left. \text{ or } Y^{(n+1)} < \ell_{\alpha}^{(k)}\left( X^{(n+1)} \right) \right\} \Biggr\} \end{aligned} \\
    & = \Biggl\{ \begin{aligned}[t] & (1-\alpha)(\#\Cal+1) \\ & \leq \sum\limits_{k = 1}^n \begin{aligned}[t] \mathds{1}\Bigl\{ & s\left( \left(X^{(n+1)}, \widetilde{M}^{(k)}\right), Y^{(n+1)};\hat{A}\left(\Phi\left(\cdot,\cdot\right),\cdot\right)\right) \\ & \left. > s\left( \left(X^{(k)}, \widetilde{M}^{(k)}\right), Y^{(k)};\hat{A}\left(\Phi\left(\cdot,\cdot\right),\cdot\right)\right) \right\} \Biggr\} \end{aligned} \end{aligned} \\
    & = \left\{ Y^{(n+1)} \notin \widehat{C}_{n,\alpha}^{\text{MDA-Nested}^{\star}}\left(X^{(n+1)},M^{(n+1)}\right)  \right\}
\end{align*}
\end{proof}
Therefore, any upper bound on the miscoverage of \newnested extends to \mask. 

\subsection{Marginal validity of \newnested.}
\label{app:nested_mv}

The proof of \Cref{thm:mv-nested} is highly inspired by the leave-one-out conformal predictors proof, from \citet{barber2021jackknife} and detailed previously in \Cref{app:jk}.

\begin{proof}

One can see this proof as analogous of the one of leave-one-out conformal predictors, where ``predicting on point $i$ with point $j$ left out'' corresponds to ``predicting on point $i$ when additionally masking it with the mask of point $j$''.

\paragraph{Step 1.}

\begin{align*}
\left\{ Y^{(n+1)} \right.&\left. \notin \widehat C_{n,\alpha}^{\text{MDA-Nested}^{\star}}\left(X^{(n+1)},M^{(n+1)}\right) \right\} \\
    & = \Biggl\{ 
        \begin{aligned}[t] & (1-\alpha)(\#\Cal+1) \\ 
        & \leq \sum\limits_{k \in \Cal} 
            \begin{aligned}[t] \mathds{1}\Bigl\{ 
                & s\left( \left(X^{(n+1)}, \widetilde{M}^{(k)}\right), Y^{(n+1)};\hat{A}\left(\Phi\left(\cdot,\cdot\right),\cdot\right)\right) \\ 
                & \left. > s\left( \left(X^{(k)}, \widetilde{M}^{(k)}\right), Y^{(k)};\hat{A}\left(\Phi\left(\cdot,\cdot\right),\cdot\right)\right) \right\} \Biggr\} \end{aligned} 
            \end{aligned} \\
    & := \left\{ (1-\alpha)(\#\Cal+1) \leq \sum\limits_{k \in \Cal} \mathds{1}\left\{ S^{(n+1),k} > S^{(k), n+1} \right\} \right\},
\end{align*}
where we defined $S^{(i),j} := s\left( \left(X^{(i)}, \max\left( M^{(i)}, M^{(j)} \right) \right), Y^{(i)} ;\hat{A}\left(\Phi\left(\cdot,\cdot\right),\cdot\right) \right)$, that is the score of the point $i$ when the mask of the point $j$ is applied to it, on top of its own mask $M^{(i)}$. 

\paragraph{Step 2.} Define the comparison matrix $\mathcal{C} \in \{0,1\}^{(\#\Cal+1) \times (\#\Cal+1)}$, s.t. for $(i,j) \in \left( \Cal \cup \{ n+1 \} \right)^2$: ${ \mathcal{C}_{i,j} = \mathds{1} \left\{ S^{(i),j} > S^{(j), i} \right\} }$. Hence, we now have (since by definition $\mathcal{C}_{n+1,n+1} = 0$):
\begin{equation*}
\left\{ Y^{(n+1)} \notin \widehat C_{n,\alpha}^{\text{MDA-Nested}^{\star}}\left(X^{(n+1)},M^{(n+1)}\right) \right\} = \left\{ \sum\limits_{k \in \Cal \cup \{ n+1 \}} \mathcal{C}_{n+1,k} \geq (1-\alpha)(\#\Cal+1) \right\}.
\end{equation*}
Denote $W(\mathcal{C}) = \{ i \in \Cal \cup \{ n+1 \} : \sum\limits_{k \in \Cal \cup \{ n+1 \}} \mathcal{C}_{i,k} \geq (1-\alpha)(\#\Cal+1) \}$. We can re-write:
\begin{equation*}
\left\{ Y^{(n+1)} \notin \widehat C_{n,\alpha}^{\text{MDA-Nested}^{\star}}\left(X^{(n+1)},M^{(n+1)}\right) \right\} = \left\{ n+1 \in W(\mathcal{C}) \right\}.
\end{equation*}
Therefore $\mathds{P}\left\{ Y^{(n+1)} \notin \widehat C_{n,\alpha}^{\text{MDA-Nested}^{\star}}\left(X^{(n+1)},M^{(n+1)}\right) \right\} = \mathds{P} \left\{ n+1 \in W(\mathcal{C}) \right\}$. Thus, we will now bound $\mathds{P} \left\{ n+1 \in W(\mathcal{C}) \right\}$.

Again, $\#W(\mathcal{C}) \leq 2\alpha(\#\Cal+1)$ deterministically \citep{barber2021jackknife}.

\paragraph{Step 3.} To conclude the proof, observe that the matrix $\mathcal{C}$ can be viewed as the output of a deterministic function $\mathdutchcal{C}$ of the exchangeable (by \ref{ass:iid}) sequence $\left( X^{(k)},M^{(k)},Y^{(k)} \right)_{k=1}^{n+1}$: $\mathcal{C} = \mathdutchcal{C}\left( \left( X^{(k)},M^{(k)},Y^{(k)} \right)_{k=1}^{n+1} \right) $. 

Thus, for any permutation $\sigma$ on $\Cal \cup \{ n+1 \}$, it holds: 
\begin{equation*}
\mathdutchcal{C}\left( \left( X^{(k)},M^{(k)},Y^{(k)} \right)_{k \in \Cal \cup \{ n+1 \}} \right) \overset{d}{=} \mathdutchcal{C}\left( \left( X^{(\sigma(k))},M^{(\sigma(k))},Y^{(\sigma(k))} \right)_{k \in \Cal \cup \{ n+1 \}} \right) := \mathcal{C}^{\sigma}.
\end{equation*}

It follows that for any $k \in \Cal \cup \{ n+1 \}$, $\mathds{P}\{k \in W(\mathcal{C}) \} = \mathds{P}\{k \in W(\mathcal{C}^{\sigma}) \}$ for any permutation $\sigma$ on $\Cal \cup \{ n+1 \}$. Therefore $\mathds{P}\{k \in W(\mathcal{C}) \}$ does not depend on $k$. Finally:
\begin{align*}
\mathds{P}\left\{ Y^{(n+1)} \notin \widehat C_{n,\alpha}^{\text{MDA-Nested}^{\star}}\left(X^{(n+1)},M^{(n+1)}\right) \right\} & = \mathds{P}\{n+1 \in W(\mathcal{C}) \}\\
& = \frac{1}{\#\Cal+1}\sum\limits_{k\Cal \cup \{ n+1 \}}\mathds{P}\{k \in W(\mathcal{C}) \}\\
& = \frac{1}{\#\Cal+1} \mathds{E}[ \# W(\mathcal{C}) ]\\
& \leq \frac{1}{\#\Cal+1} 2\alpha(\#\Cal+1)=2\alpha.
\end{align*}
\end{proof}
\vspace{-0.5cm}
\subsection{MCV of \newnested}
\label{app:nested_mcv}

To prove that \newnested and \mask are MCV-$\p\mcarymx\iid$, we leverage again the parallel with leave-one-out conformal predictors, but this time seeing the missing pattern as exogenous randomness, which is possible when working with distributions in $\p\mcarymx$. 

\begin{proof}

Under $\p\mcarymx\iid$, it holds that $M^{(n+1)} \ind \left( \left( X^{(k)}, Y^{(k)} \right)_{k \in \rm{Cal}}, \left( X^{(n+1)}, Y^{(n+1)} \right) \right)$. Thus the sequence $\left\{ \left( X^{(k)}, M^{(n+1)}, Y^{(k)} \right)_{k \in \Cal}, \left( X^{(n+1)}, M^{(n+1)}, Y^{(n+1)} \right) \right\}$ is exchangeable conditionally to $M^{(n+1)}$. 

Remark now that for any $(X,M,Y)\in\mathcal{X}\times\mathcal{M}\times\mathcal{Y}$, we can rewrite the score on this point with augmented mask $\widetilde{M} := \max\left(M,M^{(n+1)}\right)$ as:
\begin{equation*}
s\left(\left(X,\widetilde{M}\right),Y;\hat{A}\left(\Phi\left(\cdot,\cdot\right),\cdot\right)\right) := s\left(\left(X,M^{(n+1)}\right),Y;{\widetilde{A}}\left(\widetilde{\Phi}\left(\cdot,\cdot;M\right),\cdot;M\right)\right),
\end{equation*}
where, for an additional mask $M'\in\mathcal{M}$, $\widetilde{\Phi}\left(X,M;M'\right) := \Phi\left(X,\max\left(M,M'\right)\right)$ and similarly ${\widetilde{A}}\left(X,M;M'\right) := \hat{A}\left(X,\max\left(M,M'\right)\right)$.

Thus, we can re-write \newnested as:
\begin{align*}
& \left\{ Y^{(n+1)} \notin \widehat C_{n,\alpha}^{\text{MDA-Nested}^{\star}}\left(X^{(n+1)},M^{(n+1)}\right) \right\} \\
= & \Biggl\{ \begin{aligned}[t] & (1-\alpha)(\#\Cal+1) \\ & \leq \sum\limits_{k \in \Cal} \begin{aligned}[t] \mathds{1}\Bigl\{ & s\left( \left(X^{(n+1)}, \widetilde{M}^{(k)}\right), Y^{(n+1)};\hat{A}\left(\Phi\left(\cdot,\cdot\right),\cdot\right)\right) \\ & \left. > s\left( \left(X^{(k)}, \widetilde{M}^{(k)}\right), Y^{(k)};\hat{A}\left(\Phi\left(\cdot,\cdot\right),\cdot\right)\right) \right\} \Biggr\} \end{aligned} \end{aligned} \\
= & \Biggl\{ \begin{aligned}[t] & (1-\alpha)(\#\Cal+1) \\ & \leq \sum\limits_{k \in \Cal} \begin{aligned}[t] \mathds{1}\Bigl\{ & s\left( \left(X^{(n+1)}, M^{(n+1)}\right), Y^{(n+1)};\widetilde{A}\left(\widetilde{\Phi}\left(\cdot,\cdot;M^{(k)}\right),\cdot;M^{(k)}\right)\right) \\ & \left. > s\left( \left(X^{(k)}, M^{(n+1)}\right), Y^{(k)};\widetilde{A}\left(\widetilde{\Phi}\left(\cdot,\cdot;M^{(k)}\right),\cdot;M^{(k)}\right)\right) \right\} \Biggr\}. \end{aligned} \end{aligned}
\end{align*}

Therefore, an equivalent rewriting of \newnested is a specific instance of what is presented in \Cref{alg:mda_nested_cropped}, where the differences with \newnested (\Cref{alg:mda_nested}) are highlighted through \textcolor{blindgreen}{green text}.
\begin{algorithm}
\caption{MDA based on random masks}
\label{alg:mda_nested_cropped}
\begin{algorithmic}[1] 
\REQUIRE Imputation function $\Phi$, fitted predictor $\hat{A}$, conformity score function $s\left(\cdot,\cdot;f\right)$ for $f \in \mathcal{F}  := \mathcal{Y}^{\mathcal{X}\times\mathcal{M}}$, level $\alpha$, calibration set $\left\{ \left(X^{(k)},M^{(k)},Y^{(k)}\right)\right\}_{k\in\widetilde{\Cal}}$, test point $\left( X^{({n+1})},M^{(n+1)} \right)$
\ENSURE Prediction set ${\widehat{C}}_{n,\alpha}^{\text{MDA-RandomMask}}\left( X^{(n+1)},M^{(n+1)} \right)$
\STATE \textcolor{blindgreen}{Define $\mathcal{G}\left( \nu \right) := \widetilde{A}\left(\widetilde{\Phi}\left(\cdot,\cdot;\nu\right);\nu \right)$ for some $\nu\in\mathcal{M}$}
\FOR {$k \in \widetilde{\Cal}$}
\COMMENT {Additional nested masking}
\STATE \textcolor{blindgreen}{Randomly draw $\nu_k$, independently from $\left(X^{(k)},Y^{(k)},X^{(n+1)},Y^{(n+1)}\right)$}
\STATE \textcolor{blindgreen}{Fit $\hat{g}_{k} := \mathcal{G}\left( \nu_k \right) = \widetilde{A}\left(\widetilde{\Phi}\left(\cdot,\cdot;\nu_k\right);\nu_k \right)$}
\ENDFOR
\STATE \vspace{-0.5cm}
\textcolor{blindgreen}{
\begin{align*}
& {\widehat{C}}_{n,\alpha}^{\text{MDA-RandomMask}}\left(  X^{(n+1)},M^{(n+1)} \right) \\
:= & \Biggl\{ y\in\mathcal{Y}:\begin{aligned}[t]& (1-\alpha)(1+\#\Cal) > \\
& \left. \sum_{k\in\widetilde{\Cal}} \mathds{1} \left\{ s\left( \left(X^{(k)}, {M}^{(k)}\right), Y^{(k)}; \hat{g}_k \right) < s\left( {\left(X^{(n+1)}, M^{(n+1)}\right)}, y; \hat{g}_k )\right) \right\} \right\}\end{aligned} 
\end{align*}
}
\end{algorithmic}
\end{algorithm}

Indeed, conditionally on $M^{(n+1)}$, we can apply \Cref{alg:mda_nested_cropped} to the modified data set \\ $\left(X^{(k)},M^{(n+1)},Y^{(k)}\right)_{k\in\widetilde{\Cal}}$, by using the $\left(M^{(k)}\right)_{k\in\widetilde{\Cal}}$ as random draw for $\left(\nu_k\right)_{k\in\widetilde{\Cal}}$ in line 3. This is legit only when the distribution of $\left(X^{(k)},M^{(n+1)},Y^{(k)}\right)_{k\in\widetilde{\Cal}\cup\{n+1\}}$ belongs to $\p\mcarymx^{\otimes(\#\widetilde{\Cal}+1)}$, as then for any $k\in\widetilde{\Cal}$, it holds that $M^{(k)}\ind\left(X^{(k)},Y^{(k)},X^{(n+1)},Y^{(n+1)}\right)$. 

This \Cref{alg:mda_nested_cropped} is a special case of leave-one-out CP presented in \Cref{app:jk}, with a randomized algorithm that only returns a pre-determined function associated with a parameter value, without fitting anything on the $n-1$ data points. Therefore, the validity result of leave-one-out CP extends to \Cref{alg:mda_nested_cropped}. 

In particular, under $\p\mcarymx\iid$, \newnested corresponds to applying \Cref{alg:mda_nested_cropped} to the data set $\left(X^{(k)},M^{(n+1)},Y^{(k)}\right)_{k\in\Cal}$ which is exchangeable conditionally on $M^{(n+1)}$, and by using in line 3 the $\left(M^{(k)}\right)_{k\in\Cal}$ as random draw for $\left(\nu_k\right)_{k\in\Cal}$. Therefore, \newnested is MCV-$\p\mcarymx\iid$ at the level $1-2\alpha$. 
\end{proof}

The idea in this re-writing is to see that, conditionally on $M^{(n+1)}$, \newnested predicting on the test point $\left(X^{(n+1)},M^{(n+1)}\right)$ given the data set $\left(X^{(k)},M^{(k)},Y^{(k)}\right)_{k=1}^{n}$, is in fact another run of \newnested which predicts on a complete test point $\breve{X}^{(n+1)} \in \breve{\mathcal{X}}$, where $\breve{\mathcal{X}}$ is the set of dimension $|\obs\left(M^{(n+1)}\right)|$ containing only the observed dimensions of $\mathcal{X}$ according to $M^{(n+1)}$, given the cropped data set $\left(\breve{X}^{(k)},\breve{M}^{(k)},Y^{(k)}\right)_{k=1}^{n}$, with $\breve{M}^{(k)}\in\breve{\mathcal{M}}$ that, similarly to $\breve{X}$, is the set of dimension $|\obs\left(M^{(n+1)}\right)|$ containing only the observed dimensions of $\mathcal{M}$ according to $M^{(n+1)}$.

\bibliography{ref}
\bibliographystyle{apalike}

\end{document}